\documentclass[sigconf,authorversion]{acmart}

\RequirePackage[loading]{tracefnt}

\usepackage[T1]{fontenc} 
\usepackage[utf8]{inputenc}
\usepackage{balance}
\usepackage{amsthm}
\usepackage{amsmath}
\usepackage{mathtools}
\usepackage{xfrac}
\usepackage{subfig}
\usepackage{multirow}
\usepackage{xspace}
\usepackage{graphicx}
\usepackage[ruled, vlined,linesnumbered]{algorithm2e}
\usepackage[all]{nowidow}

\usepackage{bbding}
\newcommand{\vmark}{\Checkmark\xspace}
\newcommand{\xmark}{\XSolidBrush\xspace}

\usepackage{tikz}
\usepackage{xcolor}
\usepackage{comment}
\usepackage{tikz-cd}
\usepackage{cleveref}
\usepackage[normalem]{ulem}
\usetikzlibrary{fit}
\usetikzlibrary{arrows,decorations.pathmorphing,backgrounds,positioning,fit,petri, calc}
\definecolor{light blue}{HTML}{ADD8E6}
\definecolor{darker blue}{HTML}{9BC2CF}
\definecolor{ciano}{HTML}{00FFFF}
\definecolor{arancio}{HTML}{ffa31a}
\definecolor{dark green}{HTML}{009900}
\definecolor{bblue}{RGB}{0,50,200}

\definecolor{colorB1}{HTML}{1E88E5}
\definecolor{colorB2}{HTML}{D81B60}

\definecolor{cIBM1}{HTML}{648fff}
\definecolor{cIBM2}{HTML}{785ef0}
\definecolor{cIBM3}{HTML}{dc267f}

\newcommand{\para}[1]{\noindent\textbf{#1}}

\newtheorem{problem}{Problem}

\newcommand{\edit}[1]{\textcolor{black}{{#1}}\xspace}

\newcommand{\bigO}{\ensuremath{\mathcal{O}}\xspace}

\newcommand{\triSet}{\ensuremath{\Delta}\xspace}
\newcommand{\nodeDeg}[1]{\ensuremath{d_{#1}}\xspace}

\newcommand{\triSetNode}[1]{\ensuremath{\triSet_{#1}}\xspace}
\newcommand{\triSetEdge}[1]{\ensuremath{\triSet_{#1}}\xspace}

\newcommand{\wedg}{\ensuremath{\mathsf{w}}\xspace}
\newcommand{\wedgSet}{\ensuremath{\mathsf{W}}\xspace}

\newcommand{\wedgNodeSetCenter}[1]{\ensuremath{\wedgSet^{\mathsf{c}}_{#1}}\xspace}
\newcommand{\wedgNodeSetHead}[1]{\ensuremath{\wedgSet^{\mathsf{h}}_{#1}}\xspace}
\newcommand{\wedgNodeSetGeneral}[1]{\ensuremath{\wedgSet^{\ast}_{#1}}\xspace}
\newcommand{\arb}{\ensuremath{\xi}\xspace}

\newcommand{\singleTri}{\ensuremath{\delta}\xspace}
\newcommand{\clustCoeff}{\ensuremath{\alpha}\xspace}
\newcommand{\closureCoeff}{\ensuremath{\phi}\xspace}
\newcommand{\clustMetric}{\ensuremath{\psi}\xspace}
\newcommand{\clustGenPartition}{\ensuremath{\Psi}\xspace}
\newcommand{\estFunc}{\ensuremath{f}\xspace}
\newcommand{\varEst}{\ensuremath{\widehat{V}}\xspace}

\newcommand{\neighborNode}[1]{\ensuremath{\mathcal{N}_{#1}}\xspace}
\newcommand{\neighborEdge}[1]{\ensuremath{\mathcal{N}_{#1}}\xspace}

\newcommand{\G}{\ensuremath{G}\xspace}
\newcommand{\V}{\ensuremath{V}\xspace}
\newcommand{\E}{\ensuremath{E}\xspace}

\newcommand{\DomX}{\ensuremath{\mathcal{X}}\xspace}
\newcommand{\DomH}{\ensuremath{\mathcal{H}}\xspace}
\newcommand{\RangeSet}{\ensuremath{\mathcal{Q}}\xspace}
\newcommand{\Fset}{\ensuremath{\mathcal{F}}\xspace}
\newcommand{\sampleSet}{\ensuremath{\mathcal{S}}\xspace}
\newcommand{\PD}{\ensuremath{\mathsf{PD}}\xspace}
\newcommand{\VC}{\ensuremath{\mathsf{VC}}\xspace}
\newcommand{\pdim}{\ensuremath{\zeta}\xspace}
\newcommand{\bucketColor}{\ensuremath{\chi}\xspace}
\newcommand{\bucketColorSup}{\ensuremath{\widehat{\chi}}\xspace}
\newcommand{\nodesetPartition}{\ensuremath{\mathcal{V}}\xspace}

\newif\ifextended
\extendedtrue% switch to true in extended version

\DeclareMathOperator{\expectation}{\mathbb{E}}
\DeclareMathOperator{\Prob}{\mathbb{P}}
\DeclareMathOperator{\Var}{Var}

\DeclarePairedDelimiter\ceil{\lceil}{\rceil}
\DeclarePairedDelimiter\floor{\lfloor}{\rfloor}

\newcommand{\mainAlg}{T{\large\textsc{riad}}\xspace}
\newcommand{\algFixedSS}{T{\large\textsc{riad-f}}\xspace}

\newcommand{\baseFull}{\texttt{WedgeSampler}\xspace}
\newcommand{\CODEURL}{\url{https://github.com/iliesarpe/Triad}\xspace}

\newcommand\vldbdoi{10.14778/3742728.3742748}
\newcommand\vldbpages{2561 - 2574}
\newcommand\vldbvolume{18}
\newcommand\vldbissue{8}
\newcommand\vldbyear{2025}
\newcommand\vldbauthors{\authors}
\newcommand\vldbtitle{\shorttitle} 
\newcommand\vldbavailabilityurl{https://github.com/iliesarpe/Triad}
\newcommand\vldbpagestyle{empty} 

\begin{document}
\title{Efficient and Adaptive Estimation of Local Triadic Coefficients}

%%
%% The "author" command and its associated commands are used to define the authors and their affiliations.
\author{Ilie Sarpe}
%\orcid{0009-0007-5894-0774}
\affiliation{%
  \institution{KTH Royal Institute of Technology}
  %\streetaddress{P.O. Box 1212}
  \city{Stockholm}
  \state{Sweden}
  %\postcode{43017-6221}
}
\email{ilsarpe@kth.se}

\author{Aristides Gionis}
%\orcid{0000-0002-5211-112X}
\affiliation{%
  \institution{KTH Royal Institute of Technology}
  %\streetaddress{1 Th{\o}rv{\"a}ld Circle}
  \city{Stockholm}
  \country{Sweden}
}
\email{argioni@kth.se}

%%
%% The abstract is a short summary of the work to be presented in the
%% article.
\begin{abstract}
Characterizing graph properties is fundamental to the analysis and to our understanding of real-world networked systems. 
The \emph{local clustering coefficient}, and the more recently introduced, \emph{local closure coefficient},
capture powerful properties that 
are essential in a large number of applications, 
ranging from graph embeddings to graph partitioning. 
Such coefficients 
capture
the local density of the neighborhood of each node,
considering incident triadic structures and paths of length two.
For this reason, we refer to these coefficients collectively as \emph{local triadic coefficients}.

In this work, we consider the novel problem of computing efficiently the 
\emph{average} of local triadic coefficients, 
over a given \emph{partition} of the nodes of the input graph into a set of disjoint \emph{buckets}.
The \emph{average local triadic coefficients} of the nodes in each bucket
provide a better insight into the interplay 
of graph structure
and the properties of the nodes associated to each bucket.
Unfortunately, exact computation, which requires listing all triangles in a graph, 
is infeasible for large networks.
Hence, we focus on obtaining \emph{highly-accurate probabilistic estimates}.

We develop~\mainAlg, an adaptive algorithm based on sampling, 
which can be used to estimate the average local triadic coefficients 
for a partition of the nodes into buckets. 
\mainAlg is based on a new class of unbiased estimators, and
non-trivial bounds on its sample complexity,
enabling the efficient computation of highly accurate estimates.
Finally, we show how~\mainAlg can be efficiently used in practice on large networks, and 
we present a case study 
showing that average local triadic coefficients can capture high-order patterns over collaboration networks.

\end{abstract}

\maketitle

%%% do not modify the following VLDB block %%
%%% VLDB block start %%%
\pagestyle{\vldbpagestyle}
\begingroup\small\noindent\raggedright\textbf{PVLDB Reference Format:}\\
\vldbauthors. \vldbtitle. PVLDB, \vldbvolume(\vldbissue): \vldbpages, \vldbyear.\\
\href{https://doi.org/\vldbdoi}{doi:\vldbdoi}
\endgroup
\begingroup
\renewcommand\thefootnote{}\footnote{\noindent
This work is licensed under the Creative Commons BY-NC-ND 4.0 International License. Visit \url{https://creativecommons.org/licenses/by-nc-nd/4.0/} to view a copy of this license. For any use beyond those covered by this license, obtain permission by emailing \href{mailto:info@vldb.org}{info@vldb.org}. Copyright is held by the owner/author(s). Publication rights licensed to the VLDB Endowment. \\
\raggedright Proceedings of the VLDB Endowment, Vol. \vldbvolume, No. \vldbissue\ %
ISSN 2150-8097. \\
\href{https://doi.org/\vldbdoi}{doi:\vldbdoi} \\
}\addtocounter{footnote}{-1}\endgroup
%%% VLDB block end %%%

%%% do not modify the following VLDB block %%
%%% VLDB block start %%%
\ifdefempty{\vldbavailabilityurl}{}{
\vspace{.3cm}
\begingroup\small\noindent\raggedright\textbf{PVLDB Artifact Availability:}\\
The source code, data, and/or other artifacts have been made available at \url{\vldbavailabilityurl}.
\endgroup
}
%%% VLDB block end %%%

\section{Introduction}
Graphs are a ubiquitous data abstraction
used to study complex systems in different domains, 
such as social networks~\citep{Fang2019SurveySN}, 
protein interactions~\citep{Bhowmick2016ClusteringProteins}, 
information networks~\citep{sun2013mining}, and more~\cite{Newman2018Networks}. 
A graph provides a simple representation:  
entities are represented by nodes and their relations are represented by edges, 
enabling the analysis of structural properties and giving unique insights 
into the function of various systems~\citep{Xu2001Topology}. 
For example, flow analysis in transport graphs can be used for better urban design~\cite{Jiang2008Traffic},
subgraph patterns can improve recommenders~\citep{Leskovec2006PatternsPurchase};  
and dense subgraph identification captures highly collaborative communities~\cite{Li2015Community}.

The \emph{local clustering coefficient}~\citep{Watts1998} 
is among the most important structural properties for graph analysis, 
%is of interest in graph analyses, 
%capturing how nodes tend to \emph{cluster} within their neighborhoods. 
%The local clustering coefficient 
and is used in many applications related to 
databases~\citep{Ciglan2012DB}, 
social networks~\citep{Hardiman2013EstimatingCCRW}, 
graph embeddings~\citep{Bianchi2019SpectralGNN}, and 
link prediction~\citep{Wu2016LinkPred}.
The local clustering coefficient measures the fraction of connected pairs of neighbors 
of a given node $u$, 
e.g., see Figure~\ref{fig:inner:localCCvsCL} (left), 
providing a simple and interpretable value on how well the node is ``embedded'' within its local neighborhood. 
On an academic collaboration network, for example, the local clustering coefficient of an author $u$
corresponds to the fraction of the co\-authors of $u$ collaborating with each other, 
capturing a salient co\-authorship pattern of author~$u$.

\begin{figure}
	\centering
	\subfloat[]{
		\includegraphics[width=0.45\columnwidth]{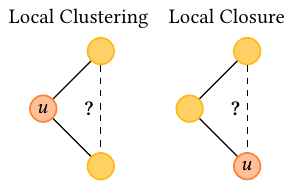}
		\label{fig:inner:localCCvsCL}
	}%
	\subfloat[]{
		\includegraphics[width=0.5\columnwidth]{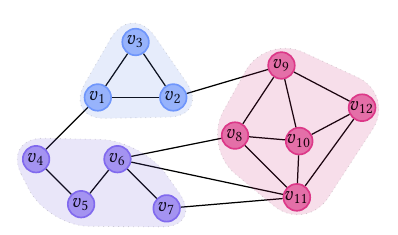} 
		\label{fig:inner:graphPart}
	}
	\caption{\ref{fig:inner:localCCvsCL}: the local clustering coefficient of a node $u$ considers the fraction of connected pairs of neighbors of $u$, while the local closure coefficient of $u$ considers the fraction of neighbors of $u$'s neighbors to which $u$ is connected. \ref{fig:inner:graphPart}: a graph with its set of nodes partitioned into three different~sets. }
	\label{fig:fullgraphAndClustvsCoslure}
\end{figure}

Recently, \citet{Yin2019Closure} introduced the \emph{local closure coefficient}, 
a new coefficient capturing the fraction of paths of length two, originating from a node $u$ that is closed by $u$, 
e.g., see Figure~\ref{fig:inner:localCCvsCL} (right). 
This novel definition directly accounts for the connections generated by $u$ in the graph, 
differently from the local clustering coefficient that only depends on connections in $u$'s neighborhood.  
The local closure coefficient is a simple concept that is gaining interest in the research community, 
as it provides additional and complementary 
insights to existing coefficients
and has applications in anomaly detection~\cite{Zhang2023Anom} and 
link prediction~\cite{Yin2019Closure}.
In this paper, we refer collectively to the local clustering coefficient and the local closure coefficient
as \emph{local triadic coefficients}.
Both local triadic coefficients are fundamental quantities 
for graph analysis,
as they capture structural properties of graphs on a local level~\cite{Yin2019Closure,Yin2018HCC,Yin2017LocalHG}.

In several applications, 
we are interested in the \emph{average} local triadic coefficient
of a subset of nodes~\cite{Li2017ClustLarge,Kaiser2008MeanCC,Ugander2011FB}. 
The most typical example is to consider the average over \emph{all} nodes in the graph, 
e.g., the \emph{average local clustering coefficient}:
a standard graph statistic available in popular graph libraries \cite{Leskovec2016SNAP,Rossi2015NetworkRep}.
In general, 
computing the {average} local triadic coefficients
of a \emph{subset} of nodes based on graph properties can yield unique insights. 
For example, such properties may be associated with node metadata (e.g., in typed networks~\cite{Shi2017Heter}), node similarities (e.g., from structural properties such core number and degrees, or role similarity~\cite{Jin2011Similarity}), or graph communities~\cite{Karypis1998metis}. 
As a concrete example, in an academic collaboration network, 
we can compute the {average} local triadic coefficient
of all authors who publish consistently in certain venues. For example, database conferences or machine learning---and average local coefficients could reveal interesting patterns for each community of interest, such as different trends in collaboration patterns.

The average local clustering coefficient
is also exploited as a measure for
community detection or clustering algorithms~\cite{Pan2019CommCC}, with good quality clusters achieving high average local clustering coefficients, e.g., compared to random partitions. 
Hence, analyzing the average local triadic coefficient for different buckets
can provide us with powerful insights for various applications, ranging from analyzing structural properties of specific groups of users with similar metadata, to community detection. Furthermore, average local triadic coefficients can also be employed to empower machine learning models, e.g., GNNs, for tasks such as graph classification or node embeddings~\cite{Bianchi2019SpectralGNN}.

Motivated by the previous settings requiring to compute the average
local triadic coefficients over (given) sets of nodes, 
in this paper, we study the following problem: 
given a partition of the nodes of a graph into $k$ sets, 
efficiently compute the average local triadic coefficient (clustering or closure)  
\emph{for each} set of the partition.

Unfortunately, to address this problem, we cannot rely on exact algorithms, 
since exact computation of the local triadic coefficients for all graph nodes
is an extremely challenging task
and requires exhaustive enumeration of {all} triangles in a graph. 
Despite extensive study of exact algorithms
for \emph{triangle counting}~\cite{Li2024FastLocalCnt,Bader2023FastCounting, Bader2024CoverTri}, 
enumeration requires time $\Theta(m^{3/2})$,
i.e., $\Theta(n^3)$ on dense graphs, 
which is extremely inefficient and 
resource-demanding for massive graphs.\footnote{We use $n$ and $m$ for the number of nodes and edges respectively.} 

To overcome this challenge, we develop an efficient \emph{adaptive approximation} algorithm: 
\mainAlg (average local {\large\textsc{tri}}adic {\large\textsc{ad}}aptive estimation), 
which can break the complexity barrier at the expense of a small approximation error.
Similar to other approximation algorithms for graph analysis,
\mainAlg relies on \emph{random sampling}~\cite{Bressan2019Motivo,Zhao2021Sketch,Wang2018Moss5}. 
Where, triangles incident to randomly sampled edges are used to update an estimate of the average triadic coefficient 
\emph{for each} set of the partition of the nodes of a graph, through a novel \emph{class} of unbiased estimators.
\mainAlg\ can approximate both the average local clustering and closure coefficients 
of arbitrary partitions of the graph nodes. 
Surprisingly, to the best of our knowledge, 
\mainAlg is also the first algorithm specifically designed to estimate 
average local closure~coefficients. 

Our design of~\mainAlg is guided by two key properties,  essential for many sampling schemes: 
($i$)~provide accurate estimates that are close to the unknown values being estimated; 
($ii$)~provide high-quality adaptive probabilistic guarantees on the distance between the estimates and the values being estimated. 
In fact, differently from existing approaches 
\mainAlg quantifies the deviation between the probabilistic estimates reported in output and the underlying unknown average coefficients through a data-dependent approach~\cite{Seshadhri2014,Zhang2017LocalCC,Lattanzi2016WeightedCC,Kutzkov2013StreamCC}. 
This approach results in an extremely efficient algorithm, 
whose performance is judiciously adapted to the input graph, 
since the estimation is based on empirical quantities computed over the collected samples.
Our contributions are as~follows.
\begin{itemize}
	\item We study the problem of efficiently obtaining high-quality estimates of the average local closure and average local clustering coefficients for each set in a partition of the nodes of a graph.
	\item We develop~\mainAlg, an efficient and adaptive algorithm that provides high-quality estimates with controlled error probability. \mainAlg is based on a novel class of estimators 
	that we optimize to achieve provably small variance and 
	a novel bound on the sample size obtained through the notion of pseudo-dimension. \mainAlg~also leverages state-of-the-art variance-aware concentration results to quantify the deviation of its estimates to the unknown estimated values, by the means of adaptive data-dependent bounds.
	\item We extensively assess the performance of \mainAlg on large graphs. Our results show that  \mainAlg provides high-quality probabilistic estimates and strong theoretical guarantees not matched by existing state-of-the-art algorithms. We also show how the estimates of~\mainAlg can be used to study publication patterns among research fields over time on a DBLP graph.
\end{itemize}

\section{Preliminaries}
\label{sec:prelims}

Let $\G=(\V,\E)$ be a simple and undirected graph with node-set $\V = \{v_1,\dots,v_n \}$ 
and edge-set $\E=\{\{u,v\}: u\in \V, v\in \V \text{ and } u\neq v\}$, where $|\V|=n$ and $|\E|=m$.

For a node $v \in V$ we denote its \emph{neighborhood} with 
$\neighborNode{v}=\{u \in \V : \text{it exists } \{u,v\}\in \E\}$ 
and its \emph{degree} with $\nodeDeg{v}= |\neighborNode{v}|$. 
Similarly, given an edge $e=\{u,v\}\in \E$ we define 
$\neighborEdge{e}=\{w \in \V : w \in (\neighborNode{u}\cap \neighborNode{v})\}$,
i.e., the neighborhood of an edge $e\in \E$ is the set of nodes from $\V$ 
that are neighbors to \emph{both} incident nodes of $e$.

A \emph{wedge} is a set of two distinct edges sharing a common node, i.e., $\wedg=\{e_1,e_2\}\subseteq E \text{ such that } |e_1\cap e_2|=1$;
a wedge also corresponds to a path of length two. 
We let $\wedgSet=\{\wedg : \wedg \text{ is a wedge in } \G\}$ 
be the set of all wedges in the graph $\G$. 
Given a node $v\in \V$ we say that a wedge $\wedg =\{e_1,e_2\}$ is \emph{centered} at $v$ if $\{v\}= e_1\cap e_2$.
The set of all such wedges is denoted with $\wedgNodeSetCenter{v}$.
Note that, for each node $v\in\V$, it holds $|\wedgNodeSetCenter{v}|={\nodeDeg{v}\choose 2}$. 
A wedge $\wedg =\{e_1,e_2\}$ is \emph{headed} at a node $v\in V$ 
if $v\in\wedg$ and $v\notin e_1\cap e_2$.\footnote{We write $v\in \wedg=\{e_1,e_2\}$ to denote that 
$v\in e_1$ or $v\in e_2$.} 
The set of wedges headed at $v\in\V$ is denoted with $\wedgNodeSetHead{v}$.
Note that $|\wedgNodeSetHead{v}| = \sum_{u\in \neighborNode{v}} (d_u -1)$.
\begin{example}
	In Figure \ref{fig:inner:graphPart},  
	$\{\{v_6,v_8\},\allowbreak\{v_8,v_{11}\}\}$ is a wedge centered at $v_8$ and $|\wedgNodeSetCenter{v_8}| = 6$, 
	since $\nodeDeg{v_8} = 4$. 
	Additionally, $\{\{v_2,v_9\},\{v_9,v_8\}\}$ is a wedge headed at $v_8$ and~$|\wedgNodeSetHead{v_8}|=13$.
\end{example}

Next, given a graph $\G=(\V,\E)$ we define a \emph{triangle} as a set of three edges 
that pairwise share an edge, i.e., $\singleTri = \{\{u,v\}, \{v,w\},$ $ \{w,u\}\}\subseteq \E$ 
and $u,v,w\in V$ are three distinct nodes. 
The set of all triangles in $\G$ is denoted by 
$\triSet=\{\singleTri : \singleTri \text{ is a triangle in } \G\}$, 
while $\triSetNode{v} = \{\singleTri \in \triSet : v \in \delta\}\subseteq\triSet$,\footnote{We write $v\in \singleTri$ to denote that there exists an edge $e\in \singleTri$ such that $v\in e$.} corresponds to the set of triangles containing $v\in V$. 
Similarly $\triSetEdge{e}= \{\singleTri \in \triSet: e\in \singleTri\}$ corresponds to the set of all triangles containing an edge $e\in\E$.

We are now ready to introduce the fundamental  triadic coefficients studied in this paper
and introduced in earlier work~\cite{Watts1998,Yin2019Closure}.

\begin{definition}\label{def:coefficientsDef}%[Local Clustering Coefficient~\cite{Watts1998}]
	Given a graph $\G=(\V,\E)$ and a node $v\in V $ we define the 
	\emph{local clustering coefficient} of $v$,  denoted by $\clustCoeff_v$,
	 and the \emph{local closure coefficient}  of $v$, denoted by $\closureCoeff_v$, respectively as
	\[
	\clustCoeff_v =  \frac{|\triSetNode{v}|}{|\wedgNodeSetCenter{v}|} = \frac{|\triSetNode{v}|}{{d_v \choose 2 }} 
	\quad 
	\text{ and } 
	\quad
	\closureCoeff_v = \frac{2|\triSetNode{v}|}{|\wedgNodeSetHead{v}|} =  \frac{2|\triSetNode{v}|}{{\sum_{u\in \neighborNode{v}} (\nodeDeg{v} -1 ) }} \enspace.
	\]
\end{definition}

\begin{figure}[t]
	\centering
	\subfloat[]{
		\includegraphics[width=0.26\columnwidth]
		{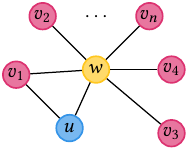}
		\label{fig:inner:exampleDifferent}
	}%
	\hspace{12mm}
	\subfloat[]{
		\includegraphics[width=0.4\columnwidth]{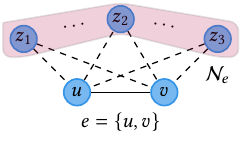}
		\label{fig:inner:exampleWeights}
	}%
	\caption{
	2A: Discussion in Section \ref{sec:prelims}.
	For node $u$ it holds $\clustCoeff_u =1$ and $\closureCoeff_u=\bigO(1/n)$. 
	For node $w$ it holds $\clustCoeff_w = \bigO(1/n^2)$ and $\closureCoeff_w = 1$.
	Thus, the local clustering and local closure coefficients can differ significantly. 
	2B:	Discussion in Section~\ref{sec:methods}.
	Consider a sampled edge $e\in\E$. 
	For each node $w\in \V$ in the graph, 
	our estimate for $|\triSetNode{w}|$ is 
	$|\triSetNode{w}|=q|\triSetEdge{e}|/p$, if $w\in \{u,v\}$, and 
	$|\triSetNode{w}|=(1-2q)|\triSetEdge{e}|/p$, if $w = z_i\in \neighborEdge{e}$.}
	\label{fig:trisAndMetrics}
\end{figure}

As an example, consider $v_9$ in Figure~\ref{fig:inner:graphPart}. 
Then $\clustCoeff_{v_9} = 2/(2\cdot  3) = 1/3$ 
and $\closureCoeff_{v_9} = (2\cdot  2)/10 = 2/5$.

Observe that the values of $\clustCoeff_v$ and $\closureCoeff_v$ for a node $v\in \V$ 
can differ significantly, 
as illustrated in the example of Figure~\ref{fig:inner:exampleDifferent}.

Next, given a subset of nodes $\V'\subseteq \V$ we define the \emph{average} local clustering coefficient (respectively, average local closure coefficient) 
as the average of the local clustering (respectively, local closure) coefficient of the nodes in the subset $\V'$, 
that is $\clustCoeff (\V')  = \tfrac{1}{|\V'|}\sum_{v\in \V'} \clustCoeff_v$ 
(respectively, $\closureCoeff (\V')  = \tfrac{1}{|\V'|}\sum_{v\in \V'} \closureCoeff_v$).

Given a set $A\neq \emptyset$, then $A_1,\dots, A_k$ is a partition of $A$ 
if $A_i\cap A_j =\emptyset $ for $i\neq j$ with $i,j\in [k]$, 
$A_i\neq \emptyset$ for all $i\ge 1$, and $\bigcup_{i\ge 1} A_i = A$. 
For ease of notation we denote a partition of $V$ into $k$ sets as $\nodesetPartition$.
In Figure~\ref{fig:inner:graphPart} we illustrate a network and a partition of the set of nodes.
For ease of notation we use $\clustMetric_v$ to refer to a local triadic coefficient of a node $v\in V$ (that is, either clustering coefficient or closure coefficient) and 
$\clustGenPartition(V') = \tfrac{1}{|V'|}\sum_{v\in V'} \clustMetric_v$. 

Following standard ideas in the literature for controlling the quality of an approximate estimate~\cite{Riondato2018Abra,Riondato2018MisoSoup,Kolda2014}, 
we consider two key properties: 
($i$)~the estimate should be close to the actual value; and 
($ii$)~there should exist rigorous guarantees quantifying the distance of the estimate to the unknown coefficient. 
These desirable properties are captured by the following problem formulation.

\begin{problem}\label{prob:partionEstimationAdditive}
	Given a graph $\G=(\V,\E)$, 
	a partition $\nodesetPartition$ of the node-set $\V$, 
	a local triadic coefficient $\clustMetric \in \{\clustCoeff, \closureCoeff\}$, and 
	parameters $(\varepsilon_j)^{k}_{j=1}\in (0,1)$ and $\eta \in (0,1)$, 
	obtain:
	\begin{enumerate}
		\item[(A)] Estimates $\estFunc(V_j)$, for $j=1,\dots, k$, such that
		\[
		\Prob\left[\sup_{j=1,\dots,k} \left| \estFunc(\V_j) - \frac{1}{|V_j|} \sum_{v\in V_j} \clustMetric_v \right| \ge \varepsilon_j \right] \le \eta \enspace.
		\]
		\item[(B)] \emph{Tight} confidence intervals $C_j$ as possible,  
		where $C_j = [f(V_j)-\widehat{\varepsilon}_j, f(V_j)+\widehat{\varepsilon}_j]$, 
		for $j=1,\dots, k$,	with 
		$|C_j| = 2\widehat{\varepsilon}_j \le 2\varepsilon_j$, 
		such that over all $k$ partitions it holds
		\[
		\Prob\left[\frac{1}{|V_j|} \sum_{v\in V_j} \clustMetric_v  \in C_j \right] \ge 1-\eta \enspace.
		\]
	\end{enumerate}
\end{problem}

To simplify our notation and when it is clear from the context we write 
$\estFunc_j$ instead of $\estFunc(V_j)$ and $\clustGenPartition_j$ instead of $\clustGenPartition(V_j)$, 
for $j\in [k]$. 

In the remaining of this section we discuss the formulation of  Problem~\ref{prob:partionEstimationAdditive}. In particular, Problem~\ref{prob:partionEstimationAdditive} takes as input a graph $G$ 
and a node partition $\nodesetPartition$.
The goal is to obtain \emph{accurate} estimates $\estFunc_j$, 
within at most $\varepsilon_j$ \emph{additive error} (i.e., $|\estFunc_j -\clustGenPartition_j|\le \varepsilon_j$) 
for the local triadic coefficients $\clustGenPartition_j$ 
with controlled error probability ($\eta$) \emph{over all sets} $j\in[k]$ of the partition~$\nodesetPartition$. 
This requirement is enforced with condition (A) in Problem~\ref{prob:partionEstimationAdditive}.

Furthermore, condition (B) ensures that the confidence intervals~$C_j$ 
(i.e., the ranges in which the values $\clustGenPartition_j$ are likely to fall)
centered around $\estFunc_j$ are small. 
This requirement provides a rigorous guarantee on the proximity of $\estFunc_j$'s 
to the corresponding~$\clustGenPartition_j$'s.  

Note that Problem~\ref{prob:partionEstimationAdditive} 
allows the user to provide non-uniform error bounds $\varepsilon_j$, $j=1,\dots,k$. 
This flexibility is highly desirable since
($i$)~there may be partitions for which estimates are required with different levels of precision 
(as controlled by $\varepsilon_j$), 
e.g., in certain applications 
the user may require higher precision on the value of $\clustGenPartition_j$ for some $j$'s, and 
($ii$)~different values of $\varepsilon_j$ may be required to distinguish between 
the values of $\clustGenPartition_j$ for different sets in $\nodesetPartition$~\cite{Borassi2019Kadabra,Pellegrina2023Silvan}. 

To clarify point ($ii$) above,  
consider $V=V_1\cup V_2$ with $\clustGenPartition_1 = 10^{-2}$ and $\clustGenPartition_2 = 10^{-5}$.
If $\varepsilon= \varepsilon_1=\varepsilon_2  = 5\cdot  10^{-2}$ then both $f_1,f_2\in[0,\varepsilon]$ 
satisfy the guarantees of Problem~\ref{prob:partionEstimationAdditive} 
but this does not allow to distinguish between the very large value difference 
of $\clustGenPartition_1$ and $\clustGenPartition_2$ 
(i.e., three orders of magnitude).
On the other hand, by allowing different accuracy levels, 
e.g., $\varepsilon_1 = 10^{-3}$ and $\varepsilon_2 = 10^{-4}$, 
we can address the issue.
The acute reader may notice in this example we assume that we know the values of
$\clustGenPartition_1$ and $\clustGenPartition_2$ in advance,  
and we set the values $\varepsilon_j$'s accordingly, 
while in practice this is rarely the case. 
However, as we show in Section~\ref{sec:methods}, 
our algorithm~\mainAlg can adaptive\-ly address the case of unknown $\clustGenPartition_j$'s.

An alternative approach would be to require \emph{relative error guarantees}, 
i.e., $|\estFunc_j -\clustGenPartition_j|\le \varepsilon_j\clustGenPartition_j$, for all $j\in[k]$.
Unfortunately this problem cannot be solved efficiently.
First, it requires a lower bound on each value $\clustGenPartition_j$, and 
second it requires an impractical number of samples 
even for moderate values of $\clustGenPartition_j$ and $\varepsilon_j$.
In particular, 
state-of-the-art methods~\cite{Seshadhri2014,Lattanzi2016WeightedCC,Zhang2017LocalCC,Lima2022ClusteringLatin}
have shown that $\Omega(1/(\varepsilon\clustGenPartition_j)^2)$ samples may be needed. 
Thus, if, say, $\clustGenPartition = 10^{-3}$ and $\varepsilon=10^{-2}$, 
then $\Omega(10^{10})$ samples would be needed, 
resulting in an extremely high and impractical running time.

\section{Methods}
\label{sec:methods}

Our algorithm~\mainAlg consists of several components. 
At its core, it is based on a new class of unbiased estimators, 
which can be of independent interest for local triangle count estimation. 
We discuss these estimators in Section~\ref{subsec:localEst}. 
We then introduce~\mainAlg, 
in Section~\ref{subsec:MainAlg}. 
We
present~\mainAlg's analysis in Section~\ref{subsec:algAnalysis}, and
some practical optimizations in Section~\ref{subsec:algPractOpt}. 
We analyze the time and memory complexity of \mainAlg in~\Cref{subsec:timeandmem}.
Finally, in~\Cref{subsec:adaptivity} we discuss the adaptive behavior of~\mainAlg.
\ifextended
All missing proofs are reported in Appendix~\ref{app:missingProofs}.
\else
Missing proofs are reported in our extended version~\cite{triadExt}.
\fi

\subsection{New estimates for local counts}\label{subsec:localEst}

We introduce a new \emph{class} of estimators to approximate, for each node $v\in \V$, 
the number of triangle counts $|\triSetNode{v}|$ (i.e., locally to~$v$), 
based on a simple sampling procedure that uniformly selects random edges. 
Such estimators stand at the core of the proposed method~\mainAlg, 
enabling a small variance of the estimates in output.

To compute the estimators for $|\triSetNode{v}|$, for $v\in\V$, 
first we sample uniformly an edge $e\in \E$,
and then collect the set of triangles $\triSetEdge{e}$. 
The total number of triangles $|\triSetEdge{e}|$ is then used to estimate $|\triSetNode{v}|$, 
\emph{for each} node $v\in\V$. 
The key, is in \emph{how} the value $|\triSetEdge{e}|$ 
is distributed over all nodes $v \in e \cup \neighborEdge{e}$, 
to obtain an unbiased estimators of $|\triSetNode{v}|$. 

In particular, our estimators distribute the weight $|\triSetEdge{e}|$ 
\emph{asymmetrically} across the nodes $v\in e$ (where $e\in \E$ is the sampled edge) 
and the nodes $v\in \neighborEdge{e}$. 
This asymmetric 
assignment
can be made before evaluating the estimates of $|\triSetNode{v}|$, 
enabling us to obtain estimates 
with extremely small variance, 
as we show in Section~\ref{exps:params:qval}, and discuss theoretically in Section~\ref{subsec:varianceOpt}. 
\begin{example}
	Consider Figure~\ref{fig:inner:exampleWeights}, fix $q\in [0, \sfrac{1}{2}]$, and 
	suppose that $e=\{u,v\}$ is sampled.
	Then the estimators of $|\triSetNode{z_i}|$ assign value $(1-2q)/p$
	for all nodes 
	$z_i \in \neighborEdge{e}$, and 
	value $q|\triSetEdge{e}|/p$ for nodes $u,v\in e$, 
	where $p$ is the sampling probability of edge~$e$.\footnote{In our analysis we consider $p=1/m$, but any importance sampling probability distribution $p_e$ over $e\in \E$ can be used, provided that $p_e>0$ if $|\Delta_e|>0$.}
\end{example}

We will now show that the proposed estimates are unbiased 
with respect to~$|\triSetNode{v}|$, for all $v\in \V$, 
and \emph{any} value $q\in [0,\sfrac{1}{2}]$,
\begin{lemma}\label{lemma:unbiasednessSingleNode}
	For any value of the parameter $q\in[0,\sfrac{1}{2}]$,
	\begin{equation}\label{eq:EstimatorQ}
	X_q(v) = \sum_{e \in E : v\in e} \frac{q |\Delta_e|X_e}{p} + (1-2q)\sum_{e \in \E  : v\in \neighborEdge{e}}\frac{X_e}{p}
	\end{equation}
	is an unbiased estimator of $|\triSetNode{v}|$ for $v\in V$. That is, $\expectation[X_q(v)] = |\triSetNode{v}|$, where the expectation is taken over a randomly sampled edge $e\in \E$, and $X_e$ is a 0--1 
	random variable indicating if $e\in \E $ is selected.
\end{lemma}

Note that the estimator in Lemma~\ref{lemma:unbiasednessSingleNode} 
allows to flexibly select the parameter $q$, to 
minimize the variance of the estimates $X_q(v)$, for $v\in V$, 
leading to very accurate estimates with small variance when~$q$ is selected properly 
(see Section~\ref{subsec:varianceOpt}). 
We next use the result of Lemma~\ref{lemma:unbiasednessSingleNode} 
to obtain estimates $\estFunc_i(e): \E\mapsto\mathbb{R}^{+}_0$ of $\clustGenPartition_i$ 
for each set $V_i$, with $i\in [k]$, of the partition $\nodesetPartition$ by sampling a random edge $e\in \E$. 
First, to unify our notation, given a node $v\in \V$ let 
$|\wedgNodeSetGeneral{v}| = |\wedgNodeSetCenter{v}|$ if $\ast = \mathsf{c}$,  and 
$|\wedgNodeSetGeneral{v}| = |\wedgNodeSetHead{v}|/2$ if $\ast = \mathsf{h}$. 
We can then write $\clustMetric_v = |\triSetNode{v}|/|\wedgNodeSetGeneral{v}|$, 
i.e., $\clustMetric_v = \clustCoeff_v$ corresponds to the local clustering coefficient if 
$\ast = \mathsf{c}$ and $\clustMetric_v = \closureCoeff_v$ otherwise. 
Using this notation, we have the following.

\begin{lemma}
\label{lemma:unbiasedEdgePartition}
Let $e\in \E$ be an edge sampled uniformly from $\E$.
Then,  the random variable $\estFunc_j(e): \E\mapsto\mathbb{R}^{+}_0$ defined by
\[
\estFunc_j (e)= \sum_{e\in \E} \frac{X_e}{p} \frac{1}{|V_j|}\sum_{v\in V_j} \frac{a_q(v,e)}{|\wedgNodeSetGeneral{v}|}  \enspace, 
\]
where
\begin{equation}
\label{eq:SingleNodeWeightEdge}
	a_q(v,e) = q|\triSetEdge{e}|\mathbf{1}[v\in e] + (1-2q)\mathbf{1}[v\in \neighborEdge{e}]\,, 
\end{equation}
is an unbiased estimate of $\clustGenPartition_j$ 
for each set $V_j\in \nodesetPartition, j\in[k] $, i.e., $\expectation[f_j(e)] = \clustGenPartition_j$. 
Here, $X_e$ is a 0--1 random variable indicating if $e\in \E$ is sampled, and 
$|\wedgNodeSetGeneral{v}| = |\wedgNodeSetCenter{v}|$ for $\clustMetric=\clustCoeff$, and 
$|\wedgNodeSetGeneral{v}| = |\wedgNodeSetHead{v}|/2$ if $\clustMetric=\closureCoeff$.
\end{lemma}
\ifextended
The proof can be found in Appendix~\ref{app:missingProofs}.
\else
The proof can be found in the full version~\cite{triadExt}.
\fi

Lemma~\ref{lemma:unbiasedEdgePartition} shows that by  sampling a single random edge $e\in E$ we can obtain an unbiased estimate of multiple coefficients $\clustGenPartition_j$ associated sets $\V_j \in \nodesetPartition$, by accurately weighting the triangles in $\triSetEdge{e}$. 
This allows to \emph{simultaneously} update the estimates of multiple buckets~$\V_j$, 
for $j\in[k]$, differently from existing approaches~\cite{Kolda2014,Seshadhri2014,Zhang2017LocalCC}, 
where each triangle identified by the algorithm is used to estimate the coefficient of a \emph{single} partition.

\subsection{The \mainAlg\ algorithm}\label{subsec:MainAlg} 
\begin{algorithm}[t]
	\caption{\mainAlg}\label{alg:adaptiveBucketSampling}
	\KwIn{$G=(V,E), (\varepsilon_j)_{j=1}^k, \eta$, partition $\nodesetPartition$, $\clustMetric\in \{\clustCoeff,\closureCoeff\}$.}
	\KwOut{Estimates $f_j$ and bounds $\widehat{\varepsilon}_j$ s.t.\ $|f_j-\clustGenPartition_j| \le \widehat{\varepsilon}_j \le \varepsilon_j$ for each $j\in [k]$ w.p.~$>1-\eta$.}
	\lIf{$\psi = \alpha$}{$|\wedgNodeSetGeneral{v}| \gets |\wedgNodeSetCenter{v}|$ for $v\in V$\label{mainAlgLine:ifClust}}
	\lElse{$|\wedgNodeSetGeneral{v}| \gets |\wedgNodeSetHead{v}|/2$ for $v\in V$\label{mainAlgLine:ifClosure}}
	$f_j \gets 0$ for each $j\in [k]$\label{mainAlgLine:initFunct}\;
	$q\gets \mathtt{Fixq}(G, \nodesetPartition)$; $\varepsilon \gets \min\{\varepsilon_j\}$\label{mainAlgLine:optimizeQ}\;
	$\zeta, R_1, \ldots, R_k \gets \mathtt{UpperBounds}(G, \nodesetPartition,|\wedgNodeSetGeneral{v}|_{v\in V}, q)$\label{mainAlgLine:getUbs}\;
	$i\gets 0$; $\mathcal{S} \gets \emptyset$ ; $s\gets0$; $\eta_0\gets \eta/2$; $R\gets \max\{R_j, j\in[k]\}$\label{mainAlgLine:setParams}\;
	$s_{\max} \gets \frac{R^2}{\varepsilon^2} (\zeta + \log(1/\eta))$; $s_0 \gets \ceil{R\frac{ 3\log(4k/\eta_0)}{\varepsilon} + 1}$\label{mainAlgLine:ubSampleSize}\;
	\While{\textbf{\emph{not}} $\mathtt{StoppingCondition}(s_{\max}, s, (f_j)_{j\ge 1},\varepsilon )$\label{mainAlgLine:Loop}}{
			$\mathcal{S}_i \gets \mathtt{UniformSample}(E, s_i)$\label{mainAlgLine:sampleEdges}\;
			\ForEach{$e=(u,v) \in \mathcal{S}_i $\label{mainAlgLine:iterateSamples}}{
				\ForEach{$w\in \mathcal{N}_e$\label{mainAlgLine:forOuterNodes}}{
					$f_j(e) \gets f_j(e) + \frac{(1-2q)}{|\wedgNodeSetGeneral{w}|}$ such that $w\in V_j$\label{mainAlgLine:updateOuter}\;
				}
				$f_j(e) \gets f_j(e) + \frac{q|\triSetEdge{e}|}{|\wedgNodeSetGeneral{u}|}$ such that $u\in V_j$\label{mainAlgLine:updateInner1}\;
				$f_j(e) \gets f_j(e) + \frac{q|\triSetEdge{e}|}{|\wedgNodeSetGeneral{v}|}$ such that $v\in V_j$\label{mainAlgLine:updateInner2}\;
			}
			$s\gets s + s_i$; $\mathcal{S} \gets \mathcal{S} \cup \mathcal{S}_i$\label{mainAlgLine:updateSamples}\;
			$f_j \gets \tfrac{m}{s|V_j|} \sum_{e\in \mathcal{S}}f_j(e)$ for each $j=1\dots,k$\label{mainAlgLine:updateEstimates}\;
			$\widehat{\varepsilon}_j \gets \mathtt{ComputeEmpiricalBound}(\mathcal{F}, \mathcal{S}, \eta_i)$\label{mainAlgLine:computeEmpiricalBounds}\;
			$i\gets i+1$; $\eta_i \gets \eta_{i-1}/2$\label{mainAlgLine:setNextIter}\;			
		}
		\KwRet{$(f_j, \widehat{\varepsilon}_j)$ for each $j=1\dots,k$}\label{mainAlgLine:return}\;
\end{algorithm}

Algorithm~\ref{alg:adaptiveBucketSampling} presents~\mainAlg.
The algorithm first initializes $|\wedgNodeSetGeneral{v}|$ for each $v\in \V$
according to the coefficients $\clustMetric$ to be estimated 
(lines~\ref{mainAlgLine:ifClust}--\ref{mainAlgLine:ifClosure}). 
\mainAlg then proceeds to initialize the variables $\estFunc_j$, for $j=1,\dots,k$, 
corresponding to the estimates of $\clustGenPartition_j$ to be output~(line~\ref{mainAlgLine:initFunct}). 
It then sets $\varepsilon$ to the smallest $\varepsilon_j$, for $j\in[k]$
(i.e., $\varepsilon$ is  
the smallest upper bound $\varepsilon_j$ required by the user, see~\Cref{prob:partionEstimationAdditive}). 
\mainAlg then selects the best value of the parameter $q\in [0,1/2]$ to guarantee a fast termination of~\mainAlg
by solving a specific optimization problem 
(line~\ref{mainAlgLine:optimizeQ},  
subroutine \texttt{Fixq}, see Section~\ref{subsec:varianceOpt}). 
Intuitively, the value of $q$ is fixed such that the maximum 
variance of $f_j$, for $j\in[k]$, is minimized,  
yielding a fast convergence of~\mainAlg. 
In fact, as we will show, the termination condition of \mainAlg
considers the \emph{empirical} variance of the estimates $f_j$, over the sampled edges.

The function $\mathtt{UpperBounds}$ (line~\ref{mainAlgLine:getUbs}) computes $R_j$, for $j=1,\dots,k$, 
corresponding to bounds to the maximum value that a random variable $f_j(e)$ 
can take over a randomly sampled edge $e \in \E$, 
i.e., $\estFunc_j(e) \le R_j$ almost surely, for each $j\in[k]$. 
In addition $\mathtt{UpperBounds}$ computes
$\zeta$, an upper bound on the pseudo-dimension of the functions $f_j(e)$, $j\in[k]$. 
We use $\zeta$ and $R = \max_j{R_j}$ (line~\ref{mainAlgLine:setParams}) to obtain a bound on the maximum number of samples $s_{\max}$ to be explored by \mainAlg (line~\ref{mainAlgLine:ubSampleSize}).
We discuss the function $\mathtt{UpperBounds}$, and prove the bound on the sample-size 
in Sections~\ref{subsec:upperbounds} and~\ref{subsec:pdimBound}, respectively.

Auxiliary variables are then initialized (line~\ref{mainAlgLine:setParams}): 
index $i$~keeps track of the iterations, 
$\mathcal{S}$~maintains the bag of sampled edges processed, 
$s$~counts the number of samples processed, and
$\eta_0$~is used to obtain the probabilistic guarantees of~\mainAlg.
Finally, $s_0$ corresponds to the initial sample size, i.e., 
the minimum sample size 
for which~\mainAlg can provide tight guarantees (line~\ref{mainAlgLine:ubSampleSize}), that we discuss in Section~\ref{subsec:adaptBounds}. 

The main loop of \mainAlg is entered in line~\ref{mainAlgLine:Loop}.
At the $i$-th iteration of the loop, 
\mainAlg samples a bag of $s_i$ edges uniformly at random (line~\ref{mainAlgLine:sampleEdges}), and 
for each edge it computes $f_j(e)$ as defined in Lemma~\ref{lemma:unbiasedEdgePartition} 
(lines~\ref{mainAlgLine:updateOuter}-\ref{mainAlgLine:updateInner2}). 
Then~\mainAlg updates the sample size and the bag of edges processed (line~\ref{mainAlgLine:updateSamples}) together with estimates $\estFunc_j$, for $j\in[k]$~(line~\ref{mainAlgLine:updateEstimates}) invoking the function $\mathtt{ComputeEmpiricalBound}$ to compute \emph{non-uniform bounds} $\widehat{\varepsilon}_j$ on the deviations $|\estFunc_j-\clustGenPartition_j|$ such that $|\estFunc_j-\clustGenPartition_j| \le \widehat{\varepsilon}_j$ with controlled error probability (line~\ref{mainAlgLine:computeEmpiricalBounds}). 
These bounds are \emph{empirical}, tight, and adaptive 
to the samples in $\sampleSet$,
leveraging the variance of the estimates $f_j(e), e\in \sampleSet$, 
and the upper bounds $R_j$, for $j\in[k]$ (see Sec.\ \ref{subsec:adaptBounds}). 
Finally, a set of variables is computed for the next iteration, 
if the stopping condition of the main-loop is not met (line~\ref{mainAlgLine:setNextIter}). 
The call to \texttt{StoppingCondition} in \mainAlg 
returns ``$\mathsf{true}$'' if one the following two conditions holds, which depend on the processed samples $\sampleSet$:

{(1)} if $s \ge s_{\max}$ then with probability at least $1-\eta/2$ it holds
$|\estFunc_j -\clustGenPartition_j|\le \varepsilon \le \varepsilon_j$;

{(2)} if $\widehat{\varepsilon}_j \le \varepsilon_j$ then with probability at least $1-\eta/2$ it holds 
$|\estFunc_j -\clustGenPartition_j| \le \widehat{\varepsilon}_j \le \varepsilon_j$, 
for each $j\in[k]$.

When the main loop terminates \mainAlg outputs 
$(f_j, \varepsilon_j)$ in case (1), or
$(f_j, \widehat{\varepsilon}_j)$ in case (2), 
for each set $V_j\in \nodesetPartition$.
Note that our estimation addresses both requirements of Problem~\ref{prob:partionEstimationAdditive} 
as we guarantee that all the estimates $\estFunc_j$ are within $\varepsilon_j$ distance to $\clustGenPartition_j$, 
and furthermore,  
we report accurate error bounds $\widehat{\varepsilon}_j\le \varepsilon_j$, for all $j\in[k]$, 
as $\clustGenPartition_j \in [f(V_j)-\widehat{\varepsilon}_j, f(V_j)+\widehat{\varepsilon}_j]$ with high probability. 
For example, even for as small $\varepsilon_j$ as $10^{-2}$ 
the reported adaptive guarantees of \mainAlg can be at most $\widehat{\varepsilon}_j \approx 10^{-3}$ 
in a few tens of seconds (see Section~\ref{exps:sota}). 

\subsection{Analysis}
\label{subsec:algAnalysis}

In this section we present in more detail all the components of~\mainAlg, 
and we analyze its accuracy. 
The next lemma states that the estimates of Algorithm~\ref{alg:adaptiveBucketSampling} are unbiased, 
which is important to prove tight concentration.
\begin{lemma}\label{lemma:unbiasedOutput}
	For the output $\estFunc_j$, $j=1,\dots,k$, of Algorithm~\ref{alg:adaptiveBucketSampling}, it~holds:
	\[
	\expectation[\estFunc_j] = \frac{1}{|V_j|}\sum_{v\in V_j}\clustMetric_v = \clustGenPartition_j  \enspace,
	\]
	that is, the estimates of Algorithm~\ref{alg:adaptiveBucketSampling} are unbiased, 
	for both local triadic coefficients $\clustMetric\in\{\clustCoeff,\closureCoeff\}$.
\end{lemma}
\ifextended
The proof can be found in Appendix~\ref{app:missingProofs}.
\else
The proof can be found in the extended version~\cite{triadExt}.
\fi

Next we provide a bound for the variance of the estimates $\estFunc_j$, for~$j\in[k]$,
returned by~\mainAlg.
\begin{lemma}\label{lemma:varianceBoundOutput}
	For the estimates $\estFunc_j$, $j=1,\dots,k$, of Algorithm~\ref{alg:adaptiveBucketSampling}, it is
	\[
	\Var[\estFunc_j] \le \frac{1-p}{sp} \left(\frac{1}{|V_j|}\sum_{v\in V_j}\clustMetric_v\right)^2  \enspace,
	\]
	for both local triadic coefficients $\clustMetric\in\{\clustCoeff,\closureCoeff\}$.
\end{lemma}
\ifextended
The proof can be found in Appendix~\ref{app:missingProofs}.
\else
The proof can be found in the extended version~\cite{triadExt}.
\fi

\subsubsection{Bounding the sample size}\label{subsec:pdimBound}
To present the bound on the sample size (i.e., $s_{\max}$), 
we first introduce the necessary notation. 
Given a finite domain $\DomX$ and $\RangeSet\subseteq 2^{\DomX}$ 
a collection of subsets of $\DomX$,\footnote{The set containing all possible subsets of $\mathcal{X}$.} 
a \emph{range-space} is the pair $(\DomX, \RangeSet)$. 
We say that a set $X\subseteq\DomX$ is \emph{shattered} by the range-set $\RangeSet$, 
if it holds $\{Q\cap X: Q\in \RangeSet \} = 2^{X}$. 
The \emph{VC-dimension} $\mathsf{VC}(\mathcal{X},\RangeSet)$ of the range-space is 
the \emph{size} of the largest subset $X\subseteq \mathcal{X}$ such that $X$ can be shattered by $\RangeSet$. 
Given a family of functions $\Fset$ from a domain $\mathcal{H}$ with range $[a,b]\subseteq \mathbb{R}$, for a function $f\in \Fset$ we define the subset $Q_f$ of the space $\mathcal{H} \times [a,b]$ as
\[
Q_f = \{(x,t) : t\le f(x)\}, f\in \mathcal{F}  \enspace.
\] 
We then define $\mathcal{F}^+ = \{Q_f, f\in \mathcal{F}\}$ 
as the range-set over the set $\mathcal{H}\times[a,b]$.
With these definitions at hand, 
the \emph{pseudo-dimension} $\mathsf{PD}(\mathcal{F})$ of the family of functions $\mathcal{F}$ 
is defined as $\PD(\mathcal{F}) = \VC(\mathcal{H}\times[a,b], \mathcal{F}^+)$. 
For illustrative examples we refer the reader 
to the literature~\cite{Riondato2018Abra,Riondato2018MisoSoup,ShalevShwartz2014UML}. 
In our work, the domain $\mathcal{H}$ corresponds to 
the set of edges to be sampled by~\mainAlg to compute the estimators 
$f_j = f_{\sampleSet,j}$,~$j=1,\dots,k$, that is, $\DomH=E$.\footnote{We write $f_{\sampleSet,j}$ to explicit that a function depends on the bag of samples $\sampleSet$.} 
We also define the set of functions $\mathcal{F} = \{f_{\sampleSet,j} : j=1,\dots,k\}$, 
which corresponds to the set containing all functions $\estFunc_j, j=1,\dots,k$ in output to~\mainAlg.

\begin{theorem}\label{theo:boundSSPDIM}
Let \Fset be the set of functions defined above with $\PD(\Fset) \le \pdim$, 
and let $\varepsilon, \eta \in (0,1)$ be two parameters. If
	\[
	|\sampleSet| \ge \frac{(b-a)^2}{\varepsilon^2} \left(\zeta + \log \frac{1}{\eta}\right)  \enspace,
	\]
	then with probability at least $1-\eta$ over the randomness of the set $\sampleSet$ it holds that $ \left|f_j - \clustGenPartition_j \right| \le \varepsilon, \text{for each } j=1,\dots,k$.
\end{theorem}

Note that we cannot compute $\PD(\Fset)$ from its definition as this requires exponential time in general, 
hence we now prove a tight and efficiently computable upper bound $\pdim$ such that $\PD(\Fset) \le \pdim$.
This enables us to use Theorem~\ref{theo:boundSSPDIM} and obtain a deterministic upper bound on the sample size of~\mainAlg (i.e., $s_{\max}$ in line~\ref{mainAlgLine:ubSampleSize}).

First, let $\bucketColor_v = |\{\V_j \in \nodesetPartition : \text{exists } u \in (\neighborNode{v}\cup \{v\}), u \in \V_j\}|$, 
i.e., $\bucketColor_v$ is the number of distinct sets $V_j$, for $j\in [k]$ from $\nodesetPartition$ 
containing a node in the set $\neighborNode{v}\cup \{v\}$ for a given node $v\in V$. 

\begin{example}\label{ex:pdims}
	%\edit{
	In the following extreme cases:
	($i$)~when every node is in a distinct bucket (i.e., $k=n$) then $\bucketColor_v \le d_{\max}+1$, for $v\in \V$, with $d_{\max}$ being the maximum degree of a node in~$\V$; 
	($ii$)~when each node is in the same bucket then $\bucketColor_v = 1$ for every node~$v\in V$.%}
\end{example}

Now let
\begin{equation}\label{eq:boundTermPDIM}
	\bucketColorSup = \max_{e=\{u,v\}\in E} \{\bucketColor_{z} : z = \arg\min\{\nodeDeg{u}, \nodeDeg{v}\} \}  \enspace. 
\end{equation}
Intuitively $\bucketColorSup$ corresponds to the largest number of buckets
for which a sampled edge $e\in \E$ yields a non-zero value for the estimate $\estFunc_j(e)$. 
Clearly, a trivial bound is $\bucketColorSup \le \min \{k, d_{\max} +1\}$,
 which can be very loose.
We next present the bound on the pseudo-dimension associated to $\PD(\Fset)$, 
recall that $\Fset$ corresponds to set of estimates $f_j, j=1,\dots,k$, of~\mainAlg.

\begin{proposition}\label{theo:PdimBound}
	The pseudo-dimension $\mathsf{PD}(\mathcal{F})=\zeta$ is bounded as $\pdim \le \floor{\log_2 \bucketColorSup} + 1$.
\end{proposition}
\ifextended
The proof can be found in Appendix~\ref{app:missingProofs}.
\else
The proof can be found in the extended version~\cite{triadExt}.
\fi

As an example of the powerful result of Theorem~\ref{theo:PdimBound} consider the following corollary.
\begin{corollary}\label{coroll:pdimLima}
 Fix $q=0$ and take each node in a different set in~$\nodesetPartition$ (i.e., $k=n$).
 Let $\G$ be a star graph. Then by Theorem~\ref{theo:PdimBound} it holds~$\pdim \le  1$.
\end{corollary}
\ifextended
The proof can be found in Appendix~\ref{app:missingProofs}.
\else
The proof can be found in the extended version~\cite{triadExt}.
\fi

Corollary~\ref{coroll:pdimLima} shows that our result provides a tight bound on the pseudo-dimension associated to the family of functions $\mathcal{F}$. 
For comparison, under the setting of Corollary~\ref{coroll:pdimLima}, 
then our setting maps to the one of~\citet{Lima2022ClusteringLatin}. 
In their work, the authors prove $\zeta\le \bigO(\log n)$ 
for the graph in Corollary~\ref{coroll:pdimLima}, 
while our result states $\zeta\le 1$ yielding a $\bigO(\log n)$ improvement, which is the maximum attainable. 
Clearly, a smaller upper bound for $\zeta$ implies a significantly smaller sample size required to guarantee the desired accuracy when such bound is used for Theorem~\ref{theo:boundSSPDIM}.

\ifextended
We can further refine the bound on the pseudo-dimension when $q=0$ or $q=1/2$, to this end let $\bucketColorSup' = \max_{e\in E} \{\bucketColor_z': z=\arg\min\{\nodeDeg{u}, \nodeDeg{v}\}\}$ and $\bucketColor_v'$ is defined as $\bucketColor_v$ but on $\neighborNode{v}$, i.e., it does not consider the extended neighborhood $\neighborNode{v}\cup \{v\}$.
\begin{corollary}\label{coroll:pdimQvals}
	The pseudo-dimension $\mathsf{PD}(\Fset)=\zeta$  is bounded by $\pdim \le \floor{\log_2 \bucketColorSup'} + 1$ when $q=0$, while $ \pdim \le 2$ for $q=\frac{1}{2}$.
\end{corollary}
The proof can be found in Appendix~\ref{app:missingProofs}.
\else
We can refine the bounds on $\zeta$ for the values of $q=0$ or $q=1/2$. We discuss such refinement in our extended manuscript~\cite{triadExt}.
\fi

We conclude by noting that $\bucketColorSup$ can be efficiently computed in $\bigO(m)$ time complexity with a linear scan of the edges of the graph.

\subsubsection{Computing adaptive error bounds}\label{subsec:adaptBounds}
We detail how~\mainAlg leverages adaptive and variance-aware bounds to determine the distance of $\estFunc_j$ from $\clustGenPartition_j$, for $j\in[k]$. We need a key concentration inequality.
\begin{theorem}[Empirical Bernstein bound~\cite{Maurer2009EmpiricalBern, Mnih2008EmpiricalStopping}]\label{theo:empiricalBern}
	Let $X_1,\allowbreak\dots,\allowbreak X_s$ be $s$ independent random variables such that for all $i=1,\dots,s$, $\expectation[X_i] = \mu$ and $\Prob[X_i\in [a,b]]=1$, and $ \widehat{\mathit{v}}= \frac{1}{s} \sum_{i=1}^s (X_i - \widebar{X}_s)^2$ where $\widebar{X}_s = \tfrac{1}{t}\sum_{i=1}^s X_i$. Then,
	\[
	\left|\frac{1}{s}\sum_{i=1}^s X_i - \mu \right| \le \sqrt{\frac{2\widehat{\mathit{v}} \log(4/\eta)}{s}} + \frac{7(b-a)\log(4/\eta)}{3(s-1)}  \enspace.
	\]
\end{theorem}
The above theorem connects the empirical variance over the samples $\estFunc_j(e)$ 
with the distance of $\estFunc_j$ to $\clustGenPartition_j$, providing a powerful result.
Therefore we can use Theorem~\ref{theo:empiricalBern} to obtain the function $\mathtt{ComputeEmpiricalBound}$. 
That is, given a partition $V_j$ and a bag of edges $\sampleSet$ sampled so far, we 
obtain 
\[
\widehat{\varepsilon}_j =  \sqrt{\frac{2\widehat{\mathit{v}}_j \log(4k/\eta_i)}{s}} + \frac{7(R_j)\log(4k/\eta_i)}{3(s-1)}
\] 
at iteration $i\ge 0$. 
Note that 
\[
\widehat{\mathit{v}}_j=  \frac{1}{s} \sum_{e\in \sampleSet} (f_j(e)- f_j)^2,
\] 
which can be obtained in linear time (i.e., $|\mathcal{S}|$), 
assuming the values of $f_j(e)$ are retained over the iterations. 
The above result provides also a criterion on how to set $s_0$. That is, $s_0$ 
should be at least $\ceil{R\frac{ 3\log(4k/\eta_0)}{\varepsilon} + 1}$ 
(see~\mainAlg in line~\ref{mainAlgLine:setParams})
in the very optimistic case that the empirical variance $\widehat{\mathit{v}}_j = 0$, for each $j\in[k]$.
We can now prove the guarantees offered by~\mainAlg.

\begin{theorem}\label{theo:probGuarantees}
The output of~\mainAlg $(\estFunc_j, \widehat{\varepsilon}_j)$, for $j\in[k]$, is such that with probability 
at least $1-\eta$ it is $|f_j-\clustGenPartition_j|\le  \widehat{\varepsilon}_j \le \varepsilon_j$, 
simultaneously for all sets $V_j$, for $j\in[k]$.
\end{theorem}
\ifextended
The proof can be found in Appendix~\ref{app:missingProofs}.
\else
The proof can be found in the extended version~\cite{triadExt}.
\fi

\subsubsection{Minimizing the variance}\label{subsec:varianceOpt}

The value of the parameter $q$ (in the estimator from Lemma~\ref{lemma:unbiasedEdgePartition}) plays a key role for~\mainAlg, enabling extremely accurate estimates if set properly (see Section~\ref{exps:params:qval}). 
There are many criteria to select the value of $q$,  
the most natural one would be to fix $q$ such that~\mainAlg processes the minimal possible sample-size to terminate its main loop. Unfortunately, we cannot optimize directly for such property.

Instead, to select the value of $q$ we minimizing the maximum variance of the functions $f_j$ over all $j\in[k]$.
In fact, a smaller variance of $f_j$ 
allows~\mainAlg to terminate its loop by processing a small number of samples.
To do so, we first sample a very small bag~$\sampleSet$ of edges from $\E$, 
and estimate the variance of the functions $f_j$ over $\sampleSet$ \emph{for each} 
possible value of $q$ 
\ifextended
(Algorithm~\ref{alg:fixQ}).
\else
(see~\cite{triadExt}).
\fi 
The variance of $f_j$ can, in fact, be expressed as a quadratic equation in the parameter~$q$, i.e., 
there exist efficiently-comput\-able values $A_j,B_j,C_j$ such that $\widehat{V}_q(f_j) = A_j + B_j q + C_j q^2$.
\ifextended
See Appendix~\ref{app:varianceEst} for how such values are computed.
\else
We detail in our extended version how such values are computed.
\fi
Using such formulation, we can solve a quadratic optimization problem, 
minimizing the maximum variance over $\widehat{V}_q(f_j)$ over all possible values of $q\in [0,1/2]$ and for each $j\in[k]$. 
We first show that a sufficiently small bag of sample\edit{s} can yield a good estimate of the variance $\widehat{V}_q(f_j)$.

\begin{lemma}\label{lemma:unbiasedVarEst}
There exist values $A_j$, $B_j$, and $C_j$ obtained over $c\ge 2$ sampled edges such that $\expectation[\widehat{V}_q(f_j)]= \Var[f_j]$ for
\[
\widehat{V}_q(f_j) = \frac{m^2}{(c-1)|V_j|^2} \left(\sum_{e\in \sampleSet}A_j + qB_j + C_j q^2\right) \enspace.
\]
\end{lemma}
\ifextended
The proof can be found in Appendix~\ref{app:missingProofs}.
\else
The proof can be found in the extended version~\cite{triadExt}.
\fi

The above lemma tells us that the estimates $\widehat{V}_q(\estFunc_j) $ provide an unbiased approximation of the variance $\Var[\estFunc_j]$ for each function~$\estFunc_j$.
The next lemma tells us that such estimates are also a good approximation, for each value of $q\in[0,1/2]$.

\begin{lemma}\label{lemma:varEpsilon}
	There exist a small value $\varepsilon'>0$ such that $\varEst_q(\estFunc_j)\in[\Var[f_j]-\varepsilon', \Var[f_j]+\varepsilon']$ for each partition $j=1,\dots,k$ and value of $q\in[0,1/2]$, with high probability.  
\end{lemma}
\ifextended
The proof can be found in Appendix~\ref{app:missingProofs}.
\else
The proof can be found in the extended version~\cite{triadExt}.
\fi

To find the optimal value of $q$, given that we have a good estimate of $\Var[f_j]$, for each $j\in[k]$, 
we solve the following convex problem, minimizing the largest estimated variance $\widehat{V}_q(f_j)$ 
over~$j\in[k]$.

\begin{problem}[Optimization of $q$]
	Let $A_j$, $B_j$ and $C_j$ be the coefficients contributing to estimate the variance $\varEst_q(f_j)$ then solving the following quadratic program yields $x^* = {q_{\mathtt{alg}}}$ such that $\varEst_{{q_{\mathtt{alg}}}}(f_j) \in \varEst_{q^*}(f_j)\pm \varepsilon'$ for small $\varepsilon'$, and controlled error probability, where $q^*$ is the optimal value $q$ minimizing the variance of functions $f_j, j=1,\dots,k$.
	\begin{eqnarray*}
  & \mbox{minimize }  & t  \\
  & \mbox{subject to } & t \ge A_j+B_jx+C_jx^2, ~\mbox{ for all }~ j=1, \dots, k,\\
	& \mbox{and }	& x \in [0,1/2]  \enspace. 
	\end{eqnarray*}
\end{problem}
Solving this quadratic problem with state-of-the-art solvers is efficient: 
$x$ is in $\mathbb{R}$,
the sample size $c$ used to compute the terms $A_j,B_j,C_j$ 
\ifextended
in Algorithm~\ref{alg:fixQ} 
\else
\fi
is a small constant,
the objective function and the constraints are convex, and 
in practice $k$ is small, 
yielding an overall efficient procedure (see Section~\ref{exps:rtAndParams}).
Note that in~\mainAlg we denote collectively with \texttt{Fixq} 
the procedure that computes both $\varEst_q(f_j)$, for $j\in [k]$, and 
identifies $x^* = {q_{\mathtt{alg}}}$.

\subsubsection{$\mathtt{UpperBounds}$}\label{subsec:upperbounds}
A detailed description of the $\mathtt{UpperBounds}$ subroutine, which requires linear time complexity in $m$, is reported
\ifextended
in~\Cref{appsubsec:upperbounds}.
\else
our extended version~\cite{triadExt}.
\fi

The next lemma shows that $\mathtt{UpperBounds}$ outputs an upper bound on the ranges of $f_j(e), e\in \E$, and $\zeta$ as from Proposition~\ref{theo:PdimBound}.

\begin{lemma}\label{lemma:upperBounds}
The output of $\mathtt{UpperBounds}$ is such that $\estFunc_j(e) \le R_j$ almost surely, 
for each $e\in \E$, and partition $\V_j, j=1,\dots,k$, 
where $\zeta$ corresponds to the pseudo-dimension bound from Proposition~\ref{theo:PdimBound}.
\end{lemma}

\subsection{Practical optimizations}\label{subsec:algPractOpt}
In this section we introduce some practical optimizations to improve the performances of~\mainAlg, 
with minimal complexity overhead.

\subsubsection{Improved empirical bounds}
We can further refine the empirical bounds for terms $\widehat{\varepsilon}_j$, for $j\in[k]$ (in line~\ref{mainAlgLine:computeEmpiricalBounds}) computed by~\mainAlg by leveraging the following result.
\begin{theorem}[Predictable Plugin-Empirical Bernstein Confidence Interval (PrPl-EB-CI)~\cite{WaudbySmith2023Betting, Shekhar2023OptimalBetting}]
\label{theo:predictablePlugin}
Let $X_1,\dots,X_s$ be $s$ i.i.d.\ random variables such that for all $i=1,\dots,s$, 
it is $\expectation[X_i] = \mu$ and $\Prob[X_i\in [a,b]]=1$, with $R=|b-a|$, and let
\begin{align*}
		&\omega(\lambda) = \frac{-\log(1-R \lambda) -  R\lambda}{4} ~\text{ for } \lambda\in[0,\sfrac{1}{R}), 
		\quad \widehat{\mu}_j = \frac{1}{j}\sum_{i=1}^j X_i, \\
		&\widehat{\sigma}_j = \frac{R^2/4+\sum_{i=1}^j (X_i-\widehat{\mu}_{i-1})^2}{j}, 
		\quad  {\lambda_{j,s}} = \min\left\{\sqrt{\frac{2\log(2/\eta)}{s\widehat{\sigma}_{j-1}}}, \frac{1}{2R} \right\}
	\end{align*}
	for $j=1,\dots, s$, where $\widehat{\mu}_0 = 0 $ and $\widehat{\sigma}_0 = R^2/4$. 
	Then it holds that 
	\begin{equation*}
		\left| \frac{\sum_{i=1}^s \lambda_{i,s} X_i}{\sum_{i=1}^s \lambda_{i,s}} - \mu\right|  \le  \frac{\log(2/\eta)+(2/R)^2 \sum_{i=1}^s [\omega(\lambda_{i,s}) (X_i-\widehat{\mu}_{i-1})^2]}{\sum_{i=1}^s \lambda_{i,s}}
	\end{equation*}
	with probability at least $1-\eta$.
\end{theorem}
While similar in spirit to~Theorem~\ref{theo:empiricalBern}, the above bound often yields a sharper empirical bound on the values $\widehat{\varepsilon}_j$. 
Note that 
the above estimator yields a slightly more complicated formulation, 
i.e., the output of~\mainAlg corresponds to $ f_j = \frac{\sum_{i=1}^s \lambda_{i,s} f_j(e_i)}{\sum_{i=1}^s \lambda_{i,s}}$ where $ f_j(e_i)$ is the estimate associated to bucket $V_j$, for $j\in [k]$ evaluated for the $i$-th sample from $\sampleSet$.

In addition, the stopping condition is evaluated using the bounds $\frac{\log(2/\eta)+(2/R)^2 \sum_{i=1}^s [\omega(\lambda_{i,s}) (X_i-\widehat{\mu}_{i-1})^2]}{\sum_{i=1}^s \lambda_{i,s}} = \widehat{\varepsilon}_j \le \varepsilon_j$, for $j\in[k]$, and $f_j = \frac{1}{s}\sum f_j(e_i)$ if $s\ge s_{\max}$. 
In Section~\ref{sec:exps}, we leverage~\Cref{theo:predictablePlugin}, but for ease of notation and presentation we introduced~\mainAlg 
with the results of Theorem~\ref{theo:empiricalBern}.

\subsubsection{Fixed sample size variant}\label{subsec:algFixed}
In this section we briefly describe a variant that we call~\algFixedSS, 
which leverages a \emph{fixed} sample size schema. 
That is, in many applications, concentration bounds (e.g., Theorem~\ref{theo:empiricalBern}), 
even if tight and empirical, can still  be conservative. 
Hence we modify~\mainAlg to leverage the novel estimators discussed in Lemma~\ref{lemma:unbiasedOutput}, 
and the adaptive procedure to select the value of $q$ as described in Section~\ref{subsec:varianceOpt}, 
but we modify the stopping condition of~\mainAlg. 
That is, we only require~\algFixedSS to process \emph{at most} a number $s\ge 1$ samples as provided in input by the user. 
This is of interest in many applications where strictly sublinear time complexity is required.
In fact, such modification strictly enforces a small number of samples to be processed. 
We show that such variant outputs highly accurate estimates 
\emph{of each} bucket,
by processing only 1\textperthousand\ edges on most graphs (see Section~\ref{exps:accuracyAndEff}).

\subsubsection{Filtering very small degree-nodes}\label{subsec:filtersmall}

To enable better performance for~\mainAlg, we process the graph $\G$ by removing the nodes with \emph{small} degree obtaining a graph $\G'$ where all the nodes' degrees, in $\G$, are above a certain threshold. 
This step decreases the variance of the estimates computed by~\mainAlg, as small degree nodes can have high values for their metric $\clustMetric$ (i.e., close to $1$), but sampling may perform poorly in approximating the values $\clustMetric$, when computing the estimate $f_j$, for $j\in[k]$, similarly to what noted by~\citet{Kutzkov2013StreamCC,Lima2022ClusteringLatin}.
The key challenge is to bound the overall total work to $\Theta(n)$ time complexity, making such processing negligible, 
and retaining the correct information to recover the solution to Problem~\ref{prob:partionEstimationAdditive} on $\G$.

Our approach identifies a threshold over the node degree distribution, namely $\beta$, for which computing all the nodes' triangles under such threshold requires at most linear time in $n$, i.e., bounded by $Cn$ for a small fixed constant $C$. 
Obtaining $\triSetNode{v}$ for a node $v\in \V$ with degree $\nodeDeg{v}$ requires at most $\bigO(\nodeDeg{v}^2)$ time. Therefore we first identify the value $\beta$ such that
$\beta = \max_{i= 1,\dots,d_{\max}}$ such that $\sum_{j=1}^i j^2 D_j \le Cn, $
where $D_i$ denotes the number of nodes in $\G$ such that their degree is exactly $i$, that is $D_i =|\{v \in V : \nodeDeg{v} = i\}|$. 
Clearly, $\beta$ can be computed in linear time $\Theta(n)$ 
by iterating all nodes, assuming constant access to their degrees $\nodeDeg{v}$. 
Given the threshold~$\beta$, for each node $v\in \V$ with degree less than $\beta$ we compute exactly the triangles $\triSetNode{v}$, 
keeping track for all the nodes $v\in \V'$ that have degree higher than $\beta$ of the triangles containing at least one removed node. 
We show that our approach is efficient and yields no additional estimation error for~\mainAlg in 
\ifextended
Appendix~\ref{app:filterSmall}.
\else
the extended version of the manuscript~\cite{triadExt}.
\fi 

\subsection{Time and memory complexity}\label{subsec:timeandmem}
\para{Time complexity.}
We recall that the filtering step requires $\bigO(n)$ total work, and the routine \texttt{UpperBounds} requires $\bigO(km)$, while \texttt{Fixq} requires $\bigO(d_{\max} + \mathtt{T}_{\text{QP}})$. 
Note that $\mathtt{T}_{\text{QP}}$, the complexity of solving the convex minimax problem, 
is negligible as the optimization is over $k$ total convex non-integer constraints, and for our problem formulation  we consider $k$ as a (possibly large) constant.
 
Finally, the largest time complexity to perform the adaptive loop over 
$T=s_{\max}$ total iterations of~\mainAlg, 
can require up to $\bigO(T(d_{\max} + T))$ time since:
(1) each edge can be incident to $d_{\max}$ nodes; 
(2) the additional $T^2$ term is from computing the empirical variance at each iteration.\footnote{This complexity can be reduced to $T$ by relying on the wimpy variance.} 
Hence the \emph{worst} case complexity is $\bigO(R^2\varepsilon^{-2}(\zeta + \log1/\eta) (d_{\max} + T) + m)$. 
Note that such analysis is extremely pessimistic, in fact,
in practice the complexity of~\mainAlg is instead close to $\bigO(R\varepsilon^{-1}(\log k/\eta) d_{\max} + km)$, as we often observe~\mainAlg to terminate after a small number of iterations of its main loop, implying that the processed edges are at most $\bigO(R\varepsilon^{-1}(\log k/\eta))$. 
Hence, when $\bigO(R\varepsilon^{-1}(\log k/\eta))$ is a small fraction $\omega\ll 1$ of $m$ (e.g., 1\% of $m$) then the total complexity is bounded as $\bigO(m(d_{\max}\omega + k))$, capturing the efficiency of~\mainAlg.

\smallskip
\noindent
\para{Memory complexity.} The memory complexity of~\mainAlg is comparable to existing state-of-the-art methods~\cite{Seshadhri2014} for estimating the local clustering coefficient, requiring $\bigO(m)$ memory. 
In more detail,~\mainAlg requires memory $\bigO(m+|\sampleSet|k + kn)$, 
where $|\sampleSet|$ is the size of the samples processed by~\mainAlg. 
Clearly $|\sampleSet|=s$ for~\algFixedSS. 
When processing very large graphs with limited resources such complexity can be prohibitive, 
hence in such cases, we should rely on a (semi-)streaming or distributed (e.g., MPC) implementation of \mainAlg~\cite{Becchetti2010,Kolda2014}, an interesting future direction.

\subsection{Adaptive guarantees}\label{subsec:adaptivity}

We briefly discuss the main advantages and limitations of~\mainAlg
in the adaptive case. 
We observe that~\mainAlg has a significant advantage to solve~\Cref{prob:partionEstimationAdditive}. Given a input graph~$G$ and a partition $\nodesetPartition$,
\mainAlg selects its estimators $f_j$, for $j\in[k]$, by properly fixing the parameter $q$ \emph{adaptive\-ly}. 
That is, the parameter $q$ is optimized by~\mainAlg directly on $G$, leading to estimators $f_j$, for $j\in[k]$, 
with small variance as captured by our theoretical results in~\Cref{lemma:varEpsilon} and in practice in~\Cref{exps:accuracyAndEff}. 
\mainAlg also adapts the number of processed samples 
($|\sampleSet|$) by leveraging the results from~\Cref{theo:predictablePlugin}. 
Such bounds involve the empirical variance $\widehat{\sigma}_j$, for $j\in[k]$, 
optimized by~\mainAlg through the parameter $q$---and the upper bounds $R_j$ 
on the range of $f_j(e)$, $j\in [k]$, $e\in E$. 
Where each $R_j$ depends on the node distribution over $\nodesetPartition$ and $G$.
Obtaining a non-trivial characterization of $R_j$ is extremely challenging, but we observe the following. 
When the values $R_j$ are large in practice the obtained empirical bounds $\widehat{\varepsilon}_j$ 
may be loose with respect to the actual deviations $|f_j-\clustGenPartition_j|$.
Instead, when $R_j \in \bigO(1)$ then~\mainAlg processes  $\widetilde{\bigO}(1/\varepsilon)$ samples (a significant improvement over $\widetilde{\bigO}(1/\varepsilon^2)$).\footnote{In $\widetilde{\bigO}(\cdot)$ we ignore logarithmic factors.} 
This is an inherent trade-off,
when $R\in \bigO(1)$ then the runtime of~\mainAlg is almost constant providing tight bounds $\widehat{\varepsilon}_j \le |f_i-\clustGenPartition_i|$, which depend on $G$ and $\nodesetPartition$.

\section{Experimental evaluation}
\label{sec:exps}

In this section we present our extensive experimental evaluation.
We start by first describing the setup, 
and then we discuss the results of our research questions.

\begin{table}[t]
\caption{Datasets used in the experimental evaluation. Statistics show: $n$ the number of nodes, $m$ the number of edges, $d_{\max}$ the maximum degree, $\bar{\clustCoeff}$ (resp. $\bar{\closureCoeff}$) the average local clustering (resp. closure) coefficient over all nodes.}
\label{tab:datasets}
\scalebox{1}{
\begin{tabular}{lrrlll}
	\toprule
	\textbf{Dataset} & $n$ & $m$ & $d_{\max}$ & $\bar{\clustCoeff}$&$\bar{\closureCoeff}$\\
	\midrule
	fb-CMU & $6.6 \, K$ & $0.3 \, M$ & $8 \cdot 10^{2}$ & 0.27 & 0.12\\
	SP & $1.6 \, M$ & $22 \, M$ & $1 \cdot 10^{4}$ & 0.11 & 0.03\\
	FR & $12 \, M$ & $72 \, M$ & $3 \cdot 10^{3}$ & 0.08 & 0.01\\
	OR & $3.1 \, M$ & $0.1 \, B$ & $3 \cdot 10^{4}$ & 0.17 & 0.06\\
	LJ & $4.8 \, M$ & $43 \, M$ & $2 \cdot 10^{4}$ & 0.27 & 0.08\\
	BM & $43 \, K$ & $14 \, M$ & $8 \cdot 10^{3}$ & 0.51 & 0.19\\
	G500 & $4.6 \, M$ & $0.1 \, B$ & $3 \cdot 10^{5}$ & 0.06 & 0.0\\
	GP & $0.1 \, M $ & $12 \, M$ & $2 \cdot 10^{4}$ & 0.49 & 0.05\\
	PT & $43 \, K$ & $44 \, K$ & $2 \cdot 10^{1}$ & 0.12 & 0.1\\
	HW & $1.1 \, M$ & $56 \, M$ & $1 \cdot 10^{4}$ & 0.77 & 0.16\\
	HG & $0.5 \, M$ & $13 \, M$ & $5 \cdot 10^{4}$ & 0.19 & 0.01\\
	BNH & $0,7 \, M$ & $0.2 \, B$ & $2 \cdot 10^{4}$ & 0.5 & 0.3\\
	TW & $0.2 \, M$ & $6.8 \, M$ & $4 \cdot 10^{4}$ & 0.16 & 0.01\\
	\bottomrule
	\end{tabular}}
\end{table}

\begin{figure*}
	\addtolength{\tabcolsep}{-0.6em}
	\begin{tabular}{ll}
		\includegraphics[width=1.06\columnwidth]{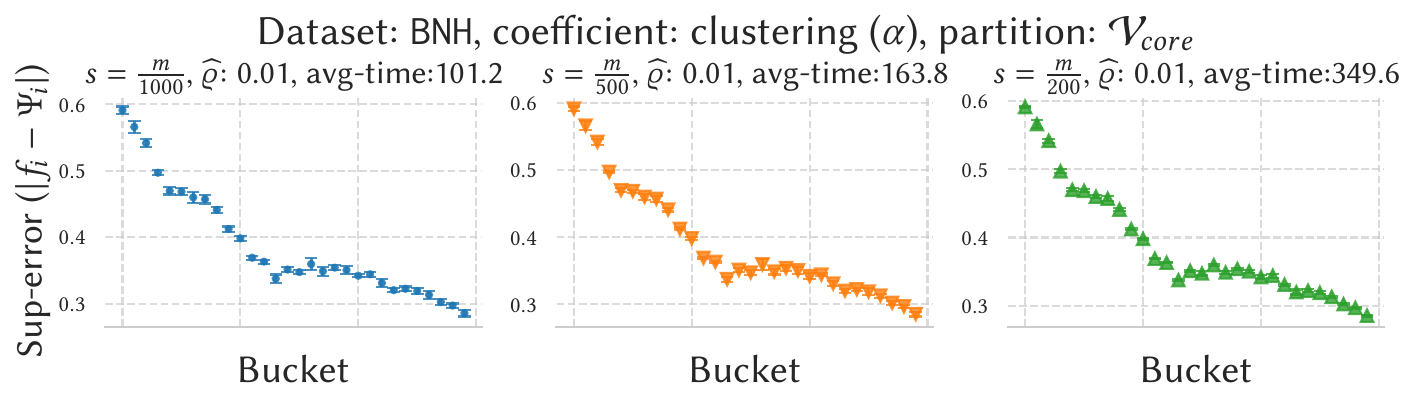}&
		\includegraphics[width=1.06\columnwidth]{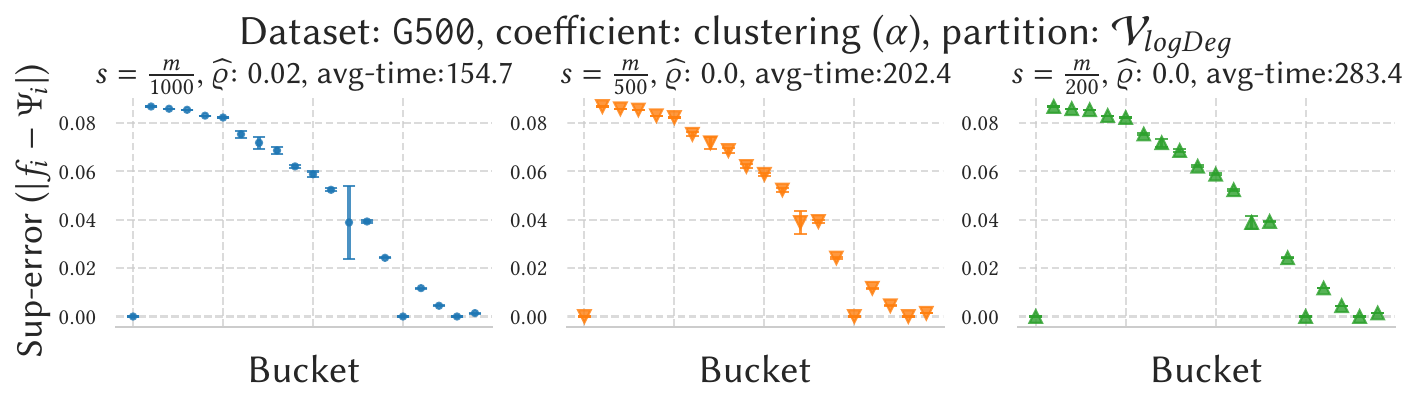}\\
		\includegraphics[width=1.06\columnwidth]{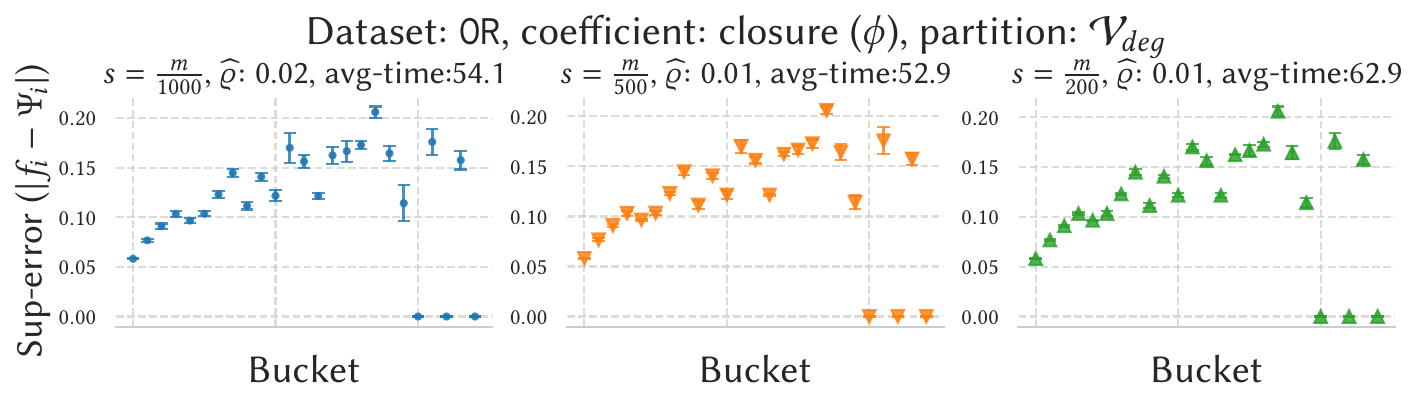}&
		\includegraphics[width=1.06\columnwidth]{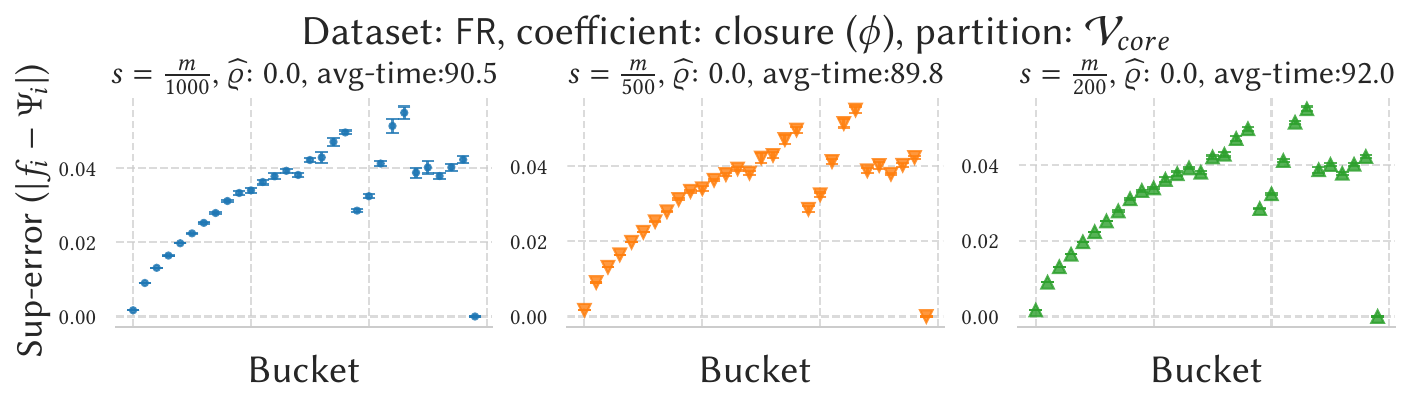}\\
	\end{tabular}
	\caption{Value $\clustGenPartition_i$ and its maximum error ($|f_i-\clustGenPartition_i|$) over five runs. We also report, the supremum error $\hat{\varrho}=\sup_{i\in[k]} |f_i-\clustGenPartition_i|$ and the average runtime over five independent runs across all buckets, for varying sample size ($s\in\{1,2,5\}$\textperthousand\ of the total edges $m$).}
	\label{fig:varyingSS}
\end{figure*}

\subsection{Setting}\label{exps:setting}

\noindent
\textbf{Implementation details.} We implemented our algorithms in C++20, and compiled it under gcc 9.4, 
with optimization flags. 
To solve the variance optimization problem (Section~\ref{subsec:varianceOpt}) we used Gurobi 11 under academic license. 
All the experiments were performed on a 72-core machine Intel Xeon Gold, running Ubuntu 20.04. 
The code to reproduce our results is publicly available.\footnote{\CODEURL.} 

\vspace{1mm}
\noindent 
\textbf{Datasets and partitions into buckets.} 
For our experiments we consider multiple datasets available online, from medium to large sized, which are reported together with a summary of their statistics in Table~\ref{tab:datasets}. 
\ifextended
More details on such datasets, are in Appendix~\ref{appsec:datasets}.
\else
Details on the datasets including URLs are available in our extended version~\cite{triadExt}.
\fi

For each dataset we considered three main different node partitions $\nodesetPartition$:
($i$)~$\mathcal{V}_{core}$ is obtained by grouping nodes with similar core-number over a total of $k=30$ buckets; 
($ii$)~$\mathcal{V}_{deg}$ is obtained by grouping together nodes with similar degrees over a total of $k=25$ buckets; 
($iii$)~$\mathcal{V}_{logDeg}$ assigns each node to a bucket as function of its degree~\cite{Kolda2014}, i.e.,
a node with degree $d$ is assigned to the bucket with index $\floor{\log(1+(d-2)/\log 2)} + 2$, hence it holds $k=\bigO(\log n)$. 

Note that all the above partition schemes place nodes with similar degree in the same bucket.
This is often the case in practical applications, 
where nodes with similar degree are associated to similar structural functions~\cite{Kolda2014}.

In addition, To test a general input to~\Cref{prob:partionEstimationAdditive} we also consider two other partitions $\nodesetPartition$. Partition $\mathcal{V}_{\mathit{rnd}}$ assigns nodes at random into $k=30$ buckets and $\nodesetPartition_{\mathit{met}}$ is obtained by clustering with $k=10$ each graph using METIS~\cite{Karypis1998metis}.
We report the results obtained on these partitions in
\ifextended
\Cref{appsec:missingSS}
\else
our extended version
\fi
as they follow similar trends to the ones discussed below.
We do not discuss the memory usage as it is similar to all algorithms (\mainAlg uses slightly more space compared to baselines as from our analysis in~\Cref{subsec:timeandmem}).

Finally we use \mainAlg-$\clustCoeff$ (resp.\  \mainAlg-$\closureCoeff$) to denote \mainAlg when used to approximate the average local clustering (resp.\ local closure) coefficient. 

\vspace{1mm}
\noindent 
\textbf{Research questions.} Our experimental evaluation investigated the following research questions.

\vspace{1mm}
\noindent 
\textbf{Q1.} 
How~\mainAlg performs in terms of accuracy and efficiency when varying 
its sample size $s$? %Is the algorithm efficient? 
(Section~\ref{exps:accuracyAndEff})

\vspace{1mm}
\noindent 
\textbf{Q2.} 
How tight are the adaptive guarantees provided by~\mainAlg, 
compared to state-of-the-art approaches? 
(Section~\ref{exps:sota})

\vspace{1mm}
\noindent 
\textbf{Q3.} 
What is the runtime of~\mainAlg; 
what is the impact of parameter~$q$, and 
what is the quality of our optimization of $q$? 
(Section~\ref{exps:rtAndParams})

\vspace{1mm}
\noindent 
\textbf{Q4.} 
Which patterns are captured by triadic coefficients and~\mainAlg over collaboration networks?
(Section~\ref{subsec:caseStudy})

\subsection{Accuracy of estimates and efficiency}\label{exps:accuracyAndEff}

In this section we answer \textbf{Q1}, i.e., we study~\mainAlg's accuracy and efficiency by varying the sample size $s$. 
This setting is fundamental
to show that~\mainAlg 
is both efficient and provides extremely accurate estimates by processing a small number of edges.

\vspace{1mm}
\noindent 
\textbf{Setting.} We consider~\algFixedSS, which retains the adaptive selection of the parameter $q$ 
while terminating~\mainAlg's main loop after processing exactly $s$ samples (see Section~\ref{subsec:algFixed}). 
To set $s$ we considered three different values: $s=1$\textperthousand, $s=2$\textperthousand, and $s=5$\textperthousand\ of the total edges $m$ of each dataset. 
We then run each configuration (dataset, value of $s$, and bucket partition) for five independent runs. 
We then measure for each bucket the supremum error $|f_i-\clustGenPartition_i|$ over the five runs, and $\widehat{\varrho}$ the supremum of such errors across all buckets of~$\nodesetPartition$. In addition, we measure the average runtime to process $s$ samples over the five runs.
Some representative results are presented in Figure~\ref{fig:varyingSS}.

\vspace{1mm}
\noindent 
\textbf{Results.} 
First, over almost all configurations tested we note that \mainAlg's estimates 
are very accurate and tightly concentrated for each bucket of the various partitions.
This is reflected by the supremum error $\widehat{\varrho}$, which is small and almost negligible even for very small sample sizes $s$ such $s$=1\textperthousand.  
This holds in particular for datasets \texttt{BNH} and \texttt{FR}, 
while~\mainAlg requires a slightly higher sample size (i.e., $s$=2\textperthousand) to provide extremely accurate estimates for datasets \texttt{G500} and \texttt{OR}. 
Note that the supremum error with a sample size of $s$=2\textperthousand$\cdot m$ tends to 0 on the considered configurations on all datasets. 
This supports the fact that \mainAlg requires only a very small number of samples to obtain highly accurate estimates for Problem~\ref{prob:partionEstimationAdditive}.

In addition, 
\mainAlg's runtime is limited by at most a few hundred of seconds on very large datasets, 
yielding estimates almost comparable to the exact unknown values,
showing that~\mainAlg is both efficient and highly accurate on both triadic coefficients.
\ifextended
Further results are presented in Appendix~\ref{appsec:missingSS}
\else
We report additional results under this setting in our extended version.
\fi

\vspace{1mm}
\noindent 
\textbf{Summary.} A very small sample size (of 1\textperthousand\ total edges) is often sufficient to obtain highly accurate estimates for~\mainAlg, \emph{simultaneously} over all buckets and different partitions for both the average local triadic coefficients, which is remarkable and extremely useful for highly-scalable network analysis.

\subsection{Comparison with state-of-the-art}\label{exps:sota}
\begin{figure*}
	\includegraphics[width=1.2\columnwidth]{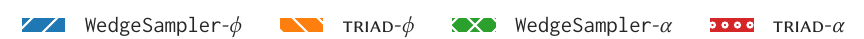}
	\addtolength{\tabcolsep}{-0.6em}
	\begin{tabular}{llllll}
	\includegraphics[width=0.35\columnwidth]{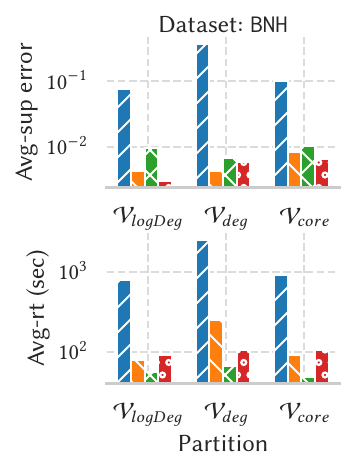} & \includegraphics[width=0.35\columnwidth]{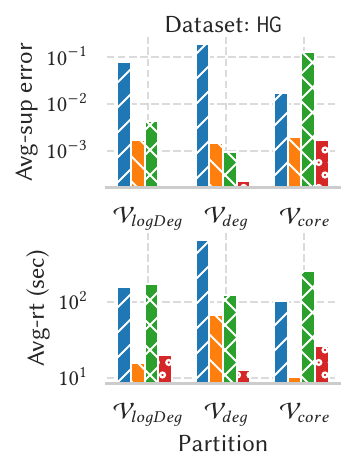} & \includegraphics[width=0.35\columnwidth]{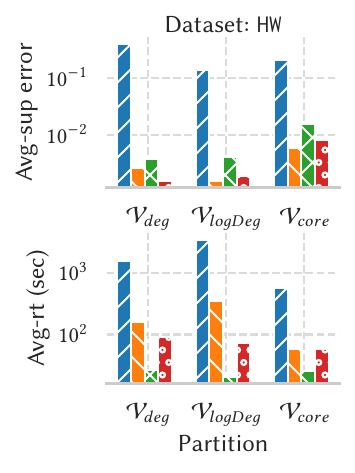} &
	\includegraphics[width=0.35\columnwidth]{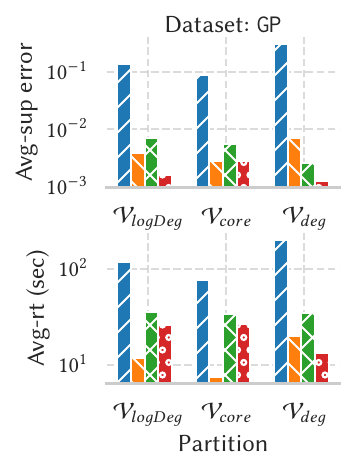}
	\includegraphics[width=0.35\columnwidth]{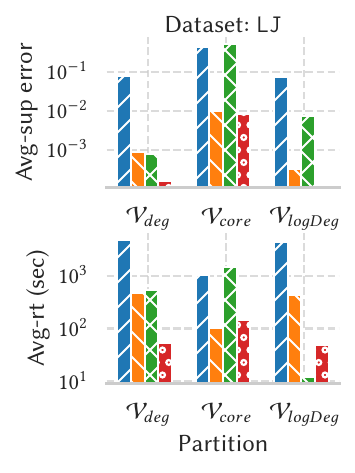} & \includegraphics[width=0.35\columnwidth]{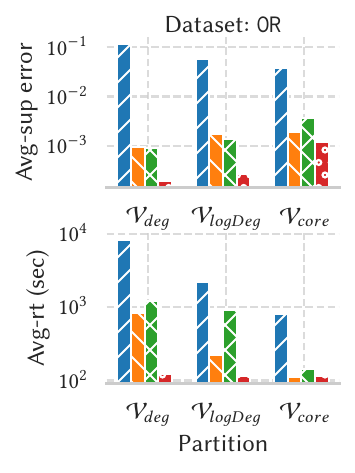}
	\end{tabular}
	\caption{Comparison of~\mainAlg and the baselines~\baseFull. For each dataset we show, (top plot): the average supremum error over all buckets over the various runs. (bottom): average runtime to perform an execution.}
	\label{fig:expsSota}
\end{figure*}

In this section we address \textbf{Q2}, i.e., we evaluate~\mainAlg and its adaptive guarantees with respect to existing state-of-the-art approaches.

\vspace{1mm}
\noindent 
\textbf{Setting.} We consider the state-of-the-art approach to approximate the local clustering coefficient values~\cite{Seshadhri2014} (see Section~\ref{sec:relworks}), denoted with~\baseFull-$\clustCoeff$ (or \texttt{WS}-$\clustCoeff$ for short).
We extend the idea of wedge sampling to approximate the local closure coefficient, as there are no algorithms tailored for the local closure coefficient. 
This baseline, denoted with~\baseFull-$\closureCoeff$, is detailed in
\ifextended
Appendix~\ref{appsec:baselines}.
\else
our extended version~\cite{triadExt}.
\fi
We fix $\varepsilon_j=\varepsilon=0.075$ for all datasets and all buckets in~$\nodesetPartition$ 
for~\mainAlg and $\eta=0.01$, additional parameters are reported in
\ifextended
Appendix~\ref{appsec:params}.
\else
our extended version.
\fi
For each configuration we run~\mainAlg and obtain $\widehat{\varepsilon}_j$, i.e.,
the adaptive upper bounds on the distance between the estimates $f_j$ 
and the unknown values $\clustGenPartition_j$ for each bucket $V_j\in \nodesetPartition$. 
We then use such values as input for~\baseFull, such that both algorithms provide the same guarantees.
For each configuration we compute the estimation error as the supremum error $|f_j-\clustGenPartition_j|$, averaged over all buckets, 
our results will show the maximum of such supremum error over five runs.
In addition, we report the average runtime for each algorithm on the various configurations, which was time-limited for all algorithms. 

\begin{figure}
	\centering
	\includegraphics[width=\columnwidth]{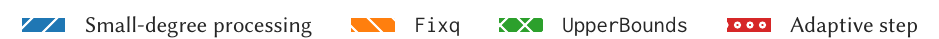}
	\addtolength{\tabcolsep}{-0.5em}
	\begin{tabular}{lr}
		\includegraphics[width=0.49\columnwidth]{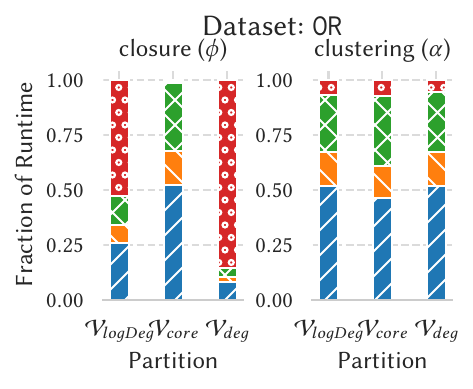} & \includegraphics[width=0.49\columnwidth]{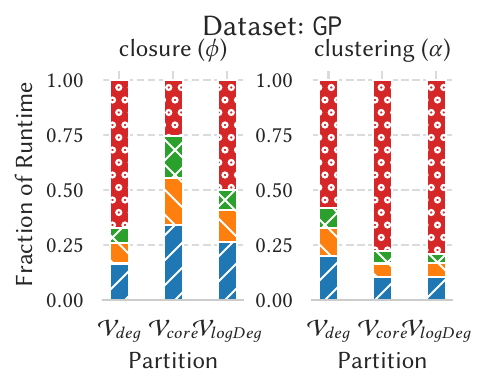}
	\end{tabular}	
	\caption{Fine grained runtime analysis. We show the average fraction of time spent in each step by~\mainAlg, the setting is from Section~\ref{exps:sota}.}
	\label{fig:runtimes}
\end{figure}

\begin{figure}
	\centering
	\includegraphics[width=0.95\columnwidth]{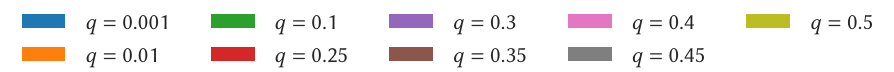}
	\addtolength{\tabcolsep}{-0.5em}
	\begin{tabular}{lr}
		\includegraphics[width=0.47\columnwidth]{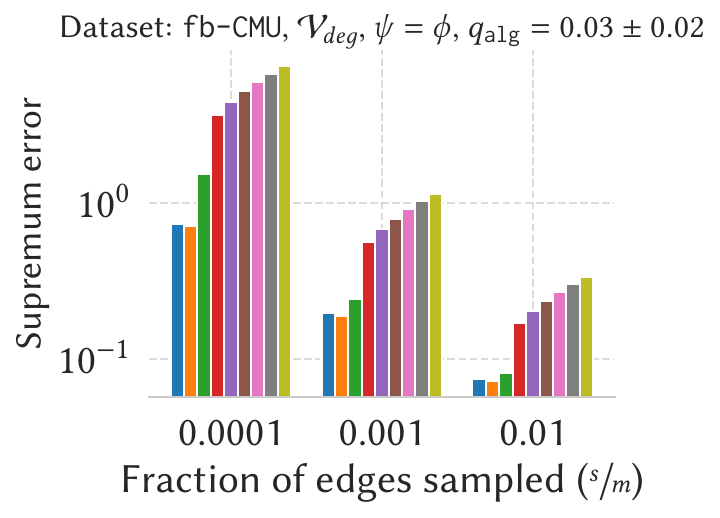} 
		& \includegraphics[width=0.47\columnwidth]{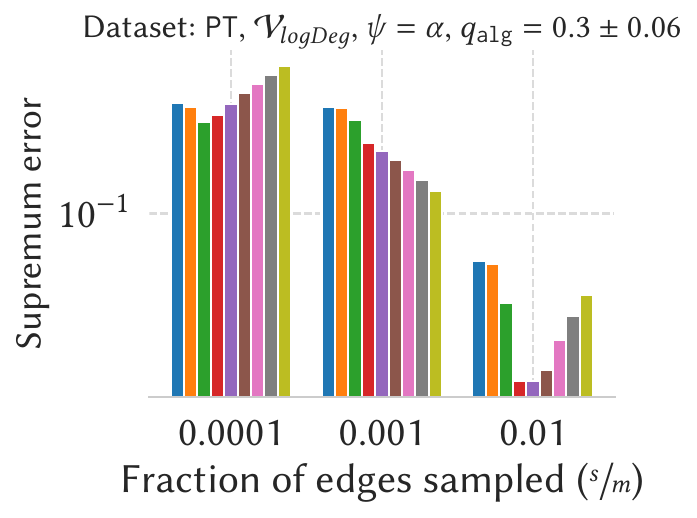}\\
	\end{tabular}	
	\caption{Supremum error over ten runs over different partitions, sample size $s$, coefficient $\clustMetric$, and values of $q$.}
	\label{fig:qvalues}
\end{figure}

\vspace{1mm}
\noindent 
\textbf{Results.} Key results are summarized in Figure~\ref{fig:expsSota}. 
We first observe that~\mainAlg reports very accurate estimates for $\clustGenPartition_j$ on most configurations, 
which are much more precise than the ones provided by~\baseFull. 
In particular, the supremum error, as desired, is of the order of $10^{-2}$ and on some configurations up to $10^{-3}$ 
(e.g.,~for datasets \texttt{HG} and \texttt{OR}). 
The baselines on most configurations achieve higher errors than~\mainAlg.
We observe that this is especially the case for~\baseFull-$\closureCoeff$ achieving higher error than~\mainAlg.
Importantly, 
we observe that the range of improvement in accuracy over the baseline is up to one order of magnitude. 

Remarkably, such results are obtained with a comparable or significantly smaller runtime with respect to the state-of-the-art baseline~\baseFull (up to one order of magnitude on datasets \texttt{HG}, \texttt{LJ} and \texttt{OR}).
In fact, our experiments confirm that~\mainAlg provides tighter bounds on the deviation between its estimates ($f_j$) and the unknown values ($\clustGenPartition_j$), 
significantly better than existing approaches, while being more efficient. 
Unfortunately, in practice, such bounds may still be loose---this can noted by observing that we provided in input to~\mainAlg $\varepsilon_j=0.075$ and under various datasets (e.g., \texttt{HG},  \texttt{OR} and  \texttt{LJ}) the maximum errors for~\mainAlg are much smaller. 

\vspace{1mm}
\noindent 
\textbf{Summary.} Our algorithm~\mainAlg provides better bounds on the deviation 
between the estimates $\estFunc_j$ and the unknown coefficients~$\clustGenPartition_j$, 
compared to existing approaches. 
While being tighter, such guarantees may still be loose for some settings, 
leaving an open question for future directions.

\subsection{Runtime and parameter sensitivity}\label{exps:rtAndParams}
\subsubsection{Runtime analysis}\label{exps:runtime}
In this section we analyze \mainAlg's runtime. 
In particular, we split the runtime into the following steps: 
($i$)~the practical optimization over the small-degree nodes; 
($ii$)~the optimization of the variance through \texttt{Fixq}; 
($iii$)~the execution of the routine~\texttt{UpperBounds}; and 
($iv$)~the adaptive loop. 

Figure~\ref{fig:runtimes} reports the average fraction of time spent by~\mainAlg in the various steps over the experiments from Section~\ref{exps:sota}. 
We report two very different behaviors.
On dataset \texttt{OR}, except for $\nodesetPartition_{\text{deg}}$, 
the time of performing the adaptive loop is negligible compared with all the other steps. 
This is in contrast with dataset~\texttt{GP}, where most of~\mainAlg's runtime is spent in its adaptive loop, 
highlighting that~\mainAlg effectively adapts to the complexity of the graph in input.
In other words, 
when the variance of the coefficients $\clustMetric_v$ 
is small, 
then the adaptive loop can terminate by processing a small amount of samples $s$ as captured by~\Cref{theo:predictablePlugin}.

Interestingly, we note that the procedure to optimize the variance (that we denote with~\texttt{Fixq}) is often negligible, especially compared with~\texttt{UpperBounds}, 
as captured by our analysis in Section~\ref{subsec:varianceOpt}.

Our results show that~\mainAlg's runtime depends on the complexity of the input graph i.e., the distribution of the unknown coefficients across buckets in $\nodesetPartition$, 
which allows~\mainAlg to compute highly accurate estimates very efficiently through its adaptive bounds.

\subsubsection{Assessing the impact of $q$}\label{exps:params:qval}
In this section we investigate the impact of the parameter $q$ on the quality of the estimates computed by~\mainAlg. 
Recall that $q$ controls how the estimates $f_j$ are computed, affecting the variance of the results.
To evaluate such parameter,
we selected two of our smallest datasets, 
for computational efficiency, 
and a fixed grid of nine values for the parameter $q$. 
For each value of $q$ we then tested different sample sizes, 
i.e., using~\algFixedSS with a sample size $s$ such that $s/m \in\{0.0001,0.001,0.01\}$. 
For each combination of dataset, partition $\mathcal{V}$, value of $q$, 
and sample size $s$, 
we performed ten runs over both triadic coefficients $\clustMetric\in\{\closureCoeff,\clustCoeff\}$. 

Some representative results are shown in Figure~\ref{fig:qvalues} 
reporting the supremum error (i.e.,~$\sup_j |\estFunc_j-\clustGenPartition_j|$) over the various configurations. 
We observe that the impact of $q$ on the estimates $f_j$, 
over all sample sizes, can be from negligible (on the bottom-left plot) to very significant (bottom-right plot). 
In general, we observe a significant reduction in the supremum error by a proper selection 
of the value of $q$---up to one order of magnitude 
in several settings (such as the top-left or bottom-right plots). 
This behavior confirms the importance of properly selecting the value of $q$. 
We note that the supremum error tends to be minimized with a specific value of $q$ 
over the different settings,
but this value is, in general, different on most configurations (e.g.,~top- and bottom-right~plots). 

Summarizing, a proper selection of the value of $q$ can have a significant impact on the estimates of~\mainAlg leading to up to one order of accuracy in the estimates.
The next section assesses how well our optimization aligns with a good value of the parameter $q$.

\subsubsection{Optimization of $q$}
We now briefly assess how well the value of $q_{\text{alg}}$ (as optimized by~\mainAlg)
aligns with a good choice for the parameter $q$. 
Results are shown in Figure~\ref{fig:qvalues}, 
where we report the average $q_{\text{alg}}$ and its standard deviation over ten runs.
The size of the bag of samples used to compute $q_{\text{alg}}$ is set to 500.

As captured by our analysis, the value of $q_{\text{alg}}$ well-aligns with a good value of $q$ obtained from the grid of tested values, on each configuration. 
In fact, $q_{\text{alg}}$ often overlaps with the $q$ yielding the minimum supremum error from the grid. For example, on the top-left plot the best value on the grid is $q=0.01$ 
while $q_{\text{alg}}  = 0.03\pm0.02$, 
highlighting that our method properly selects a good value for $q$, yielding small estimation variance. 

\subsection{Case study---academic collaborations}\label{subsec:caseStudy}
\begin{figure}
	\includegraphics[width=0.9\columnwidth]{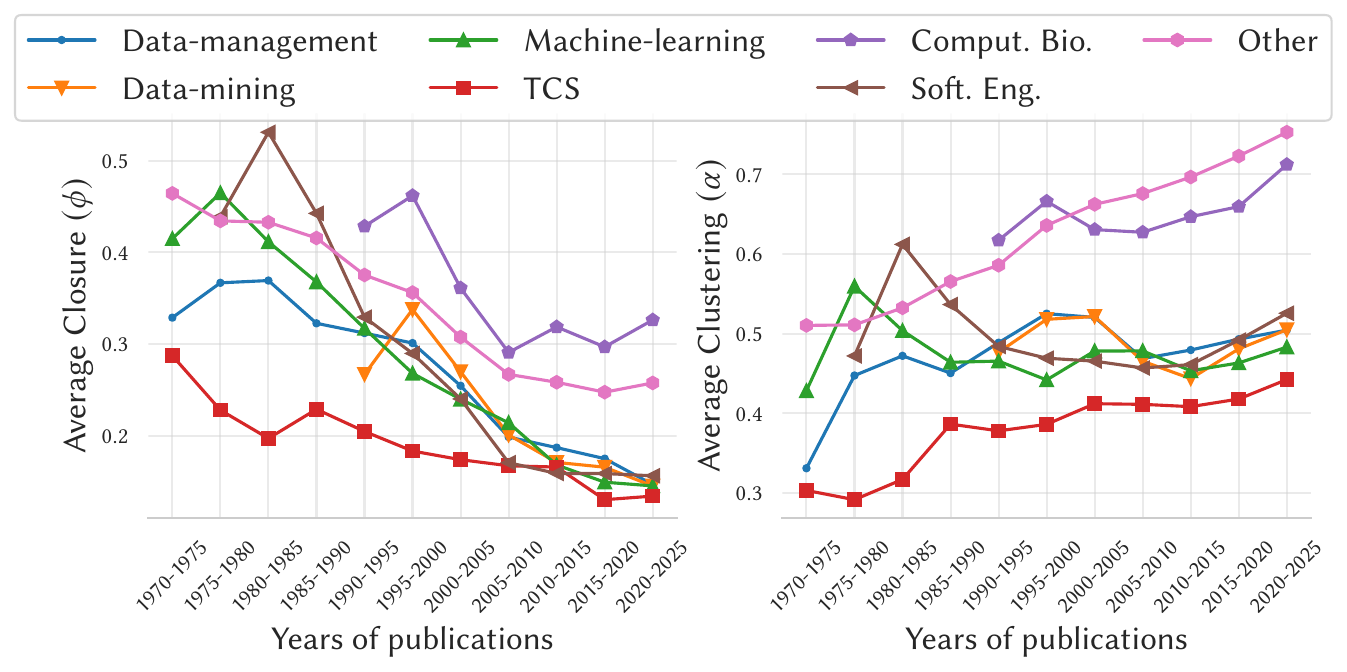}
	\caption{Average local clustering and local closure coefficient over the DBLP graph snapshots, for various computer science communities.}
	\label{fig:DBLPTemporal}
\end{figure}
\begin{figure}
	\begin{tabular}{cc}
		\includegraphics[width=0.49\columnwidth]{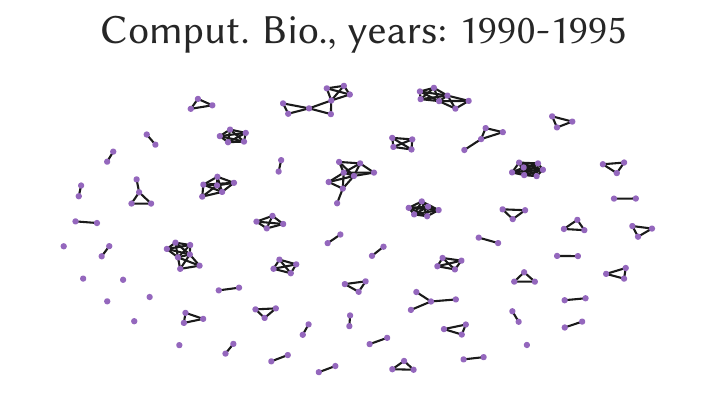}&
		\includegraphics[width=0.49\columnwidth]{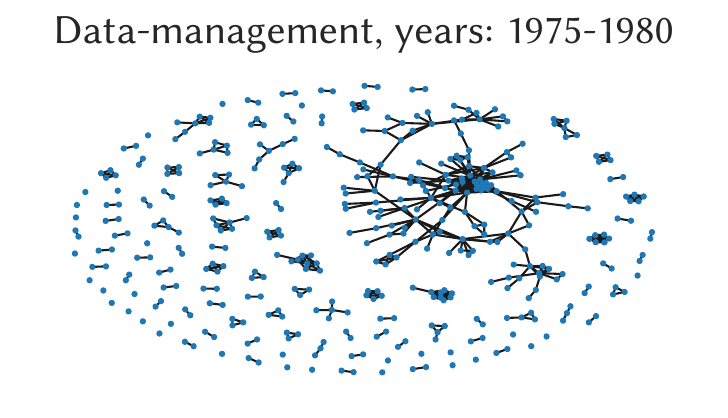}\\
	\end{tabular}
	\caption{Induced subgraphs by communities. (Left): computational biology; (Right): data-management. 
	}
	\label{fig:clusts}
\end{figure}

In this section we analyze collaboration patterns in different communities of the DBLP network using~\mainAlg, 
answering \textbf{Q4}.

\vspace{1mm}
\noindent 
\textbf{Setting.} 
DBLP collects bibliographic information about all major computer science journals and proceedings publications.\footnote{\url{https://dblp.org/}} 
For each time-period of five consecutive years from 1970 to 2024,
we collected the respective set of publications over DBLP. 
On each time period $t_1=[1970,1975], t_2=[1975,1980],\dots$ 
we build a graph $G_{t_i} = (V_{t_i}, E_{t_i})$, with 
$V_{t_i}$ consisting of authors, 
and edges corresponding to authors sharing a common publication. 
We then classified the authors, on each graph $G_{t_i}$, according to their research community. 
The resulting categories are reported in the legend of Figure~\ref{fig:DBLPTemporal}, %where also report the results,
additional details on the classification are in 
\ifextended
Appendix~\ref{appsubsec:DBLP}. 
\else
in our extended version~\cite{triadExt}.
\fi

For each graph we computed the average triadic coefficients over each category. We then investigated if the analysis of the coefficients $f_j$, for $j\in[k]$, provides us insights into similarities and differences of collaboration patterns over different communities.

\vspace{1mm}
\noindent 
\textbf{Results.} 
We observe in Figure~\ref{fig:DBLPTemporal} some interesting trends.
For most of the communities,
the average local clustering coefficient increases 
or remains stable over the years.
Instead, the average local closure coefficient mostly decreases, 
for all but the computational-biology community. 
This can be explained by the fact that, new nodes over the network, 
are likely to have a small local closure coefficient when they collaborate with an author having many coauthors (e.g., students publishing with their advisor). 
Instead, it is easier for novel nodes on the network to a have higher local clustering coefficient, 
e.g., by collaborations within research groups. 
In addition, it is also easier for authors already belonging to the network to increase their local clustering coefficient over time, 
e.g., by publishing more. 

There is a, perhaps surprising, increase in the average local closure coefficient 
for the computational-biology community over time. 
This could be explained by the fact that publications in this area 
require the joint effort of many authors, 
which likely increases the density of the connections of the authors in the graphs. 
Such an aspect can be observed in~\Cref{fig:clusts}, 
where we see that the structure of collaborations over the computational biology category forms many cliques. 
On the other hand, the structure of the graph of the data management community is more sparse, 
containing some chain structures.
We further visualize the subgraph of other categories
\ifextended
in~\Cref{appsec:community}.
\else
in our extended manuscript~\cite{triadExt}.
\fi
We conclude by noting that the average local clustering coefficient is significantly higher 
(ranging from 0.4 to more than 0.6) 
than the average local closure coefficient (which does not exceed 0.35) over all communities.
This may be related to  the nature of academic collaborations, 
where publishing with established researchers decreases the local closure coefficient of novel researchers on average.
We further discuss that the above results cannot be uniquely explained by the degree distribution of the various nodes, in
\ifextended
\Cref{appsec:community}.
\else
our extended version~\cite{triadExt}.
\fi

\vspace{1mm}
\noindent 
\textbf{Summary.} We analyzed the average local clustering (and closure) coefficients over 
different computer science communities across time. 
We observed that the values of the triadic coefficient can capture different collaboration patterns. For example, capturing highly collaborative (computational biology) and more sparse (data management) communities.
Our findings show an example of a simple analysis using Problem~\ref{prob:partionEstimationAdditive} 
to gain better insights into publication and collaboration patterns in different research communities.

\begin{table}
	\caption{Comparison of~\mainAlg and existing state-of-the-art approaches. $\clustGenPartition$: 
	if the algorithm can estimate $\clustCoeff$, $\closureCoeff$, or both. 
	``Adaptive'' denotes if the estimates adapt to the partitions~$\nodesetPartition$. 
	``Number of samples'' denotes the number of samples processed. 
	Finally ``Processing complexity'' denotes the time to process a sample, 
	assuming $\bigO(1)$ time complexity to check the existence of an edge. For \texttt{ThinkD}, $r$ denotes the number of retained edges, which depends on $\varepsilon^{-2}$.}\label{tab:sotaComp}
	\vspace{-2mm}
	\scalebox{0.72}{
	%\edit{
	\begin{tabular}{lcccc}
		\toprule
		\multirow{2}*{Algorithm} & 
		\multirow{2}*{$\clustGenPartition$} &  
		\multirow{2}*{Adaptive} & 
		{Number of} & Processing\\ 
		 & &  & samples & complexity\\ 
		\midrule
		{\mainAlg} & 
		{$\clustCoeff,\closureCoeff$} & 
		{\vmark} & 
	  $\Omega\left(\tfrac{R\log k/\eta}{\varepsilon}\right)$ and $\bigO\left(R^2\tfrac{(\zeta+\log 1/\eta)}{\varepsilon^2}\right)$ & 
		{$\bigO(d_{\max})$} \\
		 \midrule
		 \texttt{WS}-$\clustCoeff$~\cite{Seshadhri2014} & $\clustCoeff$ & \xmark & $\Theta\left( k \varepsilon^{-2} \log k/\eta\right)$& $\bigO(1)$\\
		 \texttt{WS}-$\closureCoeff$ & $\closureCoeff$ & \xmark & $\Theta\left( k \varepsilon^{-2} \log k/\eta\right)$& $\bigO(d_{\max})$\\
		 \texttt{LCE}~\cite{Lima2022ClusteringLatin} & $\clustCoeff$ & \xmark & $\Theta\left(m^2 \varepsilon^{-2} (\log d_{\max} + \log 1/\eta)\right)$ &  $\bigO(d_{\max})$\\
		 \midrule
		 \texttt{ECC}~\cite{Kutzkov2013StreamCC} & $\clustCoeff$ & \xmark & $\bigO(\varepsilon^{-2} \log n/\eta)$ & $\bigO(m\log\varepsilon^{-1})$\\
		 \texttt{ThinkD}~\cite{Shin2020Stream} & $\clustCoeff$ & \xmark & $\bigO(r)$ & $\Theta(m)$\\
		\bottomrule	\end{tabular}}
\end{table}

\section{Related Work}
\label{sec:relworks}

The problems addressed in this paper are closely related to counting triangles in graphs, 
which has been studied extensively. 
Thus, an extensive review is outside the scope~\cite{Seshadhri2019Tutorial, AlHasan2017TriCnt}, 
instead, we focus only on discussing the most relevant problem settings and~techniques.

\vspace{1mm}
\noindent 
\textbf{Clustering and closure coefficient algorithms.}
The local clustering coefficient was first introduced by \citet{Watts1998}. 
Since their seminal work, many algorithms have been developed to efficiently compute related graph statistics. 
Many approaches consider both the approximation of the global clustering coefficient~\cite{Schank2005CC,BuriolClusteringIndexes}, which is the local clustering coefficient averaged over all nodes in the graph, 
and the transitivity coefficient, 
which is the fraction of closed triangles over all the wedges in the graph~\cite{Etemadi2017BiasCorrection}, or 
weighted versions~\cite{Lattanzi2016WeightedCC}. 
Interestingly, \mainAlg can be adapted to compute all those coefficients with minimal modifications.

Many algorithms have been designed for computing the local clustering coefficient in restricted access models, 
such as
($i$)~when the graph can be explored only through random walks~\cite{Hardiman2013EstimatingCCRW};
($ii$)~when the graph is accessed in a (semi-)streaming fashion~\cite{Becchetti2010,Kutzkov2013StreamCC}; or
($iii$)~distributed environments~\cite{Kolda2014}. 
Given that these works focus on restrictive scenarios, 
they require a large number of samples and often do not offer accurate guarantees.

The works most related to our formulation are by~\citet{Etemadi2017BiasCorrection}, 
and \citet{Seshadhri2014} who developed wedge-sampling algorithms, 
which, for each partition,  sample wedges 
(i.e., our baseline considered in Section~\ref{exps:sota}). 
These algorithms require a significantly high sample size, i.e., $\Theta(k\varepsilon^{-2} \log(k/\eta))$, 
which is tight~\citep{Baigneres2004LB}. 
Therefore these algorithms become impractical for small values of $\varepsilon$, 
as also demonstrated in our experiments. 

Recently, 
\citet{Lima2022ClusteringLatin} developed an algorithm based on sampling edges 
and collecting their incident triangles, to approximate the local clustering coefficient of nodes with high-degree.
The authors prove an upper bound on the sample complexity using VC-dimension. 
While their approach is similar in spirit to ours, 
it is significantly less general, 
i.e., their approach can be obtained by our class of estimators setting $q=0$ in Equation~\eqref{eq:EstimatorQ}. 
Their approach is also significantly less efficient, 
as their algorithm is based on data-independent bounds. 
In addition, our bound, as captured by Corollary~\ref{coroll:pdimLima},  
is significantly tighter than theirs.

Surprisingly, 
not much work has been done on algorithms for the closure coefficient. 
Recent works only quantify how this coefficient evolves in random networks~\cite{Yuan2024CLTClosure}. 

\vspace{1mm}
\noindent 
\textbf{Triangle-counting algorithms.}
As already noted, triangle counting is a wide area of research 
\cite{Bader2023FastCounting,Eden2017Tri,Kolountzakis2012,Li2024FastLocalCnt,Tsourakakis2009}.
Most existing works address the problem of computing \emph{global} triangle counts, 
and cannot therefore be used in our setting. 
Several works have been instead developed for \emph{local} triangle counting. 
Exact methods~\cite{Pashanasangi2020EVOKE,Li2024FastLocalCnt} are prohibitive for large networks, 
and sampling methods are designed for streaming settings~\cite{Stefani2017Triest,Ahmed2017Streams,Shin2017WRS, Shin2020Stream}. 
Those algorithms can provide accurate estimates for the local clustering coefficient 
when nodes have very high degree, 
while they are highly inaccurate for nodes with small degree. 
Note that this is a significant limitation for the setting we consider in this paper, 
where partitions may contain a large number of nodes with small degree. 

\vspace{1mm}
\noindent 
\textbf{Subgraph-counting algorithms.}
Another related problem is the one of counting subgraph occurrences, 
for which many different methods have been proposed~\cite{Li2024FastLocalCnt,Rahman2014Graphlets,Jha2015, Rossi2019,Rossi2020,Bressan2019Motivo}. 
While these algorithms can be effectively used to count subgraph occurrences with respect to specific or multiple subgraph patterns, they cannot be easily adapted to extract high-quality local subgraph counts and average local triadic coefficients, as considered in Problem~\ref{prob:partionEstimationAdditive}. 
Finally, \citet{Ahmed2016} develop methods to estimate local subgraph counts, 
but compute \emph{exactly} all triangles in the graph. 
A summary of the key differences with most related works is reported in~\Cref{tab:sotaComp}. 
Note that the last two algorithms are for streaming settings~\cite{Kutzkov2013StreamCC, Shin2020Stream}, 
hence we did not consider them in~\Cref{sec:exps}, as they are designed for a more restrictive data-access model,
yielding more inefficient methods. We observe from~\Cref{tab:sotaComp} that depending on the evaluation of the adaptive bounds on the given datasets~\mainAlg can be highly efficient (i.e., when $R$ is small), improving significantly over existing methods as shown in~\Cref{sec:exps}.

\section{Conclusion}
We studied the problem of efficiently computing the average of local triadic coefficients. 
We designed~\mainAlg, an efficient and adaptive sampling algorithm. 
\mainAlg estimates both the average local clustering coefficient and the recently-introduced 
average local closure coefficient, 
for which no algorithmic techniques were previously known. 
We showed that~\mainAlg is efficient and reports extremely accurate estimates, 
especially compared with existing methods.

There are several interesting directions for future work, 
such considering~\mainAlg for averages of local coefficients, 
which depend on the given partitions 
(e.g., a triangle is weighted differently if it contains nodes from different buckets), and 
weighted variants of the clustering and closure coefficients~\cite{Lattanzi2016WeightedCC}. 
Another interesting direction is to adapt~\mainAlg for a multi-pass streaming setting~\cite{Becchetti2010,Kutzkov2013StreamCC}. 
Finally, it will be interesting to study whether it is possible to design tighter bounds on the sample complexity, 
e.g., based on Rademacher complexity~\cite{Pellegrina2023,Pellegrina2023Silvan}.

%\iffalse
\begin{acks}
We thank Fabio Vandin for providing us the computing infrastructure. This research is supported by the
ERC Advanced Grant REBOUND (834862), 
the EC H2020 RIA project SoBigData++ (871042), and 
the Wallenberg AI, Autonomous Systems and Software Program (WASP) funded by the Knut and Alice Wallenberg Foundation.
\end{acks}
%\fi

%\clearpage
\balance
\bibliographystyle{ACM-Reference-Format}
\bibliography{ourbib}

%%% -*-BibTeX-*-
%%% Do NOT edit. File created by BibTeX with style
%%% ACM-Reference-Format-Journals [18-Jan-2012].

\begin{thebibliography}{71}

%%% ====================================================================
%%% NOTE TO THE USER: you can override these defaults by providing
%%% customized versions of any of these macros before the \bibliography
%%% command.  Each of them MUST provide its own final punctuation,
%%% except for \shownote{}, \showDOI{}, and \showURL{}.  The latter two
%%% do not use final punctuation, in order to avoid confusing it with
%%% the Web address.
%%%
%%% To suppress output of a particular field, define its macro to expand
%%% to an empty string, or better, \unskip, like this:
%%%
%%% \newcommand{\showDOI}[1]{\unskip}   % LaTeX syntax
%%%
%%% \def \showDOI #1{\unskip}           % plain TeX syntax
%%%
%%% ====================================================================

\ifx \showCODEN    \undefined \def \showCODEN     #1{\unskip}     \fi
\ifx \showDOI      \undefined \def \showDOI       #1{#1}\fi
\ifx \showISBNx    \undefined \def \showISBNx     #1{\unskip}     \fi
\ifx \showISBNxiii \undefined \def \showISBNxiii  #1{\unskip}     \fi
\ifx \showISSN     \undefined \def \showISSN      #1{\unskip}     \fi
\ifx \showLCCN     \undefined \def \showLCCN      #1{\unskip}     \fi
\ifx \shownote     \undefined \def \shownote      #1{#1}          \fi
\ifx \showarticletitle \undefined \def \showarticletitle #1{#1}   \fi
\ifx \showURL      \undefined \def \showURL       {\relax}        \fi
% The following commands are used for tagged output and should be
% invisible to TeX
\providecommand\bibfield[2]{#2}
\providecommand\bibinfo[2]{#2}
\providecommand\natexlab[1]{#1}
\providecommand\showeprint[2][]{arXiv:#2}

\bibitem[\protect\citeauthoryear{Ahmed, Duffield, Willke, and Rossi}{Ahmed
  et~al\mbox{.}}{2017}]%
        {Ahmed2017Streams}
\bibfield{author}{\bibinfo{person}{Nesreen~K. Ahmed}, \bibinfo{person}{Nick
  Duffield}, \bibinfo{person}{Theodore~L. Willke}, {and}
  \bibinfo{person}{Ryan~A. Rossi}.} \bibinfo{year}{2017}\natexlab{}.
\newblock \showarticletitle{On sampling from massive graph streams}.
\newblock \bibinfo{journal}{\emph{Proceedings of the VLDB Endowment}}
  \bibinfo{volume}{10}, \bibinfo{number}{11} (\bibinfo{date}{Aug.}
  \bibinfo{year}{2017}), \bibinfo{pages}{1430--1441}.
\newblock
\showISSN{2150-8097}
\urldef\tempurl%
\url{https://doi.org/10.14778/3137628.3137651}
\showDOI{\tempurl}


\bibitem[\protect\citeauthoryear{Ahmed, Willke, and Rossi}{Ahmed
  et~al\mbox{.}}{2016}]%
        {Ahmed2016}
\bibfield{author}{\bibinfo{person}{Nesreen~K. Ahmed},
  \bibinfo{person}{Theodore~L. Willke}, {and} \bibinfo{person}{Ryan~A. Rossi}.}
  \bibinfo{year}{2016}\natexlab{}.
\newblock \showarticletitle{Estimation of local subgraph counts}. In
  \bibinfo{booktitle}{\emph{2016 {IEEE} International Conference on Big Data
  (Big Data)}}. \bibinfo{publisher}{{IEEE}}.
\newblock
\urldef\tempurl%
\url{https://doi.org/10.1109/bigdata.2016.7840651}
\showDOI{\tempurl}


\bibitem[\protect\citeauthoryear{Al~Hasan and Dave}{Al~Hasan and Dave}{2017}]%
        {AlHasan2017TriCnt}
\bibfield{author}{\bibinfo{person}{Mohammad Al~Hasan} {and}
  \bibinfo{person}{Vachik~S. Dave}.} \bibinfo{year}{2017}\natexlab{}.
\newblock \showarticletitle{Triangle counting in large networks: a review}.
\newblock \bibinfo{journal}{\emph{WIREs Data Mining and Knowledge Discovery}}
  \bibinfo{volume}{8}, \bibinfo{number}{2} (\bibinfo{date}{Oct.}
  \bibinfo{year}{2017}).
\newblock
\showISSN{1942-4795}
\urldef\tempurl%
\url{https://doi.org/10.1002/widm.1226}
\showDOI{\tempurl}


\bibitem[\protect\citeauthoryear{Bader}{Bader}{2023}]%
        {Bader2023FastCounting}
\bibfield{author}{\bibinfo{person}{David~A. Bader}.}
  \bibinfo{year}{2023}\natexlab{}.
\newblock \showarticletitle{Fast Triangle Counting}. In
  \bibinfo{booktitle}{\emph{2023 IEEE High Performance Extreme Computing
  Conference (HPEC)}}. \bibinfo{publisher}{IEEE}, \bibinfo{pages}{1--6}.
\newblock
\urldef\tempurl%
\url{https://doi.org/10.1109/hpec58863.2023.10363539}
\showDOI{\tempurl}


\bibitem[\protect\citeauthoryear{Bader, Li, Du, Pauliuchenka, Rodriguez, Gupta,
  Minnal, Nahata, Ganeshan, Gundogdu, and Lew}{Bader et~al\mbox{.}}{2024}]%
        {Bader2024CoverTri}
\bibfield{author}{\bibinfo{person}{David~A. Bader}, \bibinfo{person}{Fuhuan
  Li}, \bibinfo{person}{Zhihui Du}, \bibinfo{person}{Palina Pauliuchenka},
  \bibinfo{person}{Oliver~Alvarado Rodriguez}, \bibinfo{person}{Anant Gupta},
  \bibinfo{person}{Sai Sri~Vastav Minnal}, \bibinfo{person}{Valmik Nahata},
  \bibinfo{person}{Anya Ganeshan}, \bibinfo{person}{Ahmet Gundogdu}, {and}
  \bibinfo{person}{Jason Lew}.} \bibinfo{year}{2024}\natexlab{}.
\newblock \bibinfo{title}{Cover Edge-Based Novel Triangle Counting}.
\newblock
\newblock
\urldef\tempurl%
\url{https://doi.org/10.48550/ARXIV.2403.02997}
\showDOI{\tempurl}


\bibitem[\protect\citeauthoryear{Baignères, Junod, and Vaudenay}{Baignères
  et~al\mbox{.}}{2004}]%
        {Baigneres2004LB}
\bibfield{author}{\bibinfo{person}{Thomas Baignères}, \bibinfo{person}{Pascal
  Junod}, {and} \bibinfo{person}{Serge Vaudenay}.}
  \bibinfo{year}{2004}\natexlab{}.
\newblock \bibinfo{booktitle}{\emph{How Far Can We Go Beyond Linear
  Cryptanalysis?}}
\newblock \bibinfo{publisher}{Springer Berlin Heidelberg},
  \bibinfo{pages}{432--450}.
\newblock
\showISBNx{9783540305392}
\showISSN{1611-3349}
\urldef\tempurl%
\url{https://doi.org/10.1007/978-3-540-30539-2_31}
\showDOI{\tempurl}


\bibitem[\protect\citeauthoryear{Becchetti, Boldi, Castillo, and
  Gionis}{Becchetti et~al\mbox{.}}{2010}]%
        {Becchetti2010}
\bibfield{author}{\bibinfo{person}{Luca Becchetti}, \bibinfo{person}{Paolo
  Boldi}, \bibinfo{person}{Carlos Castillo}, {and} \bibinfo{person}{Aristides
  Gionis}.} \bibinfo{year}{2010}\natexlab{}.
\newblock \showarticletitle{Efficient algorithms for large-scale local triangle
  counting}.
\newblock \bibinfo{journal}{\emph{{ACM} Transactions on Knowledge Discovery
  from Data}} \bibinfo{volume}{4}, \bibinfo{number}{3} (\bibinfo{date}{oct}
  \bibinfo{year}{2010}), \bibinfo{pages}{1--28}.
\newblock
\urldef\tempurl%
\url{https://doi.org/10.1145/1839490.1839494}
\showDOI{\tempurl}


\bibitem[\protect\citeauthoryear{Bhowmick and Seah}{Bhowmick and Seah}{2016}]%
        {Bhowmick2016ClusteringProteins}
\bibfield{author}{\bibinfo{person}{Sourav~S. Bhowmick} {and}
  \bibinfo{person}{Boon~Siew Seah}.} \bibinfo{year}{2016}\natexlab{}.
\newblock \showarticletitle{Clustering and Summarizing Protein-Protein
  Interaction Networks: A Survey}.
\newblock \bibinfo{journal}{\emph{IEEE Transactions on Knowledge and Data
  Engineering}} \bibinfo{volume}{28}, \bibinfo{number}{3}
  (\bibinfo{date}{March} \bibinfo{year}{2016}), \bibinfo{pages}{638--658}.
\newblock
\showISSN{2326-3865}
\urldef\tempurl%
\url{https://doi.org/10.1109/tkde.2015.2492559}
\showDOI{\tempurl}


\bibitem[\protect\citeauthoryear{Bianchi, Grattarola, and Alippi}{Bianchi
  et~al\mbox{.}}{2019}]%
        {Bianchi2019SpectralGNN}
\bibfield{author}{\bibinfo{person}{Filippo~Maria Bianchi},
  \bibinfo{person}{Daniele Grattarola}, {and} \bibinfo{person}{Cesare Alippi}.}
  \bibinfo{year}{2019}\natexlab{}.
\newblock \bibinfo{title}{Spectral Clustering with Graph Neural Networks for
  Graph Pooling}.
\newblock
\newblock
\urldef\tempurl%
\url{https://doi.org/10.48550/ARXIV.1907.00481}
\showDOI{\tempurl}


\bibitem[\protect\citeauthoryear{Borassi and Natale}{Borassi and
  Natale}{2019}]%
        {Borassi2019Kadabra}
\bibfield{author}{\bibinfo{person}{Michele Borassi} {and}
  \bibinfo{person}{Emanuele Natale}.} \bibinfo{year}{2019}\natexlab{}.
\newblock \showarticletitle{KADABRA is an ADaptive Algorithm for Betweenness
  via Random Approximation}.
\newblock \bibinfo{journal}{\emph{ACM Journal of Experimental Algorithmics}}
  \bibinfo{volume}{24} (\bibinfo{date}{Feb.} \bibinfo{year}{2019}),
  \bibinfo{pages}{1--35}.
\newblock
\showISSN{1084-6654}
\urldef\tempurl%
\url{https://doi.org/10.1145/3284359}
\showDOI{\tempurl}


\bibitem[\protect\citeauthoryear{Bressan, Leucci, and Panconesi}{Bressan
  et~al\mbox{.}}{2019}]%
        {Bressan2019Motivo}
\bibfield{author}{\bibinfo{person}{Marco Bressan}, \bibinfo{person}{Stefano
  Leucci}, {and} \bibinfo{person}{Alessandro Panconesi}.}
  \bibinfo{year}{2019}\natexlab{}.
\newblock \showarticletitle{Motivo: fast motif counting via succinct color
  coding and adaptive sampling}.
\newblock \bibinfo{journal}{\emph{Proceedings of the VLDB Endowment}}
  \bibinfo{volume}{12}, \bibinfo{number}{11} (\bibinfo{date}{July}
  \bibinfo{year}{2019}), \bibinfo{pages}{1651--1663}.
\newblock
\showISSN{2150-8097}
\urldef\tempurl%
\url{https://doi.org/10.14778/3342263.3342640}
\showDOI{\tempurl}


\bibitem[\protect\citeauthoryear{Buriol, Frahling, Leonardi, and Sohler}{Buriol
  et~al\mbox{.}}{[n.d.]}]%
        {BuriolClusteringIndexes}
\bibfield{author}{\bibinfo{person}{Luciana~S. Buriol}, \bibinfo{person}{Gereon
  Frahling}, \bibinfo{person}{Stefano Leonardi}, {and}
  \bibinfo{person}{Christian Sohler}.} \bibinfo{year}{[n.d.]}\natexlab{}.
\newblock \bibinfo{booktitle}{\emph{Estimating Clustering Indexes in Data
  Streams}}.
\newblock \bibinfo{publisher}{Springer Berlin Heidelberg},
  \bibinfo{pages}{618--632}.
\newblock
\showISBNx{9783540755197}
\urldef\tempurl%
\url{https://doi.org/10.1007/978-3-540-75520-3_55}
\showDOI{\tempurl}


\bibitem[\protect\citeauthoryear{Chiba and Nishizeki}{Chiba and
  Nishizeki}{1985}]%
        {Chiba1985}
\bibfield{author}{\bibinfo{person}{Norishige Chiba} {and}
  \bibinfo{person}{Takao Nishizeki}.} \bibinfo{year}{1985}\natexlab{}.
\newblock \showarticletitle{Arboricity and Subgraph Listing Algorithms}.
\newblock \bibinfo{journal}{\emph{SIAM J. Comput.}} \bibinfo{volume}{14},
  \bibinfo{number}{1} (\bibinfo{date}{feb} \bibinfo{year}{1985}),
  \bibinfo{pages}{210--223}.
\newblock
\urldef\tempurl%
\url{https://doi.org/10.1137/0214017}
\showDOI{\tempurl}


\bibitem[\protect\citeauthoryear{Ciglan, Averbuch, and Hluchy}{Ciglan
  et~al\mbox{.}}{2012}]%
        {Ciglan2012DB}
\bibfield{author}{\bibinfo{person}{Marek Ciglan}, \bibinfo{person}{Alex
  Averbuch}, {and} \bibinfo{person}{Ladialav Hluchy}.}
  \bibinfo{year}{2012}\natexlab{}.
\newblock \showarticletitle{Benchmarking Traversal Operations over Graph
  Databases}. In \bibinfo{booktitle}{\emph{2012 IEEE 28th International
  Conference on Data Engineering Workshops}}. \bibinfo{publisher}{IEEE},
  \bibinfo{pages}{186--189}.
\newblock
\urldef\tempurl%
\url{https://doi.org/10.1109/icdew.2012.47}
\showDOI{\tempurl}


\bibitem[\protect\citeauthoryear{de~Lima, da~Silva, and Vignatti}{de~Lima
  et~al\mbox{.}}{2022}]%
        {Lima2022ClusteringLatin}
\bibfield{author}{\bibinfo{person}{Alane~M. de Lima}, \bibinfo{person}{Murilo
  V.~G. da Silva}, {and} \bibinfo{person}{André~L. Vignatti}.}
  \bibinfo{year}{2022}\natexlab{}.
\newblock \bibinfo{booktitle}{\emph{Estimating the Clustering Coefficient
  Using Sample Complexity Analysis}}.
\newblock \bibinfo{publisher}{Springer International Publishing},
  \bibinfo{pages}{328--341}.
\newblock
\showISBNx{9783031206245}
\showISSN{1611-3349}
\urldef\tempurl%
\url{https://doi.org/10.1007/978-3-031-20624-5_20}
\showDOI{\tempurl}


\bibitem[\protect\citeauthoryear{Eden, Levi, Ron, and Seshadhri}{Eden
  et~al\mbox{.}}{2017}]%
        {Eden2017Tri}
\bibfield{author}{\bibinfo{person}{Talya Eden}, \bibinfo{person}{Amit Levi},
  \bibinfo{person}{Dana Ron}, {and} \bibinfo{person}{C. Seshadhri}.}
  \bibinfo{year}{2017}\natexlab{}.
\newblock \showarticletitle{Approximately Counting Triangles in Sublinear
  Time}.
\newblock \bibinfo{journal}{\emph{SIAM J. Comput.}} \bibinfo{volume}{46},
  \bibinfo{number}{5} (\bibinfo{date}{Jan.} \bibinfo{year}{2017}),
  \bibinfo{pages}{1603--1646}.
\newblock
\showISSN{1095-7111}
\urldef\tempurl%
\url{https://doi.org/10.1137/15m1054389}
\showDOI{\tempurl}


\bibitem[\protect\citeauthoryear{Etemadi and Lu}{Etemadi and Lu}{2017}]%
        {Etemadi2017BiasCorrection}
\bibfield{author}{\bibinfo{person}{Roohollah Etemadi} {and}
  \bibinfo{person}{Jianguo Lu}.} \bibinfo{year}{2017}\natexlab{}.
\newblock \showarticletitle{Bias correction in clustering coefficient
  estimation}. In \bibinfo{booktitle}{\emph{2017 IEEE International Conference
  on Big Data (Big Data)}}. \bibinfo{publisher}{IEEE},
  \bibinfo{pages}{606--615}.
\newblock
\urldef\tempurl%
\url{https://doi.org/10.1109/bigdata.2017.8257976}
\showDOI{\tempurl}


\bibitem[\protect\citeauthoryear{Fang, Huang, Qin, Zhang, Zhang, Cheng, and
  Lin}{Fang et~al\mbox{.}}{2019}]%
        {Fang2019SurveySN}
\bibfield{author}{\bibinfo{person}{Yixiang Fang}, \bibinfo{person}{Xin Huang},
  \bibinfo{person}{Lu Qin}, \bibinfo{person}{Ying Zhang},
  \bibinfo{person}{Wenjie Zhang}, \bibinfo{person}{Reynold Cheng}, {and}
  \bibinfo{person}{Xuemin Lin}.} \bibinfo{year}{2019}\natexlab{}.
\newblock \showarticletitle{A survey of community search over big graphs}.
\newblock \bibinfo{journal}{\emph{The VLDB Journal}} \bibinfo{volume}{29},
  \bibinfo{number}{1} (\bibinfo{date}{July} \bibinfo{year}{2019}),
  \bibinfo{pages}{353--392}.
\newblock
\showISSN{0949-877X}
\urldef\tempurl%
\url{https://doi.org/10.1007/s00778-019-00556-x}
\showDOI{\tempurl}


\bibitem[\protect\citeauthoryear{Hardiman and Katzir}{Hardiman and
  Katzir}{2013}]%
        {Hardiman2013EstimatingCCRW}
\bibfield{author}{\bibinfo{person}{Stephen~J. Hardiman} {and}
  \bibinfo{person}{Liran Katzir}.} \bibinfo{year}{2013}\natexlab{}.
\newblock \showarticletitle{Estimating clustering coefficients and size of
  social networks via random walk}. In \bibinfo{booktitle}{\emph{Proceedings of
  the 22nd international conference on World Wide Web}}
  \emph{(\bibinfo{series}{WWW ’13})}. \bibinfo{publisher}{ACM},
  \bibinfo{pages}{539--550}.
\newblock
\urldef\tempurl%
\url{https://doi.org/10.1145/2488388.2488436}
\showDOI{\tempurl}


\bibitem[\protect\citeauthoryear{Jha, Seshadhri, and Pinar}{Jha
  et~al\mbox{.}}{2015}]%
        {Jha2015}
\bibfield{author}{\bibinfo{person}{Madhav Jha}, \bibinfo{person}{C. Seshadhri},
  {and} \bibinfo{person}{Ali Pinar}.} \bibinfo{year}{2015}\natexlab{}.
\newblock \showarticletitle{Path Sampling}. In
  \bibinfo{booktitle}{\emph{Proceedings of the 24th International Conference on
  World Wide Web}}. \bibinfo{publisher}{International World Wide Web
  Conferences Steering Committee}.
\newblock
\urldef\tempurl%
\url{https://doi.org/10.1145/2736277.2741101}
\showDOI{\tempurl}


\bibitem[\protect\citeauthoryear{Jiang, Zhao, and Yin}{Jiang
  et~al\mbox{.}}{2008}]%
        {Jiang2008Traffic}
\bibfield{author}{\bibinfo{person}{Bin Jiang}, \bibinfo{person}{Sijian Zhao},
  {and} \bibinfo{person}{Junjun Yin}.} \bibinfo{year}{2008}\natexlab{}.
\newblock \showarticletitle{Self-organized natural roads for predicting traffic
  flow: a sensitivity study}.
\newblock \bibinfo{journal}{\emph{Journal of Statistical Mechanics: Theory and
  Experiment}} \bibinfo{volume}{2008}, \bibinfo{number}{07}
  (\bibinfo{date}{July} \bibinfo{year}{2008}), \bibinfo{pages}{P07008}.
\newblock
\showISSN{1742-5468}
\urldef\tempurl%
\url{https://doi.org/10.1088/1742-5468/2008/07/p07008}
\showDOI{\tempurl}


\bibitem[\protect\citeauthoryear{Jin, Lee, and Hong}{Jin et~al\mbox{.}}{2011}]%
        {Jin2011Similarity}
\bibfield{author}{\bibinfo{person}{Ruoming Jin}, \bibinfo{person}{Victor~E.
  Lee}, {and} \bibinfo{person}{Hui Hong}.} \bibinfo{year}{2011}\natexlab{}.
\newblock \showarticletitle{Axiomatic ranking of network role similarity}. In
  \bibinfo{booktitle}{\emph{Proceedings of the 17th ACM SIGKDD international
  conference on Knowledge discovery and data mining}}
  \emph{(\bibinfo{series}{KDD ’11})}. \bibinfo{publisher}{ACM},
  \bibinfo{pages}{922--930}.
\newblock
\urldef\tempurl%
\url{https://doi.org/10.1145/2020408.2020561}
\showDOI{\tempurl}


\bibitem[\protect\citeauthoryear{Kaiser}{Kaiser}{2008}]%
        {Kaiser2008MeanCC}
\bibfield{author}{\bibinfo{person}{Marcus Kaiser}.}
  \bibinfo{year}{2008}\natexlab{}.
\newblock \showarticletitle{Mean clustering coefficients: the role of isolated
  nodes and leafs on clustering measures for small-world networks}.
\newblock \bibinfo{journal}{\emph{New Journal of Physics}}
  \bibinfo{volume}{10}, \bibinfo{number}{8} (\bibinfo{date}{Aug.}
  \bibinfo{year}{2008}), \bibinfo{pages}{083042}.
\newblock
\showISSN{1367-2630}
\urldef\tempurl%
\url{https://doi.org/10.1088/1367-2630/10/8/083042}
\showDOI{\tempurl}


\bibitem[\protect\citeauthoryear{Karypis and Kumar}{Karypis and Kumar}{1998}]%
        {Karypis1998metis}
\bibfield{author}{\bibinfo{person}{George Karypis} {and} \bibinfo{person}{Vipin
  Kumar}.} \bibinfo{year}{1998}\natexlab{}.
\newblock \showarticletitle{A Fast and High Quality Multilevel Scheme for
  Partitioning Irregular Graphs}.
\newblock \bibinfo{journal}{\emph{SIAM Journal on Scientific Computing}}
  \bibinfo{volume}{20}, \bibinfo{number}{1} (\bibinfo{date}{Jan.}
  \bibinfo{year}{1998}), \bibinfo{pages}{359--392}.
\newblock
\showISSN{1095-7197}
\urldef\tempurl%
\url{https://doi.org/10.1137/s1064827595287997}
\showDOI{\tempurl}


\bibitem[\protect\citeauthoryear{Kolda, Pinar, Plantenga, Seshadhri, and
  Task}{Kolda et~al\mbox{.}}{2014}]%
        {Kolda2014}
\bibfield{author}{\bibinfo{person}{Tamara~G. Kolda}, \bibinfo{person}{Ali
  Pinar}, \bibinfo{person}{Todd Plantenga}, \bibinfo{person}{C. Seshadhri},
  {and} \bibinfo{person}{Christine Task}.} \bibinfo{year}{2014}\natexlab{}.
\newblock \showarticletitle{Counting Triangles in Massive Graphs with
  {MapReduce}}.
\newblock \bibinfo{journal}{\emph{{SIAM} Journal on Scientific Computing}}
  \bibinfo{volume}{36}, \bibinfo{number}{5} (\bibinfo{date}{jan}
  \bibinfo{year}{2014}), \bibinfo{pages}{S48--S77}.
\newblock
\urldef\tempurl%
\url{https://doi.org/10.1137/13090729x}
\showDOI{\tempurl}


\bibitem[\protect\citeauthoryear{Kolountzakis, Miller, Peng, and
  Tsourakakis}{Kolountzakis et~al\mbox{.}}{2012}]%
        {Kolountzakis2012}
\bibfield{author}{\bibinfo{person}{Mihail~N. Kolountzakis},
  \bibinfo{person}{Gary~L. Miller}, \bibinfo{person}{Richard Peng}, {and}
  \bibinfo{person}{Charalampos~E. Tsourakakis}.}
  \bibinfo{year}{2012}\natexlab{}.
\newblock \showarticletitle{Efficient Triangle Counting in Large Graphs via
  Degree-Based Vertex Partitioning}.
\newblock \bibinfo{journal}{\emph{Internet Mathematics}} \bibinfo{volume}{8},
  \bibinfo{number}{1-2} (\bibinfo{date}{mar} \bibinfo{year}{2012}),
  \bibinfo{pages}{161--185}.
\newblock
\urldef\tempurl%
\url{https://doi.org/10.1080/15427951.2012.625260}
\showDOI{\tempurl}


\bibitem[\protect\citeauthoryear{Kutzkov and Pagh}{Kutzkov and Pagh}{2013}]%
        {Kutzkov2013StreamCC}
\bibfield{author}{\bibinfo{person}{Konstantin Kutzkov} {and}
  \bibinfo{person}{Rasmus Pagh}.} \bibinfo{year}{2013}\natexlab{}.
\newblock \showarticletitle{On the streaming complexity of computing local
  clustering coefficients}. In \bibinfo{booktitle}{\emph{Proceedings of the
  sixth ACM international conference on Web search and data mining}}
  \emph{(\bibinfo{series}{WSDM 2013})}, Vol.~\bibinfo{volume}{5}.
  \bibinfo{publisher}{ACM}, \bibinfo{pages}{677--686}.
\newblock
\urldef\tempurl%
\url{https://doi.org/10.1145/2433396.2433480}
\showDOI{\tempurl}


\bibitem[\protect\citeauthoryear{Lattanzi and Leonardi}{Lattanzi and
  Leonardi}{2016}]%
        {Lattanzi2016WeightedCC}
\bibfield{author}{\bibinfo{person}{Silvio Lattanzi} {and}
  \bibinfo{person}{Stefano Leonardi}.} \bibinfo{year}{2016}\natexlab{}.
\newblock \showarticletitle{Efficient computation of the Weighted Clustering
  Coefficient}.
\newblock \bibinfo{journal}{\emph{Internet Mathematics}} \bibinfo{volume}{12},
  \bibinfo{number}{6} (\bibinfo{date}{June} \bibinfo{year}{2016}),
  \bibinfo{pages}{381--401}.
\newblock
\showISSN{1944-9488}
\urldef\tempurl%
\url{https://doi.org/10.1080/15427951.2016.1198281}
\showDOI{\tempurl}


\bibitem[\protect\citeauthoryear{Leskovec, Singh, and Kleinberg}{Leskovec
  et~al\mbox{.}}{2006}]%
        {Leskovec2006PatternsPurchase}
\bibfield{author}{\bibinfo{person}{Jure Leskovec}, \bibinfo{person}{Ajit
  Singh}, {and} \bibinfo{person}{Jon Kleinberg}.}
  \bibinfo{year}{2006}\natexlab{}.
\newblock \bibinfo{booktitle}{\emph{Patterns of Influence in a Recommendation
  Network}}.
\newblock \bibinfo{publisher}{Springer Berlin Heidelberg},
  \bibinfo{pages}{380--389}.
\newblock
\showISBNx{9783540332077}
\showISSN{1611-3349}
\urldef\tempurl%
\url{https://doi.org/10.1007/11731139_44}
\showDOI{\tempurl}


\bibitem[\protect\citeauthoryear{Leskovec and Sosič}{Leskovec and
  Sosič}{2016}]%
        {Leskovec2016SNAP}
\bibfield{author}{\bibinfo{person}{Jure Leskovec} {and} \bibinfo{person}{Rok
  Sosič}.} \bibinfo{year}{2016}\natexlab{}.
\newblock \showarticletitle{SNAP: A General-Purpose Network Analysis and
  Graph-Mining Library}.
\newblock \bibinfo{journal}{\emph{ACM Transactions on Intelligent Systems and
  Technology}} \bibinfo{volume}{8}, \bibinfo{number}{1} (\bibinfo{date}{July}
  \bibinfo{year}{2016}), \bibinfo{pages}{1--20}.
\newblock
\showISSN{2157-6912}
\urldef\tempurl%
\url{https://doi.org/10.1145/2898361}
\showDOI{\tempurl}


\bibitem[\protect\citeauthoryear{Li and Yu}{Li and Yu}{2024}]%
        {Li2024FastLocalCnt}
\bibfield{author}{\bibinfo{person}{Qiyan Li} {and} \bibinfo{person}{Jeffrey~Xu
  Yu}.} \bibinfo{year}{2024}\natexlab{}.
\newblock \showarticletitle{Fast Local Subgraph Counting}.
\newblock \bibinfo{journal}{\emph{Proceedings of the VLDB Endowment}}
  \bibinfo{volume}{17}, \bibinfo{number}{8} (\bibinfo{date}{April}
  \bibinfo{year}{2024}), \bibinfo{pages}{1967--1980}.
\newblock
\showISSN{2150-8097}
\urldef\tempurl%
\url{https://doi.org/10.14778/3659437.3659451}
\showDOI{\tempurl}


\bibitem[\protect\citeauthoryear{Li, Qin, Yu, and Mao}{Li
  et~al\mbox{.}}{2015}]%
        {Li2015Community}
\bibfield{author}{\bibinfo{person}{Rong-Hua Li}, \bibinfo{person}{Lu Qin},
  \bibinfo{person}{Jeffrey~Xu Yu}, {and} \bibinfo{person}{Rui Mao}.}
  \bibinfo{year}{2015}\natexlab{}.
\newblock \showarticletitle{Influential community search in large networks}.
\newblock \bibinfo{journal}{\emph{Proceedings of the VLDB Endowment}}
  \bibinfo{volume}{8}, \bibinfo{number}{5} (\bibinfo{date}{Jan.}
  \bibinfo{year}{2015}), \bibinfo{pages}{509--520}.
\newblock
\showISSN{2150-8097}
\urldef\tempurl%
\url{https://doi.org/10.14778/2735479.2735484}
\showDOI{\tempurl}


\bibitem[\protect\citeauthoryear{Li, Shang, and Yang}{Li et~al\mbox{.}}{2017}]%
        {Li2017ClustLarge}
\bibfield{author}{\bibinfo{person}{Yusheng Li}, \bibinfo{person}{Yilun Shang},
  {and} \bibinfo{person}{Yiting Yang}.} \bibinfo{year}{2017}\natexlab{}.
\newblock \showarticletitle{Clustering coefficients of large networks}.
\newblock \bibinfo{journal}{\emph{Information Sciences}}
  \bibinfo{volume}{382–383} (\bibinfo{date}{March} \bibinfo{year}{2017}),
  \bibinfo{pages}{350--358}.
\newblock
\showISSN{0020-0255}
\urldef\tempurl%
\url{https://doi.org/10.1016/j.ins.2016.12.027}
\showDOI{\tempurl}


\bibitem[\protect\citeauthoryear{Maurer and Pontil}{Maurer and Pontil}{2009}]%
        {Maurer2009EmpiricalBern}
\bibfield{author}{\bibinfo{person}{Andreas Maurer} {and}
  \bibinfo{person}{Massimiliano Pontil}.} \bibinfo{year}{2009}\natexlab{}.
\newblock \showarticletitle{Empirical Bernstein Bounds and Sample Variance
  Penalization}.
\newblock  (\bibinfo{date}{July} \bibinfo{year}{2009}).
\newblock
\urldef\tempurl%
\url{https://doi.org/10.48550/ARXIV.0907.3740}
\showDOI{\tempurl}
\showeprint[arxiv]{0907.3740}~[stat.ML]


\bibitem[\protect\citeauthoryear{Mitzenmacher}{Mitzenmacher}{2017}]%
        {Mitzenmacher2017PC}
\bibfield{author}{\bibinfo{person}{Michael Mitzenmacher}.}
  \bibinfo{year}{2017}\natexlab{}.
\newblock \bibinfo{booktitle}{\emph{Probability and computing}
  (\bibinfo{edition}{second edition} ed.)}.
\newblock \bibinfo{publisher}{Cambridge University Press},
  \bibinfo{address}{Cambridge}.
\newblock
\showISBNx{9781107154889}
\newblock
\shownote{Hier auch später erschienene, unveränderte Nachdrucke.}


\bibitem[\protect\citeauthoryear{Mnih, Szepesvári, and Audibert}{Mnih
  et~al\mbox{.}}{2008}]%
        {Mnih2008EmpiricalStopping}
\bibfield{author}{\bibinfo{person}{Volodymyr Mnih}, \bibinfo{person}{Csaba
  Szepesvári}, {and} \bibinfo{person}{Jean-Yves Audibert}.}
  \bibinfo{year}{2008}\natexlab{}.
\newblock \showarticletitle{Empirical Bernstein stopping}. In
  \bibinfo{booktitle}{\emph{Proceedings of the 25th international conference on
  Machine learning - ICML ’08}} \emph{(\bibinfo{series}{ICML ’08})}.
  \bibinfo{publisher}{ACM Press}, \bibinfo{pages}{672--679}.
\newblock
\urldef\tempurl%
\url{https://doi.org/10.1145/1390156.1390241}
\showDOI{\tempurl}


\bibitem[\protect\citeauthoryear{Newman}{Newman}{2018}]%
        {Newman2018Networks}
\bibfield{author}{\bibinfo{person}{Mark Newman}.}
  \bibinfo{year}{2018}\natexlab{}.
\newblock \bibinfo{booktitle}{\emph{Networks}}.
\newblock \bibinfo{publisher}{Oxford University Press}.
\newblock
\showISBNx{9780198805090}
\urldef\tempurl%
\url{https://doi.org/10.1093/oso/9780198805090.001.0001}
\showDOI{\tempurl}


\bibitem[\protect\citeauthoryear{Pan, Xu, Wang, and Zhang}{Pan
  et~al\mbox{.}}{2019}]%
        {Pan2019CommCC}
\bibfield{author}{\bibinfo{person}{Xiaohui Pan}, \bibinfo{person}{Guiqiong Xu},
  \bibinfo{person}{Bing Wang}, {and} \bibinfo{person}{Tao Zhang}.}
  \bibinfo{year}{2019}\natexlab{}.
\newblock \showarticletitle{A Novel Community Detection Algorithm Based on
  Local Similarity of Clustering Coefficient in Social Networks}.
\newblock \bibinfo{journal}{\emph{IEEE Access}}  \bibinfo{volume}{7}
  (\bibinfo{year}{2019}), \bibinfo{pages}{121586--121598}.
\newblock
\showISSN{2169-3536}
\urldef\tempurl%
\url{https://doi.org/10.1109/access.2019.2937580}
\showDOI{\tempurl}


\bibitem[\protect\citeauthoryear{Pashanasangi and Seshadhri}{Pashanasangi and
  Seshadhri}{2020}]%
        {Pashanasangi2020EVOKE}
\bibfield{author}{\bibinfo{person}{Noujan Pashanasangi} {and}
  \bibinfo{person}{C. Seshadhri}.} \bibinfo{year}{2020}\natexlab{}.
\newblock \showarticletitle{Efficiently Counting Vertex Orbits of All 5-vertex
  Subgraphs, by EVOKE}. In \bibinfo{booktitle}{\emph{Proceedings of the 13th
  International Conference on Web Search and Data Mining}}
  \emph{(\bibinfo{series}{WSDM ’20})}. \bibinfo{publisher}{ACM}.
\newblock
\urldef\tempurl%
\url{https://doi.org/10.1145/3336191.3371773}
\showDOI{\tempurl}


\bibitem[\protect\citeauthoryear{Pellegrina}{Pellegrina}{2023}]%
        {Pellegrina2023}
\bibfield{author}{\bibinfo{person}{Leonardo Pellegrina}.}
  \bibinfo{year}{2023}\natexlab{}.
\newblock \showarticletitle{Efficient Centrality Maximization with Rademacher
  Averages}. In \bibinfo{booktitle}{\emph{Proceedings of the 29th {ACM}
  {SIGKDD} Conference on Knowledge Discovery and Data Mining}}.
  \bibinfo{publisher}{{ACM}}.
\newblock
\urldef\tempurl%
\url{https://doi.org/10.1145/3580305.3599325}
\showDOI{\tempurl}


\bibitem[\protect\citeauthoryear{Pellegrina and Vandin}{Pellegrina and
  Vandin}{2023}]%
        {Pellegrina2023Silvan}
\bibfield{author}{\bibinfo{person}{Leonardo Pellegrina} {and}
  \bibinfo{person}{Fabio Vandin}.} \bibinfo{year}{2023}\natexlab{}.
\newblock \showarticletitle{SILVAN: Estimating Betweenness Centralities with
  Progressive Sampling and Non-uniform Rademacher Bounds}.
\newblock \bibinfo{journal}{\emph{ACM Transactions on Knowledge Discovery from
  Data}} \bibinfo{volume}{18}, \bibinfo{number}{3} (\bibinfo{date}{Dec.}
  \bibinfo{year}{2023}), \bibinfo{pages}{1--55}.
\newblock
\showISSN{1556-472X}
\urldef\tempurl%
\url{https://doi.org/10.1145/3628601}
\showDOI{\tempurl}


\bibitem[\protect\citeauthoryear{Rahman, Bhuiyan, and Al~Hasan}{Rahman
  et~al\mbox{.}}{2014}]%
        {Rahman2014Graphlets}
\bibfield{author}{\bibinfo{person}{Mahmudur Rahman},
  \bibinfo{person}{Mansurul~Alam Bhuiyan}, {and} \bibinfo{person}{Mohammad
  Al~Hasan}.} \bibinfo{year}{2014}\natexlab{}.
\newblock \showarticletitle{Graft: An Efficient Graphlet Counting Method for
  Large Graph Analysis}.
\newblock \bibinfo{journal}{\emph{IEEE Transactions on Knowledge and Data
  Engineering}} \bibinfo{volume}{26}, \bibinfo{number}{10}
  (\bibinfo{date}{Oct.} \bibinfo{year}{2014}), \bibinfo{pages}{2466--2478}.
\newblock
\showISSN{1041-4347}
\urldef\tempurl%
\url{https://doi.org/10.1109/tkde.2013.2297929}
\showDOI{\tempurl}


\bibitem[\protect\citeauthoryear{Riondato and Upfal}{Riondato and
  Upfal}{2018}]%
        {Riondato2018Abra}
\bibfield{author}{\bibinfo{person}{Matteo Riondato} {and} \bibinfo{person}{Eli
  Upfal}.} \bibinfo{year}{2018}\natexlab{}.
\newblock \showarticletitle{ABRA: Approximating Betweenness Centrality in
  Static and Dynamic Graphs with Rademacher Averages}.
\newblock \bibinfo{journal}{\emph{ACM Transactions on Knowledge Discovery from
  Data}} \bibinfo{volume}{12}, \bibinfo{number}{5} (\bibinfo{date}{July}
  \bibinfo{year}{2018}), \bibinfo{pages}{1--38}.
\newblock
\showISSN{1556-472X}
\urldef\tempurl%
\url{https://doi.org/10.1145/3208351}
\showDOI{\tempurl}


\bibitem[\protect\citeauthoryear{Riondato and Vandin}{Riondato and
  Vandin}{2018}]%
        {Riondato2018MisoSoup}
\bibfield{author}{\bibinfo{person}{Matteo Riondato} {and}
  \bibinfo{person}{Fabio Vandin}.} \bibinfo{year}{2018}\natexlab{}.
\newblock \showarticletitle{MiSoSouP: Mining Interesting Subgroups with
  Sampling and Pseudodimension}. In \bibinfo{booktitle}{\emph{Proceedings of
  the 24th ACM SIGKDD International Conference on Knowledge Discovery {\&} Data
  Mining}} \emph{(\bibinfo{series}{KDD ’18})}. \bibinfo{publisher}{ACM}.
\newblock
\urldef\tempurl%
\url{https://doi.org/10.1145/3219819.3219989}
\showDOI{\tempurl}


\bibitem[\protect\citeauthoryear{Rossi and Ahmed}{Rossi and Ahmed}{2015}]%
        {Rossi2015NetworkRep}
\bibfield{author}{\bibinfo{person}{Ryan Rossi} {and} \bibinfo{person}{Nesreen
  Ahmed}.} \bibinfo{year}{2015}\natexlab{}.
\newblock \showarticletitle{The Network Data Repository with Interactive Graph
  Analytics and Visualization}.
\newblock \bibinfo{journal}{\emph{Proceedings of the AAAI Conference on
  Artificial Intelligence}} \bibinfo{volume}{29}, \bibinfo{number}{1}
  (\bibinfo{date}{March} \bibinfo{year}{2015}).
\newblock
\showISSN{2159-5399}
\urldef\tempurl%
\url{https://doi.org/10.1609/aaai.v29i1.9277}
\showDOI{\tempurl}


\bibitem[\protect\citeauthoryear{Rossi, Ahmed, Carranza, Arbour, Rao, Kim, and
  Koh}{Rossi et~al\mbox{.}}{2019}]%
        {Rossi2019}
\bibfield{author}{\bibinfo{person}{Ryan~A. Rossi}, \bibinfo{person}{Nesreen~K.
  Ahmed}, \bibinfo{person}{Aldo Carranza}, \bibinfo{person}{David Arbour},
  \bibinfo{person}{Anup Rao}, \bibinfo{person}{Sungchul Kim}, {and}
  \bibinfo{person}{Eunyee Koh}.} \bibinfo{year}{2019}\natexlab{}.
\newblock \showarticletitle{Heterogeneous Network Motifs}.
\newblock  (\bibinfo{date}{Jan.} \bibinfo{year}{2019}).
\newblock
\urldef\tempurl%
\url{https://doi.org/10.48550/ARXIV.1901.10026}
\showDOI{\tempurl}
\showeprint[arxiv]{1901.10026}~[cs.SI]


\bibitem[\protect\citeauthoryear{Rossi, Rao, Mai, and Ahmed}{Rossi
  et~al\mbox{.}}{2020}]%
        {Rossi2020}
\bibfield{author}{\bibinfo{person}{Ryan~A. Rossi}, \bibinfo{person}{Anup Rao},
  \bibinfo{person}{Tung Mai}, {and} \bibinfo{person}{Nesreen~K. Ahmed}.}
  \bibinfo{year}{2020}\natexlab{}.
\newblock \showarticletitle{Fast and Accurate Estimation of Typed Graphlets}.
  In \bibinfo{booktitle}{\emph{Companion Proceedings of the Web Conference
  2020}}. \bibinfo{publisher}{{ACM}}.
\newblock
\urldef\tempurl%
\url{https://doi.org/10.1145/3366424.3382683}
\showDOI{\tempurl}


\bibitem[\protect\citeauthoryear{Schank and Wagner}{Schank and Wagner}{2005}]%
        {Schank2005CC}
\bibfield{author}{\bibinfo{person}{Thomas Schank} {and}
  \bibinfo{person}{Dorothea Wagner}.} \bibinfo{year}{2005}\natexlab{}.
\newblock \showarticletitle{Approximating Clustering Coefficient and
  Transitivity}.
\newblock \bibinfo{journal}{\emph{Journal of Graph Algorithms and
  Applications}} \bibinfo{volume}{9}, \bibinfo{number}{2}
  (\bibinfo{year}{2005}), \bibinfo{pages}{265--275}.
\newblock
\showISSN{1526-1719}
\urldef\tempurl%
\url{https://doi.org/10.7155/jgaa.00108}
\showDOI{\tempurl}


\bibitem[\protect\citeauthoryear{Seshadhri, Pinar, and Kolda}{Seshadhri
  et~al\mbox{.}}{2014}]%
        {Seshadhri2014}
\bibfield{author}{\bibinfo{person}{C. Seshadhri}, \bibinfo{person}{Ali Pinar},
  {and} \bibinfo{person}{Tamara~G. Kolda}.} \bibinfo{year}{2014}\natexlab{}.
\newblock \showarticletitle{Wedge sampling for computing clustering
  coefficients and triangle counts on large graphs}.
\newblock \bibinfo{journal}{\emph{Statistical Analysis and Data Mining: The
  {ASA} Data Science Journal}} \bibinfo{volume}{7}, \bibinfo{number}{4}
  (\bibinfo{date}{may} \bibinfo{year}{2014}), \bibinfo{pages}{294--307}.
\newblock
\urldef\tempurl%
\url{https://doi.org/10.1002/sam.11224}
\showDOI{\tempurl}


\bibitem[\protect\citeauthoryear{Seshadhri and Tirthapura}{Seshadhri and
  Tirthapura}{2019}]%
        {Seshadhri2019Tutorial}
\bibfield{author}{\bibinfo{person}{Comandur Seshadhri} {and}
  \bibinfo{person}{Srikanta Tirthapura}.} \bibinfo{year}{2019}\natexlab{}.
\newblock \showarticletitle{Scalable Subgraph Counting: The Methods Behind The
  Madness}. In \bibinfo{booktitle}{\emph{Companion Proceedings of The 2019
  World Wide Web Conference}} \emph{(\bibinfo{series}{WWW ’19})}.
  \bibinfo{publisher}{ACM}, \bibinfo{pages}{1317--1318}.
\newblock
\urldef\tempurl%
\url{https://doi.org/10.1145/3308560.3320092}
\showDOI{\tempurl}


\bibitem[\protect\citeauthoryear{Shalev-Shwartz}{Shalev-Shwartz}{2014}]%
        {ShalevShwartz2014UML}
\bibfield{author}{\bibinfo{person}{Shai Shalev-Shwartz}.}
  \bibinfo{year}{2014}\natexlab{}.
\newblock \bibinfo{booktitle}{\emph{Understanding machine learning}}.
\newblock \bibinfo{publisher}{Cambrige University Press},
  \bibinfo{address}{Cambridge}.
\newblock
\showISBNx{9781107057135}
\newblock
\shownote{Hier auch später erschienene, unveränderte Nachdrucke.}


\bibitem[\protect\citeauthoryear{Shekhar and Ramdas}{Shekhar and
  Ramdas}{2023}]%
        {Shekhar2023OptimalBetting}
\bibfield{author}{\bibinfo{person}{Shubhanshu Shekhar} {and}
  \bibinfo{person}{Aaditya Ramdas}.} \bibinfo{year}{2023}\natexlab{}.
\newblock \bibinfo{title}{On the near-optimality of betting confidence sets for
  bounded means}.
\newblock
\newblock
\urldef\tempurl%
\url{https://doi.org/10.48550/ARXIV.2310.01547}
\showDOI{\tempurl}


\bibitem[\protect\citeauthoryear{Shi, Li, Zhang, Sun, and Yu}{Shi
  et~al\mbox{.}}{2017}]%
        {Shi2017Heter}
\bibfield{author}{\bibinfo{person}{Chuan Shi}, \bibinfo{person}{Yitong Li},
  \bibinfo{person}{Jiawei Zhang}, \bibinfo{person}{Yizhou Sun}, {and}
  \bibinfo{person}{Philip~S. Yu}.} \bibinfo{year}{2017}\natexlab{}.
\newblock \showarticletitle{A survey of heterogeneous information network
  analysis}.
\newblock \bibinfo{journal}{\emph{IEEE Transactions on Knowledge and Data
  Engineering}} \bibinfo{volume}{29}, \bibinfo{number}{1} (\bibinfo{date}{Jan.}
  \bibinfo{year}{2017}), \bibinfo{pages}{17--37}.
\newblock
\showISSN{2326-3865}
\urldef\tempurl%
\url{https://doi.org/10.1109/tkde.2016.2598561}
\showDOI{\tempurl}


\bibitem[\protect\citeauthoryear{Shin}{Shin}{2017}]%
        {Shin2017WRS}
\bibfield{author}{\bibinfo{person}{Kijung Shin}.}
  \bibinfo{year}{2017}\natexlab{}.
\newblock \showarticletitle{WRS: Waiting Room Sampling for Accurate Triangle
  Counting in Real Graph Streams}. In \bibinfo{booktitle}{\emph{2017 IEEE
  International Conference on Data Mining (ICDM)}}. \bibinfo{publisher}{IEEE},
  \bibinfo{pages}{1087--1092}.
\newblock
\urldef\tempurl%
\url{https://doi.org/10.1109/icdm.2017.143}
\showDOI{\tempurl}


\bibitem[\protect\citeauthoryear{Shin, Oh, Kim, Hooi, and Faloutsos}{Shin
  et~al\mbox{.}}{2020}]%
        {Shin2020Stream}
\bibfield{author}{\bibinfo{person}{Kijung Shin}, \bibinfo{person}{Sejoon Oh},
  \bibinfo{person}{Jisu Kim}, \bibinfo{person}{Bryan Hooi}, {and}
  \bibinfo{person}{Christos Faloutsos}.} \bibinfo{year}{2020}\natexlab{}.
\newblock \showarticletitle{Fast, Accurate and Provable Triangle Counting in
  Fully Dynamic Graph Streams}.
\newblock \bibinfo{journal}{\emph{ACM Transactions on Knowledge Discovery from
  Data}} \bibinfo{volume}{14}, \bibinfo{number}{2} (\bibinfo{date}{Feb.}
  \bibinfo{year}{2020}), \bibinfo{pages}{1--39}.
\newblock
\showISSN{1556-472X}
\urldef\tempurl%
\url{https://doi.org/10.1145/3375392}
\showDOI{\tempurl}


\bibitem[\protect\citeauthoryear{Stefani, Epasto, Riondato, and Upfal}{Stefani
  et~al\mbox{.}}{2017}]%
        {Stefani2017Triest}
\bibfield{author}{\bibinfo{person}{Lorenzo~De Stefani},
  \bibinfo{person}{Alessandro Epasto}, \bibinfo{person}{Matteo Riondato}, {and}
  \bibinfo{person}{Eli Upfal}.} \bibinfo{year}{2017}\natexlab{}.
\newblock \showarticletitle{TRIÈST: Counting Local and Global Triangles in
  Fully Dynamic Streams with Fixed Memory Size}.
\newblock \bibinfo{journal}{\emph{ACM Transactions on Knowledge Discovery from
  Data}} \bibinfo{volume}{11}, \bibinfo{number}{4} (\bibinfo{date}{June}
  \bibinfo{year}{2017}), \bibinfo{pages}{1--50}.
\newblock
\showISSN{1556-472X}
\urldef\tempurl%
\url{https://doi.org/10.1145/3059194}
\showDOI{\tempurl}


\bibitem[\protect\citeauthoryear{Sun and Han}{Sun and Han}{2013}]%
        {sun2013mining}
\bibfield{author}{\bibinfo{person}{Yizhou Sun} {and} \bibinfo{person}{Jiawei
  Han}.} \bibinfo{year}{2013}\natexlab{}.
\newblock \showarticletitle{Mining heterogeneous information networks: a
  structural analysis approach}.
\newblock \bibinfo{journal}{\emph{ACM SIGKDD explorations newsletter}}
  \bibinfo{volume}{14}, \bibinfo{number}{2} (\bibinfo{year}{2013}),
  \bibinfo{pages}{20--28}.
\newblock


\bibitem[\protect\citeauthoryear{Tsourakakis, Kang, Miller, and
  Faloutsos}{Tsourakakis et~al\mbox{.}}{2009}]%
        {Tsourakakis2009}
\bibfield{author}{\bibinfo{person}{Charalampos~E. Tsourakakis},
  \bibinfo{person}{U. Kang}, \bibinfo{person}{Gary~L. Miller}, {and}
  \bibinfo{person}{Christos Faloutsos}.} \bibinfo{year}{2009}\natexlab{}.
\newblock \showarticletitle{{DOULION}}. In
  \bibinfo{booktitle}{\emph{Proceedings of the 15th {ACM} {SIGKDD}
  international conference on Knowledge discovery and data mining}}.
  \bibinfo{publisher}{{ACM}}.
\newblock
\urldef\tempurl%
\url{https://doi.org/10.1145/1557019.1557111}
\showDOI{\tempurl}


\bibitem[\protect\citeauthoryear{Ugander, Karrer, Backstrom, and
  Marlow}{Ugander et~al\mbox{.}}{2011}]%
        {Ugander2011FB}
\bibfield{author}{\bibinfo{person}{Johan Ugander}, \bibinfo{person}{Brian
  Karrer}, \bibinfo{person}{Lars Backstrom}, {and} \bibinfo{person}{Cameron
  Marlow}.} \bibinfo{year}{2011}\natexlab{}.
\newblock \bibinfo{title}{The Anatomy of the Facebook Social Graph}.
\newblock
\newblock
\urldef\tempurl%
\url{https://doi.org/10.48550/ARXIV.1111.4503}
\showDOI{\tempurl}


\bibitem[\protect\citeauthoryear{Wang, Zhao, Zhang, Li, Cheng, Lui, Towsley,
  Tao, and Guan}{Wang et~al\mbox{.}}{2018}]%
        {Wang2018Moss5}
\bibfield{author}{\bibinfo{person}{Pinghui Wang}, \bibinfo{person}{Junzhou
  Zhao}, \bibinfo{person}{Xiangliang Zhang}, \bibinfo{person}{Zhenguo Li},
  \bibinfo{person}{Jiefeng Cheng}, \bibinfo{person}{John~C.S. Lui},
  \bibinfo{person}{Don Towsley}, \bibinfo{person}{Jing Tao}, {and}
  \bibinfo{person}{Xiaohong Guan}.} \bibinfo{year}{2018}\natexlab{}.
\newblock \showarticletitle{MOSS-5: A Fast Method of Approximating Counts of
  5-Node Graphlets in Large Graphs}.
\newblock \bibinfo{journal}{\emph{IEEE Transactions on Knowledge and Data
  Engineering}} \bibinfo{volume}{30}, \bibinfo{number}{1} (\bibinfo{date}{Jan.}
  \bibinfo{year}{2018}), \bibinfo{pages}{73--86}.
\newblock
\showISSN{1041-4347}
\urldef\tempurl%
\url{https://doi.org/10.1109/tkde.2017.2756836}
\showDOI{\tempurl}


\bibitem[\protect\citeauthoryear{Watts and Strogatz}{Watts and
  Strogatz}{1998}]%
        {Watts1998}
\bibfield{author}{\bibinfo{person}{Duncan~J. Watts} {and}
  \bibinfo{person}{Steven~H. Strogatz}.} \bibinfo{year}{1998}\natexlab{}.
\newblock \showarticletitle{Collective dynamics of `small-world' networks}.
\newblock \bibinfo{journal}{\emph{Nature}} \bibinfo{volume}{393},
  \bibinfo{number}{6684} (\bibinfo{date}{jun} \bibinfo{year}{1998}),
  \bibinfo{pages}{440--442}.
\newblock
\urldef\tempurl%
\url{https://doi.org/10.1038/30918}
\showDOI{\tempurl}


\bibitem[\protect\citeauthoryear{Waudby-Smith and Ramdas}{Waudby-Smith and
  Ramdas}{2023}]%
        {WaudbySmith2023Betting}
\bibfield{author}{\bibinfo{person}{Ian Waudby-Smith} {and}
  \bibinfo{person}{Aaditya Ramdas}.} \bibinfo{year}{2023}\natexlab{}.
\newblock \showarticletitle{Estimating means of bounded random variables by
  betting}.
\newblock \bibinfo{journal}{\emph{Journal of the Royal Statistical Society
  Series B: Statistical Methodology}} \bibinfo{volume}{86}, \bibinfo{number}{1}
  (\bibinfo{date}{Feb.} \bibinfo{year}{2023}), \bibinfo{pages}{1--27}.
\newblock
\showISSN{1467-9868}
\urldef\tempurl%
\url{https://doi.org/10.1093/jrsssb/qkad009}
\showDOI{\tempurl}


\bibitem[\protect\citeauthoryear{Wu, Lin, Wang, and Gregory}{Wu
  et~al\mbox{.}}{2016}]%
        {Wu2016LinkPred}
\bibfield{author}{\bibinfo{person}{Zhihao Wu}, \bibinfo{person}{Youfang Lin},
  \bibinfo{person}{Jing Wang}, {and} \bibinfo{person}{Steve Gregory}.}
  \bibinfo{year}{2016}\natexlab{}.
\newblock \showarticletitle{Link prediction with node clustering coefficient}.
\newblock \bibinfo{journal}{\emph{Physica A: Statistical Mechanics and its
  Applications}}  \bibinfo{volume}{452} (\bibinfo{date}{June}
  \bibinfo{year}{2016}), \bibinfo{pages}{1--8}.
\newblock
\showISSN{0378-4371}
\urldef\tempurl%
\url{https://doi.org/10.1016/j.physa.2016.01.038}
\showDOI{\tempurl}


\bibitem[\protect\citeauthoryear{Xu}{Xu}{2001}]%
        {Xu2001Topology}
\bibfield{author}{\bibinfo{person}{Junming Xu}.}
  \bibinfo{year}{2001}\natexlab{}.
\newblock \bibinfo{booktitle}{\emph{Topological Structure and Analysis of
  Interconnection Networks}}.
\newblock \bibinfo{publisher}{Springer US}.
\newblock
\showISBNx{9781475733877}
\showISSN{1568-1696}
\urldef\tempurl%
\url{https://doi.org/10.1007/978-1-4757-3387-7}
\showDOI{\tempurl}


\bibitem[\protect\citeauthoryear{Yin, Benson, and Leskovec}{Yin
  et~al\mbox{.}}{2018}]%
        {Yin2018HCC}
\bibfield{author}{\bibinfo{person}{Hao Yin}, \bibinfo{person}{Austin~R.
  Benson}, {and} \bibinfo{person}{Jure Leskovec}.}
  \bibinfo{year}{2018}\natexlab{}.
\newblock \showarticletitle{Higher-order clustering in networks}.
\newblock \bibinfo{journal}{\emph{Physical Review E}} \bibinfo{volume}{97},
  \bibinfo{number}{5} (\bibinfo{date}{May} \bibinfo{year}{2018}),
  \bibinfo{pages}{052306}.
\newblock
\showISSN{2470-0053}
\urldef\tempurl%
\url{https://doi.org/10.1103/physreve.97.052306}
\showDOI{\tempurl}


\bibitem[\protect\citeauthoryear{Yin, Benson, and Leskovec}{Yin
  et~al\mbox{.}}{2019}]%
        {Yin2019Closure}
\bibfield{author}{\bibinfo{person}{Hao Yin}, \bibinfo{person}{Austin~R.
  Benson}, {and} \bibinfo{person}{Jure Leskovec}.}
  \bibinfo{year}{2019}\natexlab{}.
\newblock \showarticletitle{The Local Closure Coefficient: A New Perspective On
  Network Clustering}. In \bibinfo{booktitle}{\emph{Proceedings of the Twelfth
  ACM International Conference on Web Search and Data Mining}}
  \emph{(\bibinfo{series}{WSDM ’19})}. \bibinfo{publisher}{ACM}.
\newblock
\urldef\tempurl%
\url{https://doi.org/10.1145/3289600.3290991}
\showDOI{\tempurl}


\bibitem[\protect\citeauthoryear{Yin, Benson, Leskovec, and Gleich}{Yin
  et~al\mbox{.}}{2017}]%
        {Yin2017LocalHG}
\bibfield{author}{\bibinfo{person}{Hao Yin}, \bibinfo{person}{Austin~R.
  Benson}, \bibinfo{person}{Jure Leskovec}, {and} \bibinfo{person}{David~F.
  Gleich}.} \bibinfo{year}{2017}\natexlab{}.
\newblock \showarticletitle{Local Higher-Order Graph Clustering}. In
  \bibinfo{booktitle}{\emph{Proceedings of the 23rd ACM SIGKDD International
  Conference on Knowledge Discovery and Data Mining}}
  \emph{(\bibinfo{series}{KDD ’17})}. \bibinfo{publisher}{ACM}.
\newblock
\urldef\tempurl%
\url{https://doi.org/10.1145/3097983.3098069}
\showDOI{\tempurl}


\bibitem[\protect\citeauthoryear{Yuan}{Yuan}{2024}]%
        {Yuan2024CLTClosure}
\bibfield{author}{\bibinfo{person}{M. Yuan}.} \bibinfo{year}{2024}\natexlab{}.
\newblock \showarticletitle{Central limit theorem for the average closure
  coefficient}.
\newblock \bibinfo{journal}{\emph{Acta Mathematica Hungarica}}
  \bibinfo{volume}{172}, \bibinfo{number}{2} (\bibinfo{date}{March}
  \bibinfo{year}{2024}), \bibinfo{pages}{543--569}.
\newblock
\showISSN{1588-2632}
\urldef\tempurl%
\url{https://doi.org/10.1007/s10474-024-01416-z}
\showDOI{\tempurl}


\bibitem[\protect\citeauthoryear{Zhang, Xiang, Guo, Zhou, and Yang}{Zhang
  et~al\mbox{.}}{2023}]%
        {Zhang2023Anom}
\bibfield{author}{\bibinfo{person}{Chi Zhang}, \bibinfo{person}{Wenkai Xiang},
  \bibinfo{person}{Xingzhi Guo}, \bibinfo{person}{Baojian Zhou}, {and}
  \bibinfo{person}{Deqing Yang}.} \bibinfo{year}{2023}\natexlab{}.
\newblock \showarticletitle{SubAnom: Efficient Subgraph Anomaly Detection
  Framework over Dynamic Graphs}. In \bibinfo{booktitle}{\emph{2023 IEEE
  International Conference on Data Mining Workshops (ICDMW)}}.
  \bibinfo{publisher}{IEEE}, \bibinfo{pages}{1178--1185}.
\newblock
\urldef\tempurl%
\url{https://doi.org/10.1109/icdmw60847.2023.00154}
\showDOI{\tempurl}


\bibitem[\protect\citeauthoryear{Zhang, Zhu, Qin, Cheng, and Yu}{Zhang
  et~al\mbox{.}}{2017}]%
        {Zhang2017LocalCC}
\bibfield{author}{\bibinfo{person}{Hao Zhang}, \bibinfo{person}{Yuanyuan Zhu},
  \bibinfo{person}{Lu Qin}, \bibinfo{person}{Hong Cheng}, {and}
  \bibinfo{person}{Jeffrey~Xu Yu}.} \bibinfo{year}{2017}\natexlab{}.
\newblock \bibinfo{booktitle}{\emph{Efficient Local Clustering Coefficient
  Estimation in Massive Graphs}}.
\newblock \bibinfo{publisher}{Springer International Publishing},
  \bibinfo{pages}{371--386}.
\newblock
\showISBNx{9783319556994}
\showISSN{1611-3349}
\urldef\tempurl%
\url{https://doi.org/10.1007/978-3-319-55699-4_23}
\showDOI{\tempurl}


\bibitem[\protect\citeauthoryear{Zhao, Yu, Zhang, Li, and Rong}{Zhao
  et~al\mbox{.}}{2021}]%
        {Zhao2021Sketch}
\bibfield{author}{\bibinfo{person}{Kangfei Zhao}, \bibinfo{person}{Jeffrey~Xu
  Yu}, \bibinfo{person}{Hao Zhang}, \bibinfo{person}{Qiyan Li}, {and}
  \bibinfo{person}{Yu Rong}.} \bibinfo{year}{2021}\natexlab{}.
\newblock \showarticletitle{A Learned Sketch for Subgraph Counting}. In
  \bibinfo{booktitle}{\emph{Proceedings of the 2021 International Conference on
  Management of Data}} \emph{(\bibinfo{series}{SIGMOD/PODS ’21})}.
  \bibinfo{publisher}{ACM}.
\newblock
\urldef\tempurl%
\url{https://doi.org/10.1145/3448016.3457289}
\showDOI{\tempurl}


\end{thebibliography}

\ifextended
\newpage
\appendix

\section{Missing Proofs}\label{app:missingProofs}
\begin{proof}[Proof of Lemma~\ref{lemma:unbiasednessSingleNode}]
First we let $X_e$ be a Bernoulli random variable taking value 1 with probability $p= \sfrac{1}{m}$ and 0 otherwise for each edge $e\in \E$. Then it follows that $\mathbb{E}[X_e] = p$, hence applying the linearity of expectation to Equation~\eqref{eq:EstimatorQ} we get
\begin{align*}
	\mathbb{E}[X_q(v)] &= \mathbb{E}\left[\sum_{e \in E : v\in e} \frac{q |\Delta_e|X_e}{p} + (1-2q)\sum_{e \in \E  : v\in \neighborEdge{e}}\frac{X_e}{p}\right]\\
	&= \sum_{e \in E : v\in e} \frac{q |\Delta_e|\mathbb{E}[X_e]}{p} + (1-2q)\sum_{e \in \E  : v\in \neighborEdge{e}}\frac{\mathbb{E}[X_e]}{p} \\
	&= \sum_{e \in E : v\in e} q |\Delta_e| + (1-2q)\sum_{e \in \E  : v\in \neighborEdge{e}}1 \\
	= q& \left( \sum_{e \in E : v\in e} |\Delta_e| -  2\sum_{e=\{u,v\} : v\in \mathcal{N}_u \cap \mathcal{N}_v}1\right) + \sum_{e \in \E  : v\in \neighborEdge{e}}1\\
	&=  q\left(2|\Delta_v| - 2|\Delta_v|\right) + |\Delta_v| =  |\Delta_v|.
\end{align*}
Where we used the fact that $\sum_{e \in E : v\in e}|\Delta_e| = 2|\Delta_v|$ as each triangle $\singleTri \in \triSetNode{v}$ is counted twice and $\sum_{e\in \E : v\in \neighborEdge{e}} 1 = |\Delta_v|$ as the number of triangles containing $v\in \V$ corresponds to the number of edges connecting two neighbors of $v\in \V$.
\end{proof}

\begin{proof}[Proof of Lemma~\ref{lemma:unbiasedEdgePartition}]
	First observe that the function in Equation~\eqref{eq:SingleNodeWeightEdge} is deterministic for any edge $e\in \E$, i.e., the indicator functions are based on properties of the input graph $\G$, and clearly $\G$ is fixed.
	
	Taking the expectation of $\estFunc_i(e)$ we have 
	\begin{align*}
		\expectation[f_j(e)] &\stackrel{(A)}{=}\sum_{e\in \E} \frac{\expectation[X_e]}{p} \frac{1}{|V_j|}\sum_{v\in V_j} \frac{a_q(v,e)}{|\wedgNodeSetGeneral{v}|}\\
		&\stackrel{(B)}{=}  \frac{1}{|V_j|} \sum_{v\in V_j} \frac{1}{{|\wedgNodeSetGeneral{v}|}} \sum_{e\in \E}  a_q(v,e)\\
		&\stackrel{(C)}{=}  \frac{1}{|V_j|} \sum_{v\in V_i} \clustMetric_v  =\frac{1}{|V_j|}\sum_{v\in V_j}\clustMetric_v = \clustGenPartition_j. \\
	\end{align*}
	Where $(A)$ comes from the linearity of expectation, $(B)$ comes from  $\expectation[X_e]=p$, by swapping the summations, and the fact that $1/|\wedgNodeSetGeneral{v}|$ does not depend on $e\in \E$ for any of nodes $v\in \V_i, i=1,\dots,k$, and $(C)$ by noting that 
	\[
	\sum_{e\in \E}  a_q(v,e) =  \sum_{e \in E : v\in e} q |\Delta_e| + (1-2q)\sum_{e \in \E  : v\in \neighborEdge{e}}1 = |\triSetNode{v}|
	\]
	from the proof of Lemma~\ref{lemma:unbiasednessSingleNode}, concluding therefore the proof as $|\triSetNode{v}|/|\wedgNodeSetGeneral{v}| = \clustMetric_v$.
\end{proof}

\begin{proof}[Proof of~\Cref{lemma:unbiasedOutput}]
	From Lemma~\ref{lemma:unbiasedEdgePartition} it holds that for each edge $e\in E$ and 
	each function $\estFunc_j(e)$, % for $j\in[k]$, 
	it is $\expectation[\estFunc_j(e)] = \clustGenPartition_j$, \edit{where $\expectation[\cdot]$ is over a random edge $e\in E$}.
	After noting that \edit{by definition} $\estFunc_j = \tfrac{1}{s}\sum_{e\in \mathcal{S}} \estFunc_j (e)$,
	the result follows by the linearity of expectation \edit{and $|\sampleSet| =s$}.
\end{proof}

\begin{proof}[Proof of Lemma~\ref{lemma:varianceBoundOutput}]
	First note that $\estFunc_j$ in output to Algorithm~\ref{alg:adaptiveBucketSampling} is the sample average over $s$ \emph{i.i.d.}\ random variables 
	\begin{equation}\label{eq:VarProofSingleTerm}
		\estFunc_j(e) = \sum_{e\in \E} \frac{X_e}{p} \frac{1}{|V_i|}\sum_{v\in V_i} \frac{a_q(v,e)}{|\wedgNodeSetGeneral{v}|}.
	\end{equation}
	Therefore to bound the variance, note
	\begin{equation}\label{eq:VarProofCombiningRes}
		\Var\left[\tfrac{1}{s}\sum_{i=1}^s \estFunc_j(e)\right] = \tfrac{1}{s^2} \sum_i \Var[\estFunc_j(e)] \le \frac{\hat{\sigma}^2}{s}.
	\end{equation}
	Where $\hat{\sigma}^2$ is a bound on the variance on $\estFunc_j(e)$,  we used the fact that the variables are i.i.d.\ and that $\Var[cX] = c^2\Var[X]$ for a random variable $X$ and a constant $c$.
	We now bound the variance of $\estFunc_j(e)$. Recall that
	\begin{equation}\label{eq:varianceBySecondMoment}
		\Var[X] = \expectation[X^2] - (\expectation[X])^2,
	\end{equation}
	hence,
	\begin{align*}
		&\expectation[\estFunc_j(e)^2]  =\\
		&= \expectation\left[ \sum_{e_1\in \E} \sum_{e_2 \in \E} \frac{X_{e_1}X_{e_2}}{p^2} \frac{1}{|V_j|}\sum_{v_1\in V_j} \frac{a_q(v_1,e_1)}{|\wedgNodeSetGeneral{v_1}|} \frac{1}{|V_j|}\sum_{v_2\in V_j} \frac{a_q(v_2,e_2)}{|\wedgNodeSetGeneral{v_2}|} \right]\\
		&\stackrel{(A)}{=}\sum_{e_1\in \E} \sum_{e_2 \in \E} \frac{\expectation[X_{e_1}X_{e_2}]}{p^2} \frac{1}{|V_j|}\sum_{v_1\in V_j} \frac{a_q(v_1,e_1)}{|\wedgNodeSetGeneral{v_1}|} \frac{1}{|V_j|}\sum_{v_2\in V_j} \frac{a_q(v_2,e_2)}{|\wedgNodeSetGeneral{v_2}|}\\
		&\stackrel{(B)}{\le}  \frac{1}{p} \frac{1}{|V_j|^2} \sum_{e_1\in \E}  \sum_{v_1\in V_j} \frac{a_q(v_1,e_1)}{|\wedgNodeSetGeneral{v_1}|} \sum_{v_2\in V_j} \sum_{e_2 \in \E}  \frac{a_q(v_2,e_2)}{|\wedgNodeSetGeneral{v_2}|} \\
		&\stackrel{(C)}{=}  \frac{1}{p} \frac{1}{|\V_j|^2} \sum_{e_1\in \E}  \sum_{v_1\in \V_j} \frac{a_q(v_1,e_1)}{|\wedgNodeSetGeneral{v_1}|} \sum_{v_2\in \V_j} \clustMetric_{v_2} = \frac{1}{p}  \left(\frac{1}{|\V_j|}  \sum_{v\in \V_j} \clustMetric_{v}\right)^2.
	\end{align*}
	Step $(A)$ follows from the linearity of expectation, step $(B)$ from the fact that $\expectation[X_{e_1}X_{e_2}] \le \expectation[X_{e_1}] = p$ noting that such random variables are 0--1 variables and for the events $I \doteq$``$X_{e_1}=1$'', and $I' \doteq \text{``}X_{e_2}=1$'' it holds $I\subseteq I'$, and $(C)$ follows from the same argument used in the proof of Lemma~\ref{lemma:unbiasedEdgePartition}. The final claim follows by combining the above result with Equations~\eqref{eq:varianceBySecondMoment} and~\eqref{eq:VarProofCombiningRes}, and the fact that $\estFunc_j(e)$ is an unbiased estimator of the triadic measure over each partition $j\in[k]$ (Lemma~\ref{lemma:unbiasedEdgePartition}).
\end{proof}

\begin{proof}[Proof of Theorem~\ref{theo:PdimBound}]
	We now show that $\PD(\Fset) = \VC(\DomH\times[a,b], \Fset^+)$ is bounded by $\floor{\log_2 \bucketColorSup} + 1$. 
	Let $\mathsf{VC}(\mathcal{H}\times[a,b], \mathcal{F}^+)=\pdim$ be the size of the largest set of elements of the space $\mathcal{H}\times[a,b]$ that can be shattered by the range-set $\mathcal{F}^+$, and let $Q$ be the shattered set attaining such size, that is $0<|Q|=\pdim$. Without loss of generality $Q$ contains one element $a=(e,x)$ for some $e\in E$ and $x>0$ (note that $x$ needs to be strictly greater than 0 otherwise the set $Q$ cannot be shattered, see Lemma~\ref{lemma:RVNozeroElements}), and furthermore note that there cannot be another element in $Q$ of the form $a=(e,y)$ for $y\neq x$ (see Lemma~\ref{lemma:RVatmostoneEdge}). There exist therefore $2^{\pdim-1}$ distinct and non-empty subsets of $Q$ containing the element $a$ that we label as $Q_1,\dots,Q_{2^{\pdim-1}}$, since the set $Q$ is shattered by $\mathcal{F}^+$. This implies therefore that there are $R_1,\dots R_{2^{\pdim-1}}$ ranges for some partitions $V_j, j =1\dots, k$ for which it holds $Q_i = Q \cap R_i$ with $i=1,\dots,2^{\pdim-1}$, by the definition of $Q$.
	Given an edge $e$ from $\DomH=\E$, an element  of the form $(e,x)$ can belong to at most $\bucketColorSup$ distinct ranges, which is easy to see by the following argument: (1) for a function $f_j(e)$ of a partition $j=1,\dots,k$ to be non-zero it should hold that there exists a node $v\in V_j$ such that $v \in \singleTri, \singleTri\in\triSetEdge{e}$ and (2) for an edge $e=\{u,v\}$ it holds $|\triSetEdge{e}| \le \min\{\nodeDeg{u},\nodeDeg{v}\}$ and further $\triSetEdge{e} = \{\{u,v\} \cup \{w\} : (\neighborNode{u}\setminus\{v\})\cap(\neighborNode{v}\setminus\{u\}) \})$ implying that all the nodes forming a triangle with edge $e=\{u,v\}$ should be in the neighborhood of the node with minimum degree among $u$ and $v$. Therefore for a fixed edge $e=\{u,v\}$ at most $\bucketColor_z$ where $z= \arg\min\{\nodeDeg{u},\nodeDeg{v}\}$ distinct functions $f_j(e)$ can have non-zero values. Hence the element $a=(e,x) \in Q$  can be contained in at most $\bucketColorSup$ ranges, therefore
	\[
	2^{\pdim-1}\le \bucketColorSup \implies \pdim \le \floor{\log_2 \bucketColorSup} + 1.
	\]
\end{proof}

\begin{proof}[Proof of Corollary~\ref{coroll:pdimQvals}]
	First we consider the case where $q=0$. By a similar argument to the proof of Theorem~\ref{theo:PdimBound} consider an element of the form $e=(u,x)$ in the shattered set $Q$ such that $|Q|=\pdim$, then the two nodes forming the edge $e=\{u,v\}$ do not contribute to the functions $f_j(e)$ for the partitions they belong to, that is $2^{\pdim-1} \le \bucketColorSup'$, taking the logarithm yields the claim. While for $q=1/2$ it can only be that at most two partitions (if $k\ge 2$) have $f_j(e) > 0$ and hence $\bucketColorSup' \le 2$ implying that $\pdim \le 2$.
\end{proof}

\begin{proof}[Proof of~\Cref{coroll:pdimLima}]
	Consider a star graph $G=(V,E)$ over $n$ nodes, then in such case $\bucketColorSup = 1$ as \emph{for each} edge $e=\{u,v\}\in \E$ it holds that $ \min\{\nodeDeg{u},\nodeDeg{v}\} = 1$, hence \edit{by Proposition~\ref{theo:PdimBound}: $\zeta \le \floor{\log_2 1} + 1 = 1$.}
\end{proof}

\begin{proof}[Proof of~\Cref{theo:probGuarantees}]
	\edit{To prove the claim we bound  the probability of failure of~\mainAlg, that is let $F$ be the event that~\mainAlg's guarantees do not hold. We now show $\Prob[F]\le \eta$. Clearly if this is the case, then~\mainAlg succeeds with probability at least $1-\eta$.}
	
	Let us define the following events:
	\begin{itemize}
		\item $F_1=s\ge s_{\max}$ and there exists $j\in [k]$ such that $|f_j-\clustGenPartition_j| > \varepsilon_j$; and
		\item $F_2=s< s_{\max}$, $\widehat{\varepsilon}_j\le \varepsilon_j$, for each $j\in [k]$, and there exists $j\in [k]$ such that $|f_j-\clustGenPartition_j| > \varepsilon_j$.
	\end{itemize}
	These are the two conditions under which Algorithm~\ref{alg:adaptiveBucketSampling} fails to report sufficiently accurate estimate\edit{s}, and under which its guarantees do not hold for the confidence intervals built around $f_j$, for $j\in[k]$ \edit{as from the statement}. By the combination of Theorem~\ref{theo:PdimBound} and~\ref{theo:boundSSPDIM} 
	%for $s_{\max}$ 
	then for~\mainAlg  it holds that $\Prob[F_1] \le \eta/2$, next consider $\widehat{\varepsilon}_j$ to be computed through Theorem~\ref{theo:empiricalBern} at each iteration $i$ using $\eta'=\eta/2^{i+1}$ then by a union bound over the possible (infinite) iterations~\cite{Borassi2019Kadabra,Pellegrina2023Silvan} it holds that $\Prob[F_2]\le \eta/2$. Hence, the failure probability of~\mainAlg $\Prob[F_1\cup F_2] \le \Prob[F_1] + \Prob[F_2] \le \eta/2 + \eta/2 \le \eta$, yielding the claim.
\end{proof}

\begin{proof}[Proof of Lemma~\ref{lemma:unbiasedVarEst}]
	Clearly $\frac{1}{c-1} \sum_{e\in\sampleSet} (f_j(e)-f_j)^2$ is an unbiased estimator of the variance of each function $f_j$ for $|\sampleSet|=c$, then 
	notice that $f_j = \frac{m}{c|V_j|}(\sum_{e}a_j(e) + q\sum_e b_j(e))$ and therefore
	$(f_j(e) - f_j)^2 = \frac{m^2}{|V_j|^2}[\hat{a}_j(e)^2 +q\hat{a}_j(e)\hat{b}_j(e)+ q^2\hat{b}_j(e)^2 ]$, hence 
	\begin{align*}
		\widehat{V}_q(f_j) &= \frac{m^2}{(c-1)|V_j|^2} \left(\sum_{e\in \sampleSet}A_j + qB_j + C_j q^2\right)\\
		&= \frac{m^2}{(c-1)|V_j|^2} \left(\sum_{e\in \sampleSet}\hat{a}_j(e)^2 + q \sum_{e\in \sampleSet}\hat{a}_j(e)\hat{b}_j(e) + q^2\sum_{e\in \sampleSet}\hat{b}_j(e)^2 \right)\\
		&=  \frac{1}{(c-1)}\sum_{e\in \sampleSet} (\estFunc_j(e) - \estFunc_j)^2
	\end{align*}
\end{proof}

\begin{proof}[Proof of Lemma~\ref{lemma:varEpsilon}]
	First note that $\varEst_q(\estFunc_j)$ is an unbiased estimate of $\Var[f_j]$, additionally given that $\varEst_q(\estFunc_j)$ has bounded range, i.e., it is a finite sum over bounded terms. Hence, it exists a value $B_q$ for each value of $q$ such that $\varEst_q(\estFunc_j) \le B_q$ therefore by Hoeffding's inequality (i.e., Theorem~\ref{th:Hoeffding}) it holds $\varepsilon'=B_q\sqrt{\frac{\log(2k/\eta')}{2|\sampleSet'|}}$, for any bound on the error probability $\eta'$.
\end{proof}

\begin{proof}[Proof of Lemma~\ref{lemma:upperBounds}]
	Recall that for each partition $j=1,\dots, k$, and each edge $e\in \E$ it holds that 
	\[
	\estFunc_j(e) = \frac{1}{p} \frac{1}{|V_i|}\sum_{v\in V_i} \frac{a_q(v,e)}{|\wedgNodeSetGeneral{v}|}.
	\]
	First notice that for a fixed edge $e=\{u,v\}\in \E$ the only nodes that can form a triangle with edge $e\in \E$ are those nodes such that $w \in (\neighborNode{u} \cap \neighborNode{v})$ for which it clearly holds that $(\neighborNode{u} \cap \neighborNode{v}) \subseteq \neighborNode{z}$ where $z=\arg\min\{\nodeDeg{u}, \nodeDeg{v}\}$, and further $|\triSetEdge{e}| \le \nodeDeg{z}$. Then recall that by definition $a_q(v,e) =  q|\triSetEdge{e}|\mathbf{1}[v\in e] + (1-2q)\mathbf{1}[v\in \neighborEdge{e}]$, hence we have 
	$q|\triSetEdge{e}|\mathbf{1}[w\in e] \le qd_z$ for both the two nodes $u,v \in e$ and
	\[
	(1-2q)\mathbf{1}[w\in \neighborEdge{e}] \le (1-2q) \text{ for }w\in \neighborNode{z}.
	\]
	Therefore for each partition $V_j$ we have that
	\[
	f_j(e) \le \frac{1}{p} \frac{1}{|V_i|} \left(\sum_{v\in V_i, v\in e} \frac{q d_z}{|\wedgNodeSetGeneral{v}|} + \sum_{v\in V_i, v\in \neighborNode{z}} \frac{(1-2q)}{|\wedgNodeSetGeneral{v}|}\right).
	\]
	which follows by the fact that $f_j(e)$ is the sum of non-negative terms, clearly $f_j(e) \le \max_e \{f_j(e)\} = R_j$ as computed by the algorithm, the final result follows from the definition of $p=\sfrac{1}{m}$.
	Finally the fact that $\floor{\log_2 \bucketColorSup}+1$ is an upper bound to the pseudo-dimension associated to the functions $f_j$, for $j\in[k]$, follows from Proposition~\ref{theo:PdimBound}
\end{proof}

\section{Auxiliary theoretical results}\label{appsec:auxTheo}
\begin{lemma}
	[\cite{Chiba1985}]\label{lemma:ChibaN}
	Given a graph $\G=(\V,\E)$ with $n$ nodes and $m$ edges, then
	\begin{equation*}
		\sum_{\{u,v\}\in E} \min\{\nodeDeg{u},\nodeDeg{v}\} \le 2 \arb m,
	\end{equation*}
	where $\arb$ corresponds to the \emph{arboricity} of the graph $G$, i.e., the minimum number of edge disjoint spanning-forests required in which $G$ can be decomposed.
\end{lemma}

\begin{lemma}[\cite{Riondato2018MisoSoup}]\label{lemma:RVatmostoneEdge}
	If $B\subseteq\DomH\times [a,b]$ is shattered by $\Fset^+$, it may contain \emph{at most one} element $(e,x)\in \DomH\times [a,b]$ for each $e\in \DomH$.
\end{lemma}
\begin{lemma}[\cite{Riondato2018MisoSoup}]\label{lemma:RVNozeroElements}
	If $B\subseteq\DomH\times [a,b]$ is shattered by $\Fset^+$, it cannot contain \emph{any} element of the form $(e,a)\in \DomH\times [a,b]$ for any $e\in \DomH$.
\end{lemma}

\begin{algorithm}[t]
	\caption{\texttt{FindThreshold}}\label{alg:findbeta}
	\KwIn{$\G=(\V,\E)$ with $|\V|=n$.}
	\KwOut{Value of the small-node degree threshold $\beta$.}
	$D=(0,\dots,0)\in \mathbb{R}^{n-2}; \beta \gets 2; T\gets 0$\;
	\lFor{$v\in V$}{
		$D_{\nodeDeg{v}}\gets D_{\nodeDeg{v}}+1$}
	\While{$T < c |V|$}{
		$T \gets \beta^2 D_{\beta}$\;
		$\beta \gets \beta + 1$\;
	}
	\KwRet{$\beta-1$}\;
	
\end{algorithm}

\begin{theorem}[Hoeffding bound~\citep{Mitzenmacher2017PC}]\label{th:Hoeffding}
	Let $X_1,\dots,X_s$ be $s$ independent random variables such that for all $i=1,\dots,s$, $\expectation[X_i] = \mu$ and $\Prob[X_i\in [a,b]]=1$. Then,
	\[
	\Prob\left(\left|\frac{1}{s}\sum_{i=1}^s X_i - \mu \right| \ge \varepsilon\right) \le 2\exp(-2s\varepsilon^2/(b-a)^2).
	\]
\end{theorem}

\section{Missing subroutines}

\subsection{\edit{Upper-bounds computation}}\label{appsubsec:upperbounds}
\begin{algorithm}[t]
	\caption{$\mathtt{UpperBounds}$}\label{alg:upperBondRange}
	\KwIn{$G=(V,E), \nodesetPartition, |\wedgNodeSetGeneral{v}|_{v\in V}$, $q$.}
	\KwOut{Upper bounds: $\zeta, R_1,\dots,R_k$.}
	$\mathbf{r}_v \gets (0,\dots,0) \in \mathbb{R}^{k}$; $\bucketColor_v \gets 0$ for each $v\in V$\label{alglineRange:initVectors}\;
	\For{$v\in V$\label{alglineRange:forNodes}}{
		\ForEach{$w\in \mathcal{N}_v$\label{alglineRange:forNodesVects}}{
			$\mathbf{r}_v[j] \gets \mathbf{r}_v[j] + \frac{(1-2q)}{|\wedgNodeSetGeneral{w}|}$ such that $w\in V_j$\label{alglineRange:updateVectors}\;
		}
		$\bucketColor_v \gets |\{j : w\in V_j \text{ and } w\in (\neighborNode{v}\cup \{v\})\}|$\label{alglineRange:ubCols}\;
	}
	$R_i \gets 0$ for each $i=1\dots,k$; $\bucketColorSup \gets 0$\label{alglineRange:initRanges}\;
	\ForEach{$e=(u,v) \in E$\label{alglineRange:forEdges}}{
		$z\gets \arg\min \{d_u,d_v\}$\label{alglineRange:minDeg}\;
		\ForEach{$r_j \in \mathbf{r}_z$\label{alglineRange:forVect}}{
			$\widehat{R}_j \gets r_j  + q \tfrac{\nodeDeg{z}}{|\wedgNodeSetGeneral{u}|} \mathbf{1}[u\in V_j] +  q\tfrac{\nodeDeg{z}}{|\wedgNodeSetGeneral{v}|} \mathbf{1}[v\in V_j] $\label{alglineRange:ubPart}\;
			$R_j \gets \max\left\{R_j, \widehat{R}_j\frac{m}{|V_j|}\right\}$\label{alglineRange:ubRange}\;
		}
		$\bucketColorSup \gets \max\{\bucketColorSup, \zeta_z\}$\label{alglineRange:ubPdim}\;
	}
	\KwRet{$\floor{\log_2 \bucketColorSup}+1, R_1,\dots,R_k$\label{alglineRange:return}}\;
\end{algorithm}

Algorithm~\ref{alg:upperBondRange} details the procedure used to compute an upper bound to the range of functions $f_j(e)$ for each of the possible sampled edges $e\in \E$ and partitions $j\in[k]$. 
The main idea is to first iterate through all the nodes maintaining a vector of size $k$ 
(lines \ref{alglineRange:initVectors}-\ref{alglineRange:updateVectors}). 
Each entry in a vector is then used to compute an upper bound on the range of $f_j(e)$, 
for $j\in[k]$, by iterating all edges $e\in \E$ (lines~\ref{alglineRange:initRanges}-\ref{alglineRange:ubRange}). 
For each edge $e\in \E$ we consider the vector corresponding to the node with minimum degree over $e\in \E$ 
to obtain the upper bounds $R_j$, for $j\in [k]$ (line~\ref{alglineRange:minDeg}). 

\emph{Time complexity.} Next we analyze the time complexity of Algorithm~\ref{alg:upperBondRange},
showing that in practice it is bounded by $\bigO(m)$. 
First, computing the vectors $\mathbf{r}_v$, for $v\in V$, requires at most $\bigO(m)$ by the handshaking lemma. 
Next, the algorithm iterates the vectors $\mathbf{r}_z, z\in V$ for all $z$ corresponding to the minimum degree nodes over all edges $e\in \E$. 
By storing only the non-zero entries over vectors $\mathbf{r}_z, z\in V$, the total non-zero entries are bounded by $\bigO(k)$ and by $\bigO(d_z)$. 
Consider the bound $\bigO(k)$ then the time complexity of  Algorithm~\ref{alg:upperBondRange} is bounded by $\bigO(km)$, otherwise by
\ifextended
Lemma~\ref{lemma:ChibaN}, 
\else
a known result~\citep{Chiba1985} (see our extended version),
\fi
the bound is $\bigO(\arb m)$, where $\arb$ is the arboricity of $\G$. 
For practical values of $k$, i.e., when $k$ is a large constant, the above bound is linear in $m$. Also note that if $k\ge d_{\max}$ where $d_{\max}$ is the maximum degree of a node in $\G$ then we can avoid iterating all entries in $\mathbf{r}_z, z\in V$, and use a tight deterministic upper bound on each $f_j$ reducing the time complexity $\bigO(\arb m)$ to $\bigO(m)$, yielding linear complexity again \edit{(see~\Cref{subsec:pdimBound})}.

\subsection{Variance estimation}\label{app:varianceEst}
\begin{algorithm}[t]
	\caption{$\mathtt{Fixq}$}\label{alg:fixQ}
	\KwIn{$G=(V,E), \nodesetPartition, |\wedgNodeSetGeneral{v}|_{v\in V}$, $q$, sample size $c$.}
	\KwOut{Terms $A_j,B_j,C_j, j\in[k]$.}
	$\mathcal{S} \gets \mathtt{UniformSample}(E,c)$\tcp*{$|\mathcal{S}|=c$}
	\For{$e\in \mathcal{S}$}{
		$a_j(e) \gets \sum_{v\in V_j, v\in \neighborEdge{e}} 1/|\wedgNodeSetGeneral{v}|$\;
		$b_j(e) \gets \sum_{v\in V_j, v\in e}|\triSetEdge{e}|/|\wedgNodeSetGeneral{v}| + \sum_{v\in V_j, v\in \neighborEdge{e}} 1/|\wedgNodeSetGeneral{v}|$\;
		$\bar{a} \gets \bar{a} + a_j(e)$\;
		$\bar{b} \gets \bar{b} + b_j(e)$\;
	}
	\For{$e\in \mathcal{S}$}{
		$\hat{a}_j(e) \gets a_j(e) -\bar{a}/c$\;
		$\hat{b}_j(e) \gets b_j(e) -\bar{b}/c$\;
	}
	\For{$j\gets 1$ to $k$}{
		$A_j \gets \sum_{e\in \mathcal{S}} \hat{a}_j(e)^2$\; 
		$B_j \gets \sum_{e\in \mathcal{S}} 2\hat{a}_j(e)\hat{b}_j(e)$\;
		$C_j \gets \sum_{e\in \mathcal{S}} \hat{b}_j(e)^2$\;
	}
	\KwRet{$A_j,B_j,C_j, j\in[k]$;}
\end{algorithm}
The procedure that we use to estimate the variance of the random variables $f_j$ is detailed in Algorithm~\ref{alg:fixQ}.

\subsection{Filtering small-degree nodes}\label{app:filterSmall}
In this section we briefly show that filtering small degree nodes, as described in Section~\ref{subsec:filtersmall}, yields no estimation error for~\mainAlg.
\begin{lemma}\label{lemma:triangleDecomp}
Let $v\in \V$ and $|\triSetNode{v}^L|$ be the counts returned for node $v$ by Algorithm~\ref{alg:preprocess}, and $\G'$ be the filtered graph of the input graph $\G$. Then $|\triSetNode{v}| = |\triSetNode{v}^L| + |\triSetNode{v}'|$, where $\triSetNode{v}'$ (resp. $\triSetNode{v}$) is the set of triangles containing node $v\in \V'$ in $\G'$ (resp. $\G$).
\end{lemma}
\begin{proof}
Pick an arbitrary node $v\in \V'$ and let $\triSetEdge{v}$ be the set of triangles containing node $v\in \V'$ in $\G = (\V,\E)$. Then the set $\triSetEdge{v} = \{\singleTri_1, \dots, \singleTri_{|\triSetEdge{v}|}\}$ can be partitioned as follows, let $\triSetNode{v}^L = \{\singleTri \in \triSetNode{v}: \text{it exists } w\in (\V\setminus \V') \text{ and } w \in \delta\}$ corresponding to the set of triangles containing at least one node removed from $\G$ by Algorithm~\ref{alg:preprocess} as from the statement, and $\triSetNode{v}' = \triSetNode{v} \setminus \triSetNode{v}^L$. Therefore it follows that $\triSetNode{v} = \triSetNode{v}^L \cup \triSetNode{v}'$, yielding the claim, as $\triSetNode{v}^L \cap \triSetNode{v}' = \emptyset$.
\end{proof}
Lemma~\ref{lemma:triangleDecomp} allows us to safely remove small degree nodes while focusing on estimating the number of remaining coefficients $\clustGenPartition_j$ in the filtered graph $\G'$. The next lemma shows that if we obtain function $f_j$
with respect to some slightly modified scores $\clustGenPartition'_j$ we can recover $|f_j-\clustGenPartition_j| \le \varepsilon_j$, for $j\in [k]$ on the original graph $\G$, i.e., the desired output.

\begin{corollary}\label{corollary:approxMetricDecomp}
Let $\clustMetric_v' = \frac{|\triSetNode{v}'|}{|\wedgNodeSetGeneral{v}|}$ where $\wedgNodeSetGeneral{v}$ refers to the graph $\G$ (i.e., not preprocessed) and $\clustGenPartition_j' = 1/|\V_j| \sum_{v\in \V_j} \clustMetric_v'$ then if $f_j'$ is an estimate such that 
$
|f_j'-\clustGenPartition_j' | \le \varepsilon_j$, for $j\in [k]$ in $\G'$
then $f_j = f_j'+\frac{1}{|\V_j|} \sum_{v\in \V_j}\frac{|\triSetNode{v}^L|}{|\wedgNodeSetGeneral{v}|}$ is such that $|f_j-\clustGenPartition_j| \le \varepsilon_j$, for $j\in [k]$ in $\G$. 

\end{corollary}
\begin{algorithm}[t]
	\caption{Filter small-degree nodes}\label{alg:preprocess}
	\KwIn{$G=(V,E)$.}
	\KwOut{Filtered graph $G'=(V',E')$, counts $|\triSetNode{v}^L|$ for $v\in V$.}
	$V'\gets V$\;
	\ForEach{$v\in V' : \nodeDeg{v} =1$}{
		$V'\gets V'\setminus \{v\}$\;	
	}
	$E' \gets \mathtt{Update}(E, V')$\;
	$\beta \gets \mathtt{FindThreshold}(G)$ \tcp*{$\Theta(n)$}
	\ForEach{$v\in V' : 2\le d(v) \le \beta$}{
		\ForEach{$\singleTri \in \triSetNode{v}$}{
			\lForEach{$w\in \singleTri$}{
				$|\triSetNode{w}^L| \gets |\triSetNode{w}^L| +1$}
		}
		$\V' \gets \V' \setminus{v}$\;
		$\E' \gets  \mathtt{Update}(\E, \V')$\;
	}
	\KwRet{$\G'=(\V',\E'), |\triSetNode{v}^L|$ for $v\in\V$}\;
	
\end{algorithm}

\begin{proof}
The proof directly follows from the definition of $\clustMetric_v$ which by Lemma~\ref{lemma:triangleDecomp} can be written as $\clustMetric_v = \frac{|\triSetNode{v}^L|+ |\triSetNode{v}'|}{|\wedgNodeSetGeneral{v}|}$.
Hence, $f_j = \frac{1}{|\V_j|}\sum_{v\in \V_j} \clustMetric_v = \frac{1}{|\V_j|}\sum_{v\in \V_j} \frac{|\triSetNode{v}'|}{|\wedgNodeSetGeneral{v}|} + \frac{1}{|\V_j|}\sum_{v\in \V_j} \frac{|\triSetNode{v}^L|}{|\wedgNodeSetGeneral{v}|} $ hence by the definition of $f_j'$, and the fact that there is no error over the second term, the claim follows.
\end{proof}

The above result captures the fact that we can execute~\mainAlg on the processed graph $\G'$, by only adopting $\wedgNodeSetGeneral{v}$, for $v\in V$, from the original graph $\G$ without any loss on the guarantees over $\G$, with negligible overhead (i.e., only $\bigO(n)$ total additional work).

The exact procedure for  filtering small degree nodes is detailed in Algorithm~\ref{alg:preprocess}.

\section{Baselines}\label{appsec:baselines}

\begin{algorithm}[t]
	\caption{\texttt{\baseFull}}\label{alg:fullWS}
	\SetKwComment{Comment}{$\triangleright$\ }{}
	\KwIn{$\G=(\V,\E)$, partitions $V_i$, $\varepsilon_i, i=1,\dots,k, \eta$, $\clustMetric$.}
	\KwOut{$f(V_i) \in\Psi(V_i) \pm \varepsilon_i$ w.p.\ $\ge 1-\eta$.}
	$\widehat{\clustMetric}_{v} \gets 0, v\in \V$; $s\gets \ceil{\tfrac{1}{2\varepsilon_i^2} \log(2k/\eta)}$\;
		\ForEach{$i \gets 1$ to  $k$}{
			$\mathcal{S} \gets \emptyset$\;
				\For{$j \gets 1$ to $s$}{
					$v\gets \mathtt{Uniform}(V_i)$\;
					\If{$\clustMetric = \clustCoeff$}{
						$\mathcal{S}\gets \mathcal{S} \cup \mathtt{SampleWedge}(v)$\;}
					\ElseIf{$\clustMetric=\closureCoeff$}{
						$\mathcal{S}\gets \mathcal{S} \cup \mathtt{Sample2Path}(v)$\;
					}
				}
				\ForEach{$\mathtt{w} \in \mathcal{S} : ``\mathtt{w}=\mathtt{closed}$''}
				{
					$\estFunc_i \gets \estFunc_i + \tfrac{1}{s}$\;
			}}
		\KwRet{$\estFunc_i, i=1,\dots,k$}\;
	\end{algorithm}

	\begin{algorithm}[t]
		\caption{\texttt{SampleWedge}~\citep{Kolda2014, Seshadhri2014}}\label{alg:WS}
		\SetKwComment{Comment}{$\triangleright$\ }{}
		\KwIn{$\G=(\V,\E), v\in \V$.}
		\KwOut{A wedge sampled uniformly at random from $\wedgNodeSetCenter{v}$.}
			$u_1 \gets \mathtt{Uniform}(\neighborNode{v})$\;
			$u_2 \gets \mathtt{Uniform}(\neighborNode{v}\setminus\{u_1\})$\;
			$\wedg \gets \langle u_1,u_2\rangle$\;
		\KwRet{$\wedg$}\;
	\end{algorithm}
	
	\begin{algorithm}[t]
		\caption{\texttt{Sample2Path}}\label{alg:2PS}
		\SetKwComment{Comment}{$\triangleright$\ }{}
		\KwIn{$\G=(\V,\E), v\in \V$.}
		\KwOut{A 2-path sampled uniformly at random from $\wedgNodeSetHead{v}$.}
		$p_{u_1}, \dots p_{u_{|\nodeDeg{v}|}} \gets \tfrac{\nodeDeg{u_1}-1}{\sum_{u\in \neighborNode{v}} (\nodeDeg{u}-1)}, \dots, \tfrac{\nodeDeg{u_{|\nodeDeg{v}|}} -1}{ \sum_{u\in \neighborNode{v}} (\nodeDeg{u}-1) }, u_i \in \neighborNode{v}$\;
			$u_1 \gets \mathtt{Sample}(\neighborNode{v}, p_i)$\;
			$u_2 \gets \mathtt{Uniform}(\neighborNode{u_1}\setminus\{v\})$\;
			$\wedg \gets \langle u_1,u_2\rangle$\;
		\KwRet{$\wedg$}\;
	\end{algorithm}
	
The baseline considered is based on the state-of-the-art approach devised for estimating local clustering coefficients by sampling~\citep{Seshadhri2014}. We denoted such algorithm by~\baseFull. 
We present such baseline in Algorithm~\ref{alg:fullWS}.
Note that the~\baseFull exactly coincides with the approach of~\citet{Seshadhri2014} when the considered metric is the clustering coefficient, in particular,
such algorithm proceeds as follows: consider each partition independently $\V_i$,  $i=1,\dots,k$:  
($i$)~uniformly sample a random node $v$ from each partition $\V_i$, $i=1,\dots,k$; 
($ii$)~for the sampled node, uniformly sample a wedge centered at the node and update an estimate if the wedge is closed, i.e., it forms a triangle; 
($iii$)~repeat such steps until concentration (i.e., for $\bigO(\varepsilon_i^{-2} \log(k/\eta))$).  
Such approach is based on the fact that a random sampled wedge from a node $v\in V_i$ 
provides an unbiased estimate of $\clustCoeff_v$. 
Inspired by such approach we designed a non trivial baseline for estimating the closure coefficient $\closureCoeff_v$ for a node $v\in \V$, which is based on uniformly sampling a wedge \emph{headed} at $v\in \V$, that is we uniformly sample $\wedg$ from $\wedgNodeSetHead{v}$, which requires particular attention to avoid bias in the final sampling procedure, in particular we have the following.

\begin{lemma}\label{lemma:unbiasedBaseline}
	Each wedge $\wedg \in \wedgNodeSetHead{v}$ in output to Algorithm~\ref{alg:2PS} is sampled uniformly at random from $\wedgNodeSetHead{v}$ for $v\in V$, and the probability that a sampled path is closed, i.e., forms a triangle with $v\in\V$ corresponds to $\closureCoeff_v$.
\end{lemma}
\begin{proof}
	Let us compute the probability to sample a path $\wedg = \langle u_1, u_2 \rangle \in \wedgNodeSetHead{v}$ of length-2 starting from a node $v\in \V$,
	\begin{align*}
	&\mathbb{P}[\wedg= \langle u_1, u_2 \rangle \text{ is sampled} ] = \mathbb{P}[u_1 \text{ is 1-st node}] \mathbb{P}[u_2 \text{ is 2-nd node}] \\
	&= p_{u_1} \cdot p_{u_2} = \tfrac{\nodeDeg{u_1}-1}{\sum_{u\in \neighborNode{v}} (\nodeDeg{u}-1)} \cdot \tfrac{1}{\nodeDeg{u_1}-1} =  \tfrac{1}{\sum_{u\in \neighborNode{v}} (\nodeDeg{u}-1)}
	\end{align*}
	Therefore the 2-path $\wedg = \langle u_1, u_2 \rangle$ is sampled uniformly among all the $\sum_{u\in \neighborNode{v}} (\nodeDeg{u}-1)$ paths of length 2 headed at $v\in \V$. To compute the probability that such path is closed, notice that a path can be closed only if $\{u_1,u_2,v\}\in \triSet$, that is if the three nodes form a triangle, hence
	\[
	\mathbb{P}[\wedg= \langle u_1, u_2 \rangle \text{ is closed} ] = \tfrac{2|\triSetNode{v}|}{\sum_{u\in \neighborNode{v}} (\nodeDeg{u}-1)} = \closureCoeff_v.
	\] 
	where we used the fact that each triangle contributes to two distinct wedges that can be sampled by Algorithm~\ref{alg:2PS}.
\end{proof}

\begin{corollary}
	The output $f_i$ of Algorithm~\ref{alg:fullWS} is unbiased and it holds that $\Prob[|f_i-\clustGenPartition_i| \ge \varepsilon_i] \le \eta$ simultaneously for all $i=1,\dots,k$.
\end{corollary}
\begin{proof}
	When Algorithm~\ref{alg:WS} is executed for the clustering coefficient $\clustCoeff$, 
	this immediately follows from the results of~\citet{Seshadhri2014}.
	While when Algorithm~\ref{alg:WS} is executed for the closure coefficient $\closureCoeff$ we e can write $f_i$ as,
	\[
	\estFunc_i = \frac{1}{s}\sum_{i=1}^s \sum_{v\in V} X_\wedg^v 
	\]
	where $X_\wedg^v$ is an indicator random variable if the wedge $\wedg$ headed at $v\in \V$sampled by the algorithm  is closed, clearly $\Prob[X_\wedg^v = 1] = \frac{\closureCoeff_v}{|V_i|}$ as it is the results of the event that node $v\in \V_i$ is sampled (which occurs with probability $1/|\V_i|$) and that a closed wedge headed at $v\in \V_i$ is sampled that occurs with probability $\closureCoeff_v$ by Corollary~\ref{lemma:unbiasedBaseline}, hence $\expectation[\estFunc_i] = \clustGenPartition_i$. The concentration of $f_i$ follows by Hoeffding's bound~(Theorem~\ref{th:Hoeffding}) and the application of the union bound over the $k$ partitions.
\end{proof}

\section{Parameter setting}\label{appsec:parameterGuidance}
\edit{Following the discussion in~\Cref{subsec:adaptivity}, since~\mainAlg adapts both the parameter $q$ and the bounds $\widehat{\varepsilon}_j$ we found that best parameter setting is as follows:
\begin{itemize}
	\item Setting the parameter $C$ (the threshold defining the filtering of small-degree nodes) according to the degree distribution of the graph. 
	That is, such that the minimum-degree of the nodes $v\in V'$ in the remaining graph $G'$ have degree $\ge 10$ in the original graph $G$ (note that they may have smaller degree in $G'$). 
	This helps to obtain a better bound for the values $R_j$, for $j\in[k]$.
	\item Setting the parameter $\eta = 0.01$, as it gives high probabilistic guarantees of success and it has little impact on the final sample size.
	\item Setting $c$, i.e., the sample size used to estimate the variance of the functions $f_j$, for $j\in [k]$ and optimize the parameter $q$ to a small constant, (we use $c=500$), therefore fixing the value of $q$ becomes extremely efficient.
	\item Setting the parameter $\theta$ controlling the geometric sample schedule $s_i = \theta s_{i-1}$ to $\theta = 1.4$, in accordance with previous work~\cite{Riondato2018Abra, Pellegrina2023Silvan}. Often yielding good results in practice as shown in~\Cref{sec:exps}.
	\item Setting the parameters $\varepsilon_j = \varepsilon \in (0.15,0.05)$, this is given by the fact that~\mainAlg's guarantees are adaptive. Hence, as from~\Cref{exps:sota} it is often the case that~\mainAlg returns bounds $\widehat{\varepsilon}_j$ much smaller than $\varepsilon_j$ (up to the order of $10^{-4}$). This is clearly, another advantage of our algorithm $\varepsilon$.
\end{itemize}
Concerning~\algFixedSS we observe that it is often sufficient to set the sample size $s = m/500$ where $m$ is the number of edges in the graph $G$, to obtain empirical errors $|f_j-\clustGenPartition_j|$ that are negligible, yielding also remarkably accurate relative approximations.
}

\begin{figure*}
	\addtolength{\tabcolsep}{-0.5em}
	\begin{tabular}{ll}
		\includegraphics[width=1.1\columnwidth]{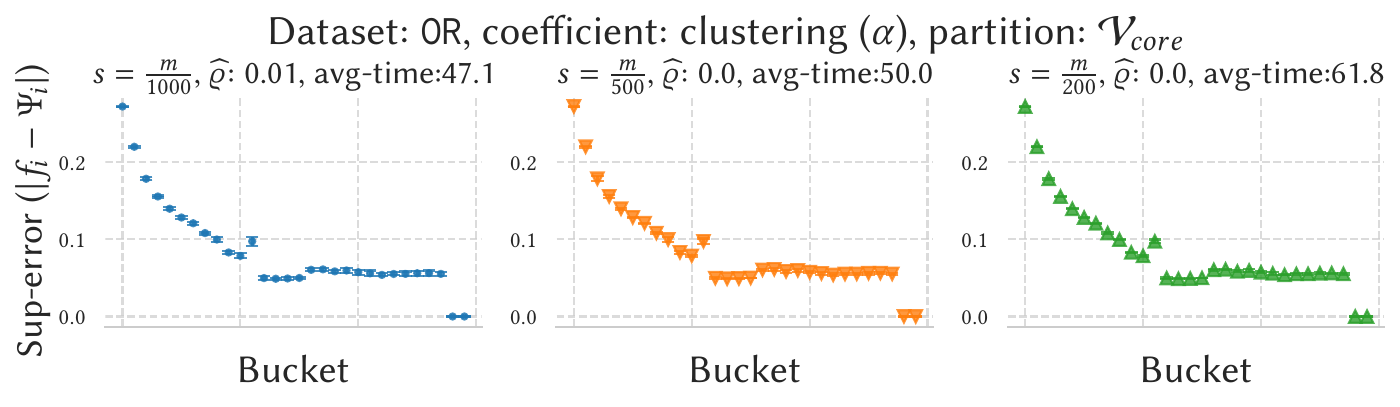}&
		\includegraphics[width=1.1\columnwidth]{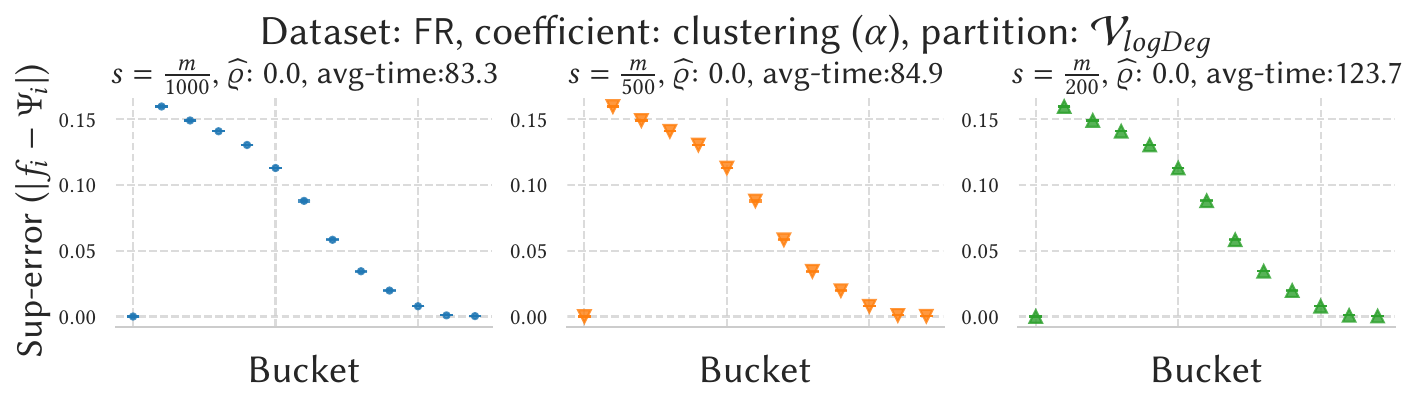}\\
		\includegraphics[width=1.1\columnwidth]{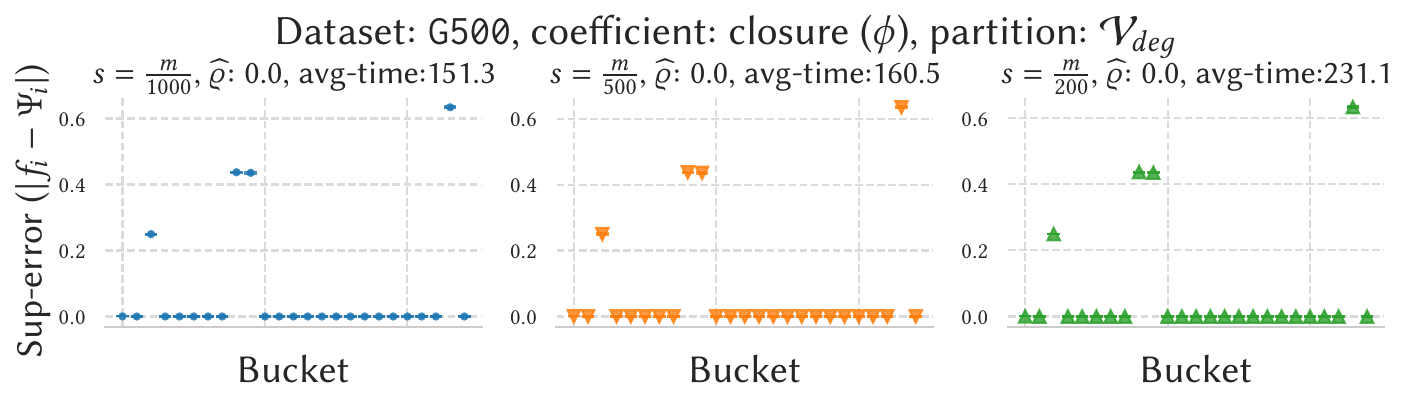}&
		\includegraphics[width=1.1\columnwidth]{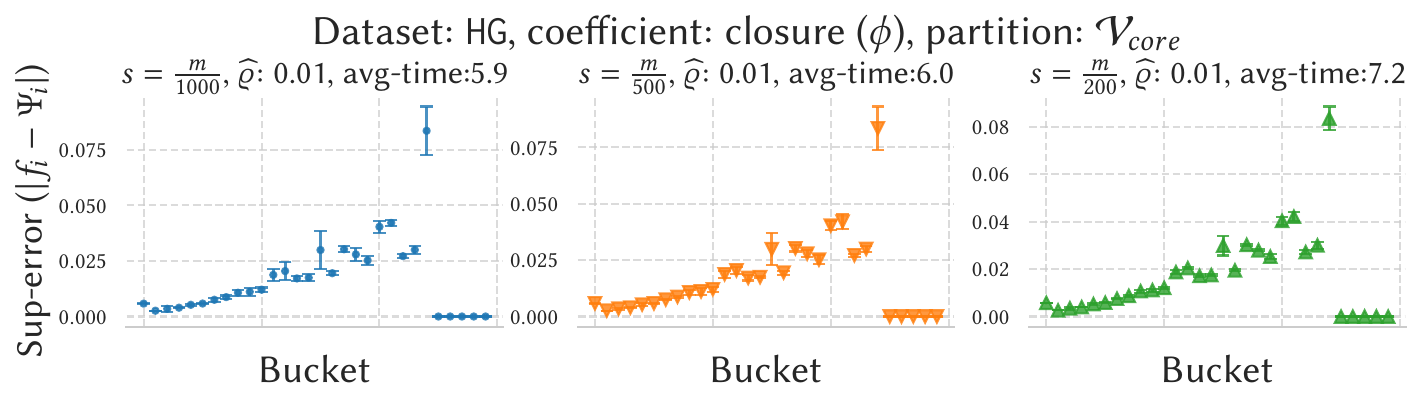}\\
		\includegraphics[width=1.1\columnwidth]{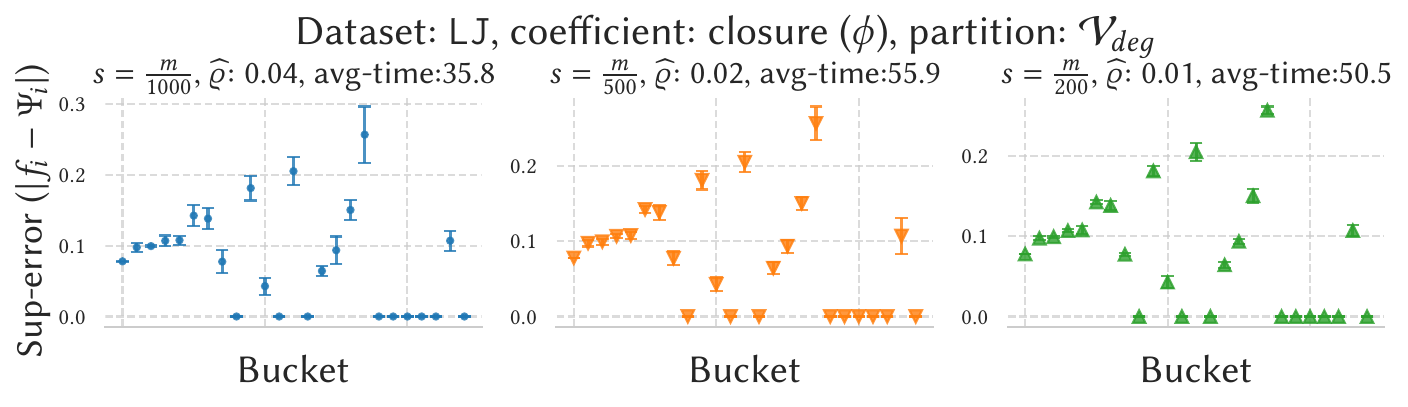}&
		\includegraphics[width=1.1\columnwidth]{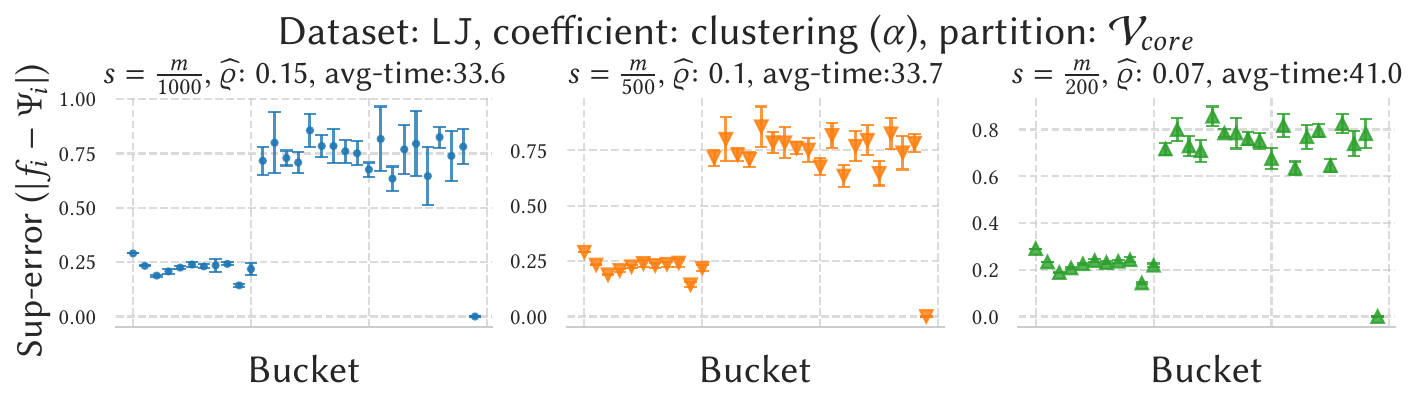}\\
		\includegraphics[width=1.1\columnwidth]{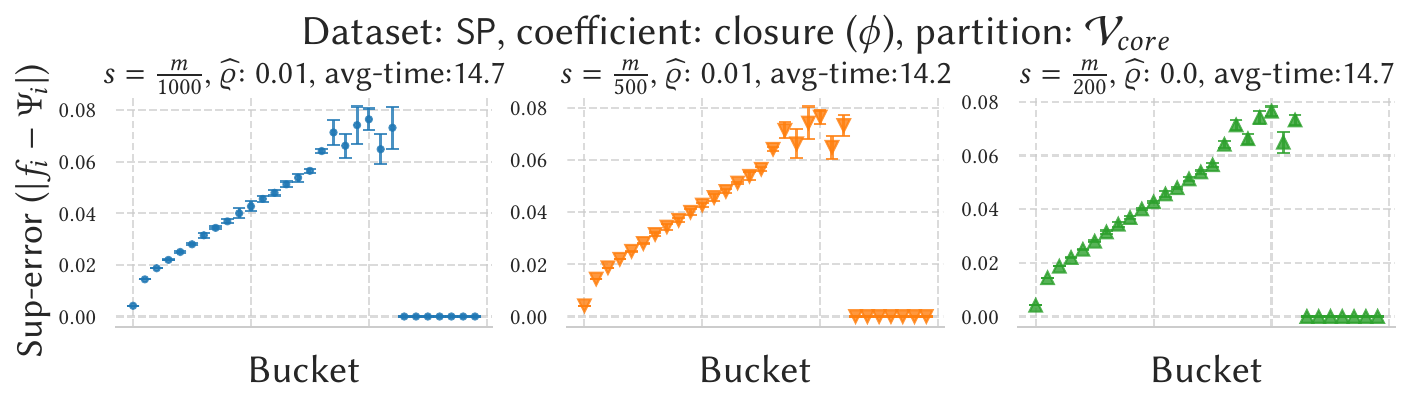}&
		\includegraphics[width=1.1\columnwidth]{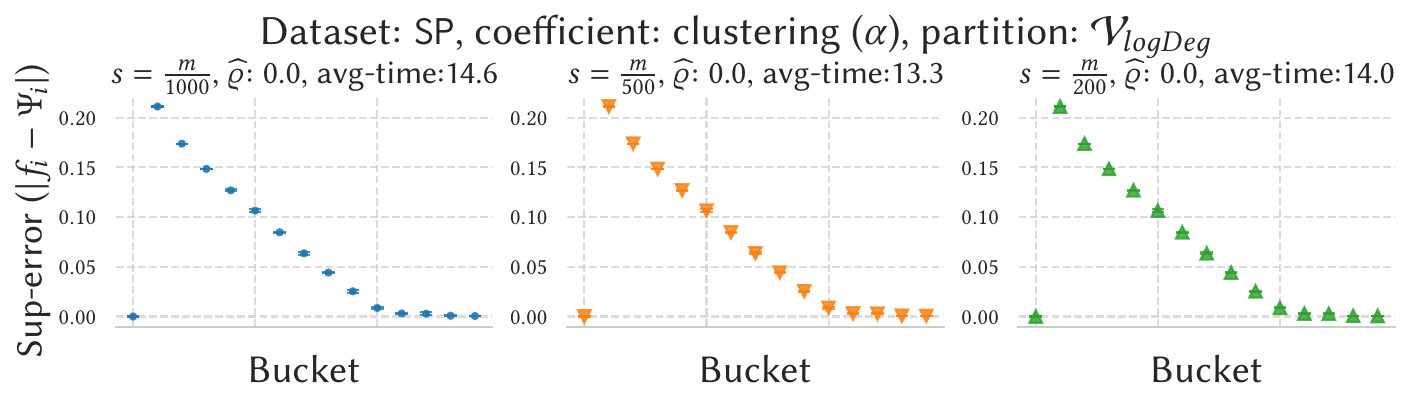}\\
	\end{tabular}
	\caption{\edit{Value $\clustGenPartition_i$} and its \edit{maximum error ($|f_i-\clustGenPartition_i|$) over five runs. We also report, the supremum error $\hat{\varrho}=\sup_{i\in[k]} |f_i-\clustGenPartition_i|$} and the average runtime over five independent runs across all buckets, for varying sample size \edit{($s\in\{1,2,5\}$\textperthousand\ of the total edges $m$).}}
	\label{fig:varyingSSApp}
\end{figure*}

\section{Additional results}

\subsection{Reproducibility}
The code used for our experiments is available online (\CODEURL). 
Below we provide additional details on the setup of the experiments.
\subsubsection{Datasets}\label{appsec:datasets}
All the datasets that we considered are available online, in particular we downloaded some of the largest graphs from SNAP~\cite{Leskovec2016SNAP}, and the network repository~\cite{Rossi2015NetworkRep}. 
\edit{More in more detail,
\begin{itemize}
	\item fb-CMU (\url{https://networkrepository.com/fb-CMU-Carnegie49.php})
	\item SP (\url{https://snap.stanford.edu/data/soc-Pokec.html})
	\item FR (\url{https://snap.stanford.edu/data/com-Friendster.html})
	\item OR (\url{https://snap.stanford.edu/data/com-Orkut.html})
	\item LJ (\url{https://snap.stanford.edu/data/soc-LiveJournal1.html})
	\item BM (\url{https://networkrepository.com/bio-mouse-gene.php})
	\item G500 (\url{https://networkrepository.com/graph500-scale23-ef16-adj.php})
	\item GP (\url{https://snap.stanford.edu/data/gplus_combined.txt.gz})
	\item PT (\url{https://networkrepository.com/proteins-all.php})
	\item HW (\url{https://networkrepository.com/ca-hollywood-2009.php})
	\item HG (\url{https://snap.stanford.edu/data/higgs-twitter.html})
	\item BNH (\url{https://networkrepository.com/bn-human-BNU-1-0025919-session-1.php})
	\item TW (\url{https://snap.stanford.edu/data/twitch_gamers.html})
\end{itemize}}
Since we deal with simple and undirected graphs, we only performed standard pre-processing, i.e., removing self loops and parallel edges, and remapping nodes from $1,\dots,n$.

\subsubsection{Parameter and settings}\label{appsec:params} We further clarify all the parameter settings used in the various experiments. Note that all the scripts for reproducing the results are provided with the source code (i.e., containing also the parameters described hereby). This section is in addition to that.
\begin{itemize}
	\item For the setting of Section~\ref{exps:accuracyAndEff}: we do not set the values of the parameters $\varepsilon,\delta$, as we used~\algFixedSS. The value of $C$ used for the optimization in Section~\ref{subsec:algPractOpt} is fixed to 30 for all datasets but \texttt{G500}, \texttt{BNH}, \texttt{BM} where $C=100$. As discussed in Section~\ref{exps:accuracyAndEff}, we tested different values of $s$ according to the size of the graph considered setting $s$=1\textperthousand, $s$=2\textperthousand, $s$=5\textperthousand. The value of the parameter $q$ is optimized according to Section~\ref{subsec:varianceOpt} and a total samples of 10\% of $s$.
	\item For the setting in Section~\ref{exps:sota}: we run all datasets with $\varepsilon=0.075$ and $\delta=0.01$. For each partition we obtain $\widehat{\varepsilon}$ from~\mainAlg and we use it as input for the baselines. The value of $C$ used for the optimization in Section~\ref{subsec:algPractOpt} is fixed to 150 for all datasets but \texttt{G500},\texttt{GSH}, \texttt{BNH}, \texttt{BM} where $C=500$. The value of the parameter $q$ is optimized according to Section~\ref{subsec:varianceOpt} over a total of 500 sampled edges.
	The time limit was set to 10 minutes for all datasets but the largest/densest ones (i.e., \texttt{G500}, \texttt{BNH}, \texttt{BM}) where the time limit was set to 1 hour.
	\item For the setting of Section~\ref{subsec:caseStudy}: we set $\varepsilon_j=\varepsilon=0.05$ for each $j$ corresponding to a different community, and $\eta=0.01$. The results shown, are over a single run.
\end{itemize}

\begin{figure}
	\includegraphics[width=\columnwidth]{figures/adapt/testlegend}
	\addtolength{\tabcolsep}{-0.6em}
	\begin{tabular}{llllll}
		\includegraphics[width=0.35\columnwidth]{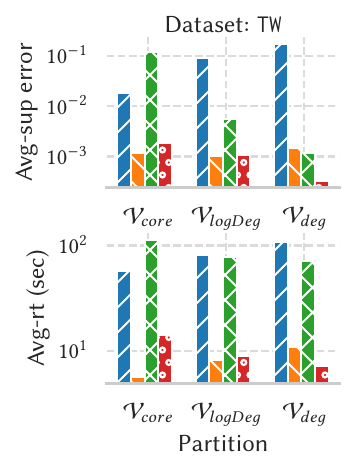} & \includegraphics[width=0.35\columnwidth]{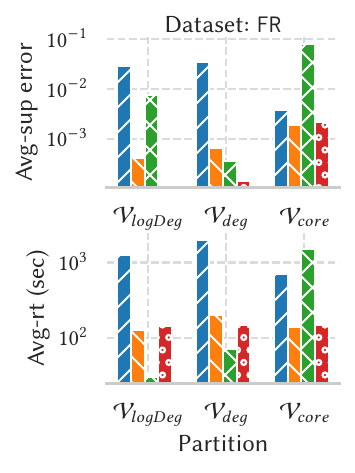} & \includegraphics[width=0.35\columnwidth]{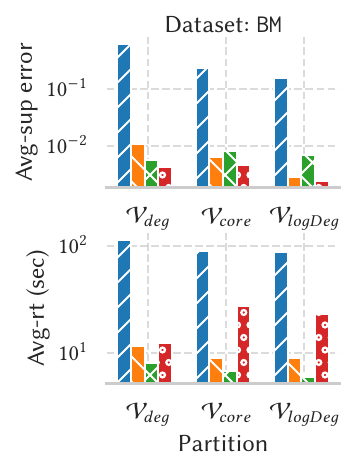} & 
	\end{tabular}
	\caption{Comparison of~\mainAlg and the baselines~\baseFull. For each dataset we show, (top plot): the average supremum error over all buckets over the various runs. (bottom): average runtime to perform an execution.}
	\label{fig:expsSotaApp}
\end{figure}

\subsubsection{DBLP dataset}\label{appsubsec:DBLP}
Here we provide additional information on how we build the DBLP dataset. 
We used the DBLP dump referring to August 2024 (\url{https://dblp.org/xml/release/dblp-2024-08-01.xml.gz}), all the code to generate such data is reported in our repository.
To classify the authors on each snapshot from Section~\ref{subsec:caseStudy}.
considering the following classification rule,
\begin{itemize}
	\item We first consider the following fields represented by a set of conferences:
	\begin{itemize}
		\item Data management: VLDB, PVLDB, SIGMOD, PODS, ICDE, PODC, SIGIR;
		\item Data mining: KDD, WWW, CIKM, SDM, DAMI, TKDD, TKDE, WSDM, SIGKDD;
		\item Machine learning: NIPS, ICML, ICLR, AAAI, IJCAI, AAMAS;
		\item Theoretical computer science (TCS): STOC, FOCS, SODA, ITCS, ICALP, ISAAC, ESA, FOSSACS, STACS, COLT, MFCS, STACS;
		\item Computational biology: RECOMB, CIBB, ISMB, WABI, PSB;
		\item Software engineering: ASE, DEBS, ESEC, FSE, ESEM, ICPE, ICSE, ISSTA, MOBILESoft, MoDELS, CBSE, PASTE, SIGSOFT;
	\end{itemize}
	\item For each author we count the number of publication in the above fields, and assign the label according to field in which the author published most. 
	If the author did not publish in any of the fields, we assign it to the ``other'' category.
\end{itemize}

We note that for some of the graphs that refer to years from 1970 to 1995, for some communities we have no information, i.e., the number of authors in such communities is 0. This is likely due to the classification above, and the set of conferences that we used to classify the authors, which started after the considered years. To obtain the score in Figure~\ref{fig:DBLPTemporal} we use~\mainAlg with $\varepsilon=0.05$ and $\eta=0.01$ on each dataset over a single run.

\begin{figure*}
	\addtolength{\tabcolsep}{-0.5em}
	\begin{tabular}{ll}
		\includegraphics[width=1.1\columnwidth]{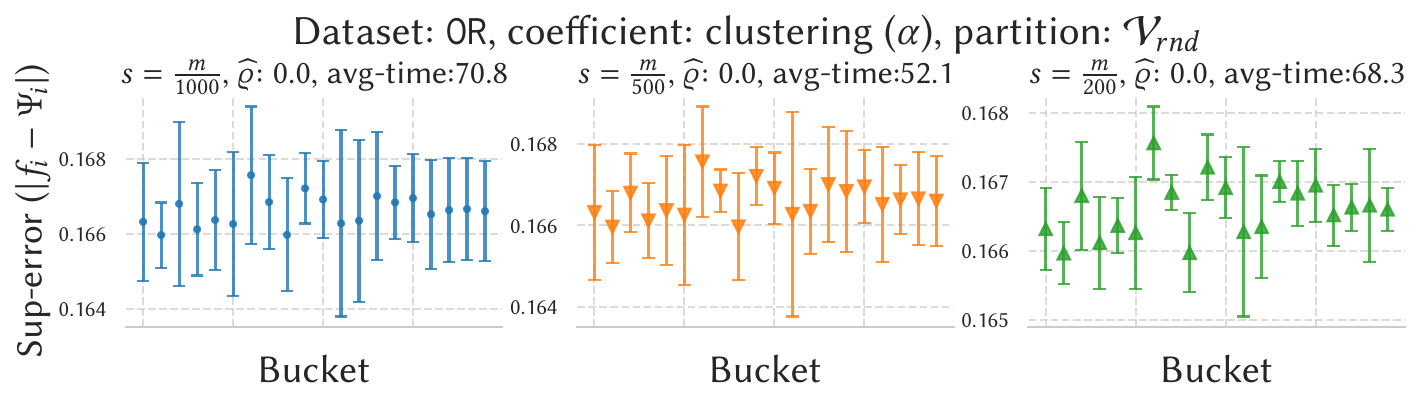}&
		\includegraphics[width=1.1\columnwidth]{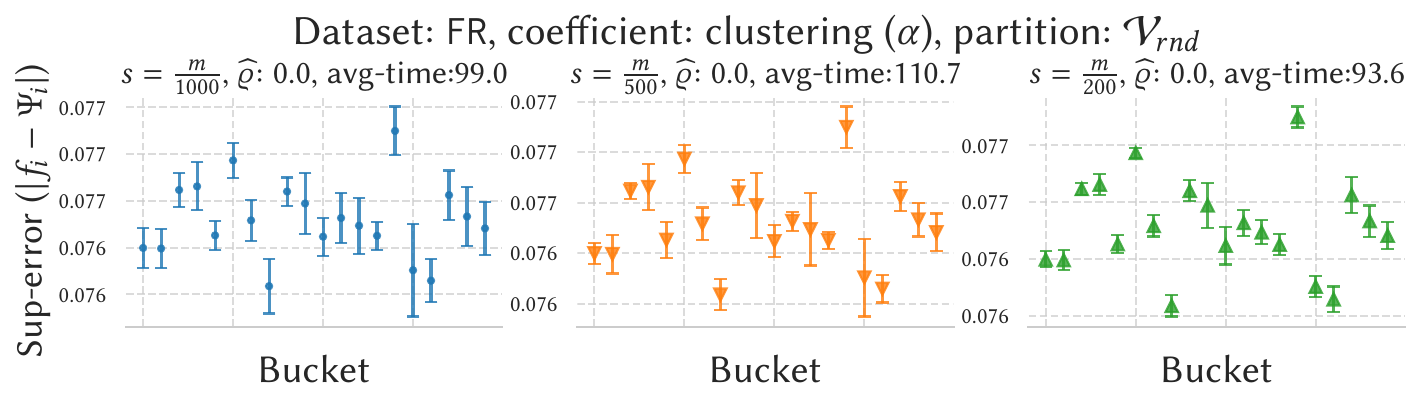}\\
		\includegraphics[width=1.1\columnwidth]{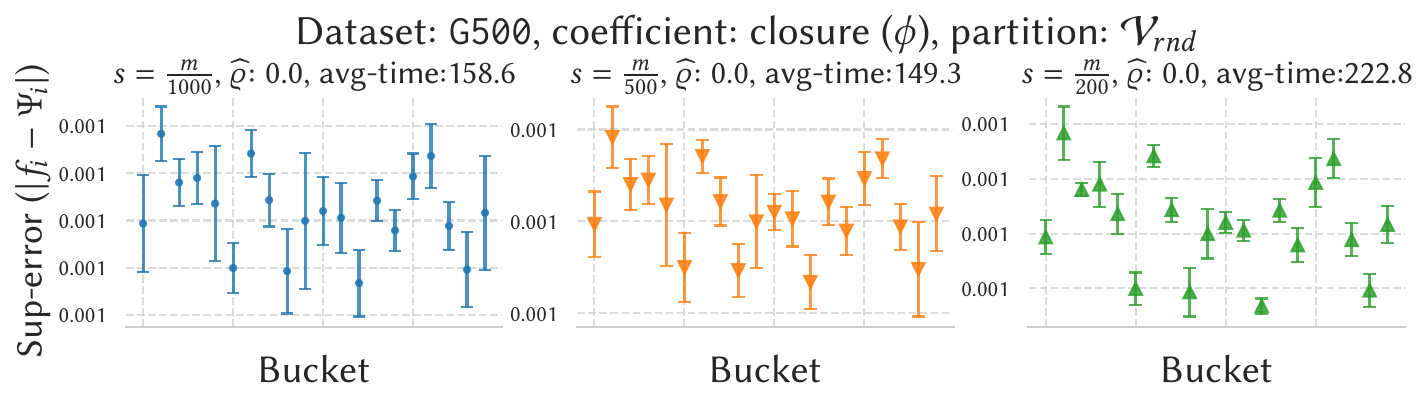}&
		\includegraphics[width=1.1\columnwidth]{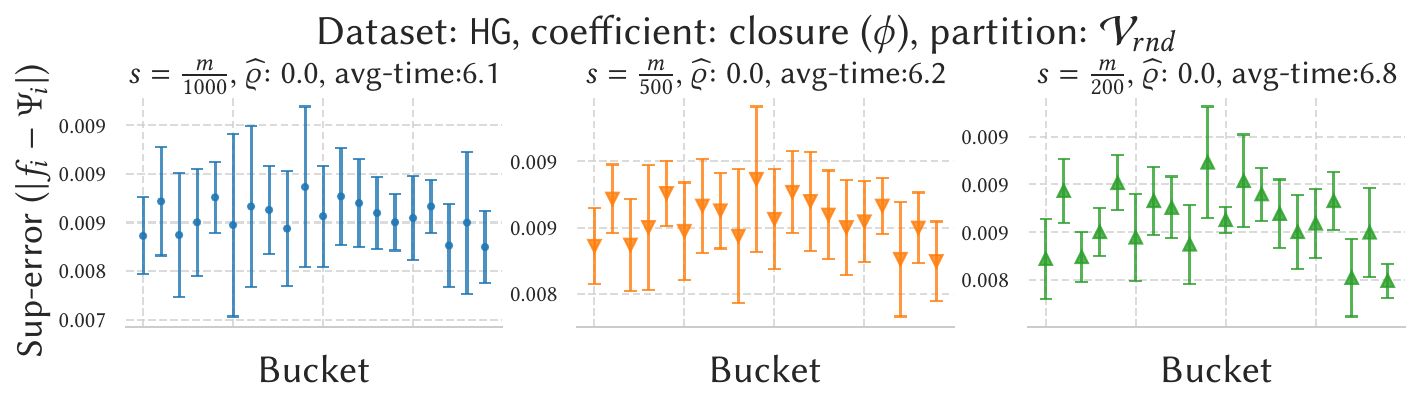}\\
		\includegraphics[width=1.1\columnwidth]{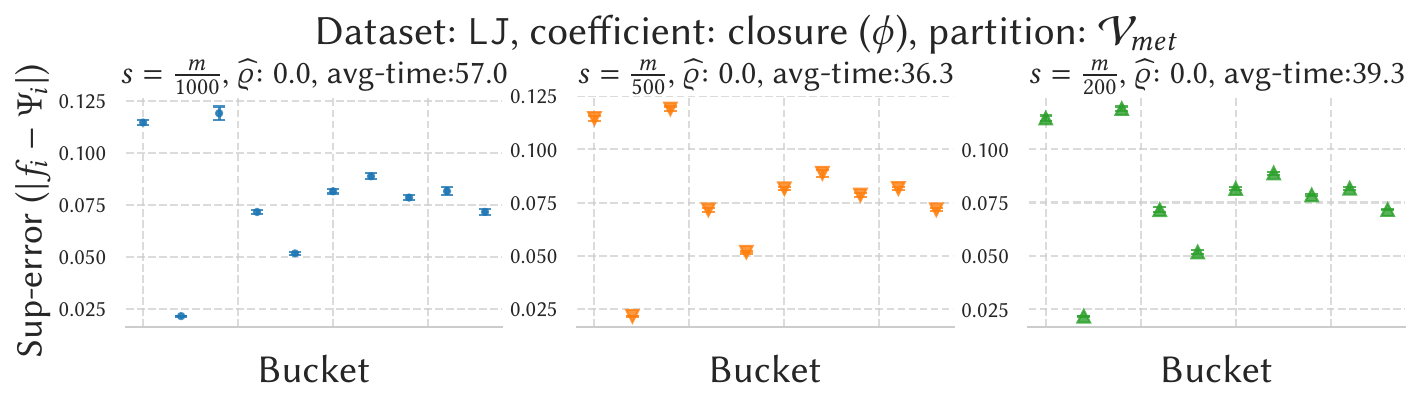}&
		\includegraphics[width=1.1\columnwidth]{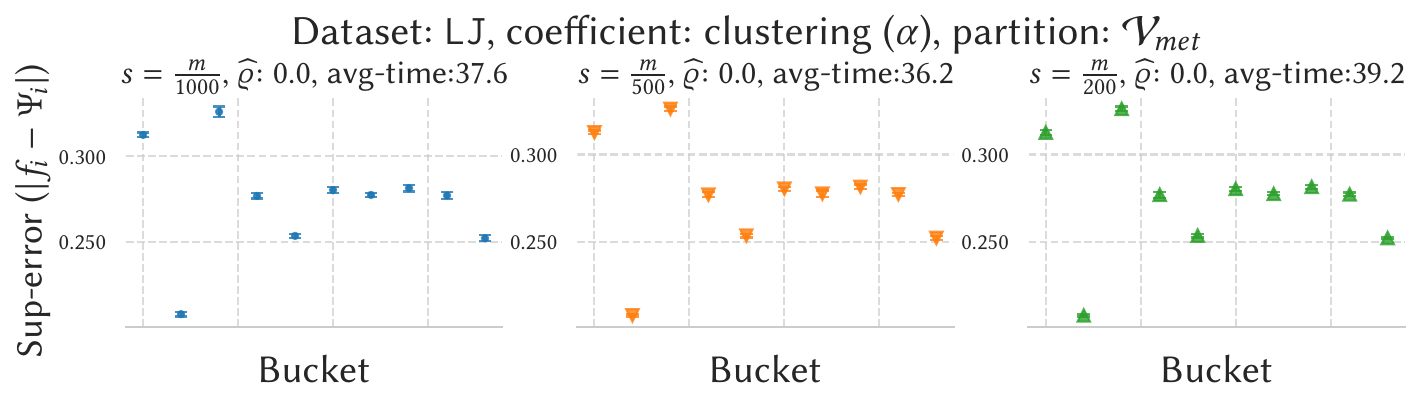}\\
		\includegraphics[width=1.1\columnwidth]{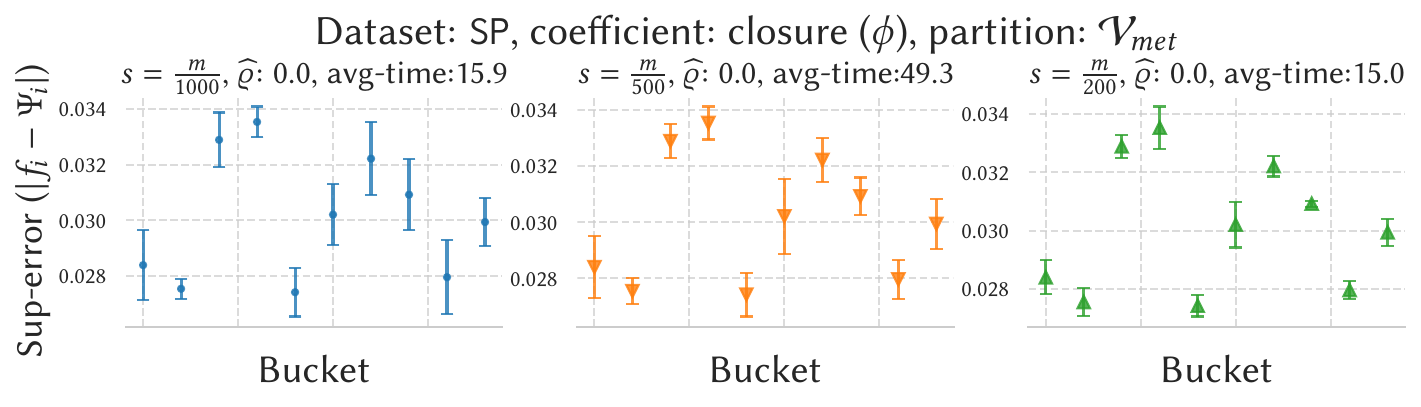}&
		\includegraphics[width=1.1\columnwidth]{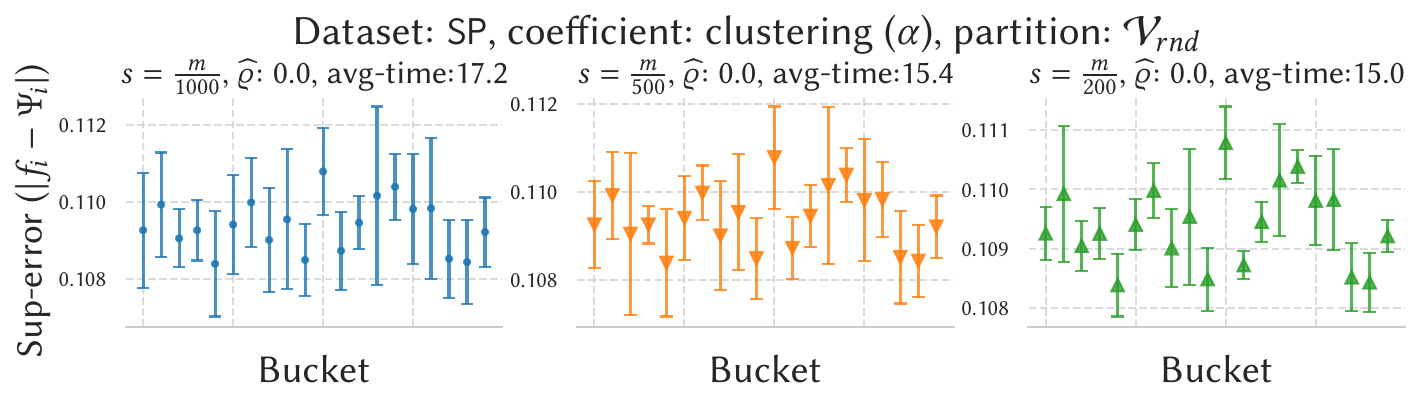}\\
	\end{tabular}
	\caption{\edit{Value $\clustGenPartition_i$} and its \edit{maximum error ($|f_i-\clustGenPartition_i|$) over five runs. We also report, the supremum error $\hat{\varrho}=\sup_{i\in[k]} |f_i-\clustGenPartition_i|$} and the average runtime over five independent runs across all buckets, for varying sample size \edit{($s\in\{1,2,5\}$\textperthousand\ of the total edges $m$).}}
	\label{fig:varyingSSAppNewParts}
\end{figure*}

\subsection{Extending the results of Section~\ref{exps:accuracyAndEff}}\label{appsec:missingSS}
Figure~\ref{fig:varyingSSApp} reports additional configurations for the setting of the experiments in Section~\ref{exps:accuracyAndEff}. Results follow similar trends to the ones already discussed.

Further, we also considered two additional partitions, 
$\nodesetPartition_{\mathit{rnd}}$ and $\nodesetPartition_{\mathit{met}}$. 
$\nodesetPartition_{\mathit{rnd}}$ is obtained by randomly assigning the nodes over $k=30$ random buckets and $\nodesetPartition_{\mathit{met}}$ is obtained by using the METIS algorithm~\citep{Karypis1998metis} with $k=10$ clusters, (we use PyMetis: \url{https://github.com/inducer/pymetis/tree/main})
First note that clearly under $ \nodesetPartition_{\mathit{rnd}}$ each value $\clustGenPartition_i$ corresponds to an unbiased estimate of the value $\clustGenPartition$ computed on the node-set $V$, i.e., the global average clustering or closure coefficient. Hence all such values are expected to be close. Results are reported in~\Cref{fig:varyingSSAppNewParts}, which confirm our observation, i.e., when considering  $\nodesetPartition_{\mathit{rnd}}$ all the values $\clustGenPartition_i$  are close, and are also close to $\bar{\alpha}$ or $\bar{\phi}$ from~\Cref{tab:datasets}. Overall the results support what we observed in~\Cref{exps:accuracyAndEff} where a sample size of 2\textpertenthousand\ is often sufficient to achieve a supremum error of the order of $10^{-3}$, yielding highly accurate estimates for all buckets.

\subsection{Extending the results of Section~\ref{exps:sota}}\label{appsec:missingAdapt}
Figure~\ref{fig:expsSotaApp} reports additional configurations for the setting of the experiments in Section~\ref{exps:sota}. Also in this case, results follow similar trends to the ones already discussed.

\begin{figure*}
	\addtolength{\tabcolsep}{-0.5em}
	\begin{tabular}{llll}
		\includegraphics[width=0.5\columnwidth]{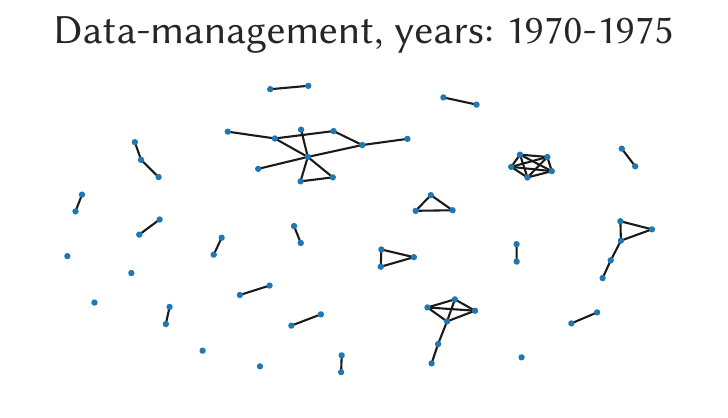}&
		\includegraphics[width=0.5\columnwidth]{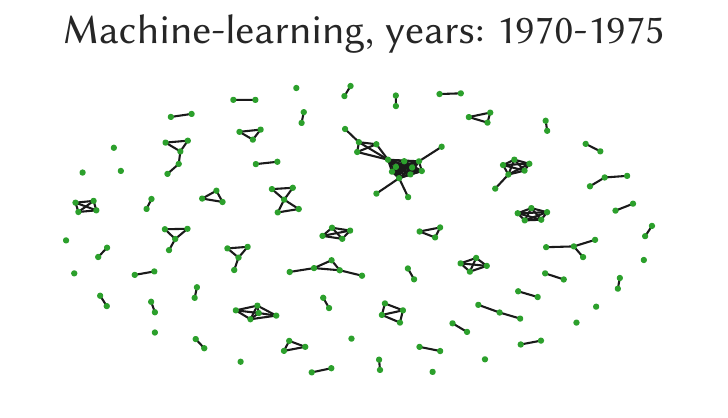}&
		\includegraphics[width=0.5\columnwidth]{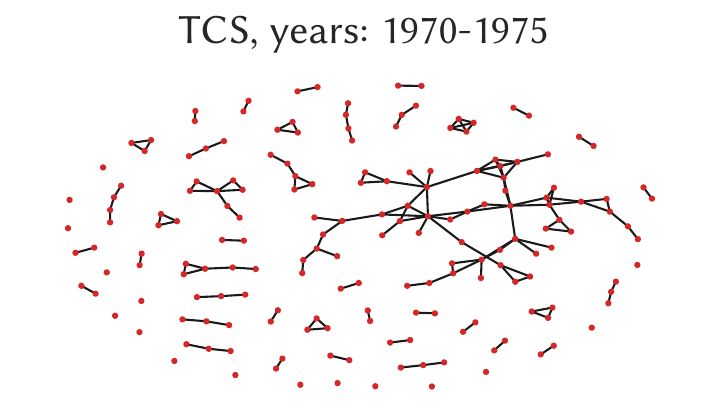}&
		\includegraphics[width=0.45\columnwidth]{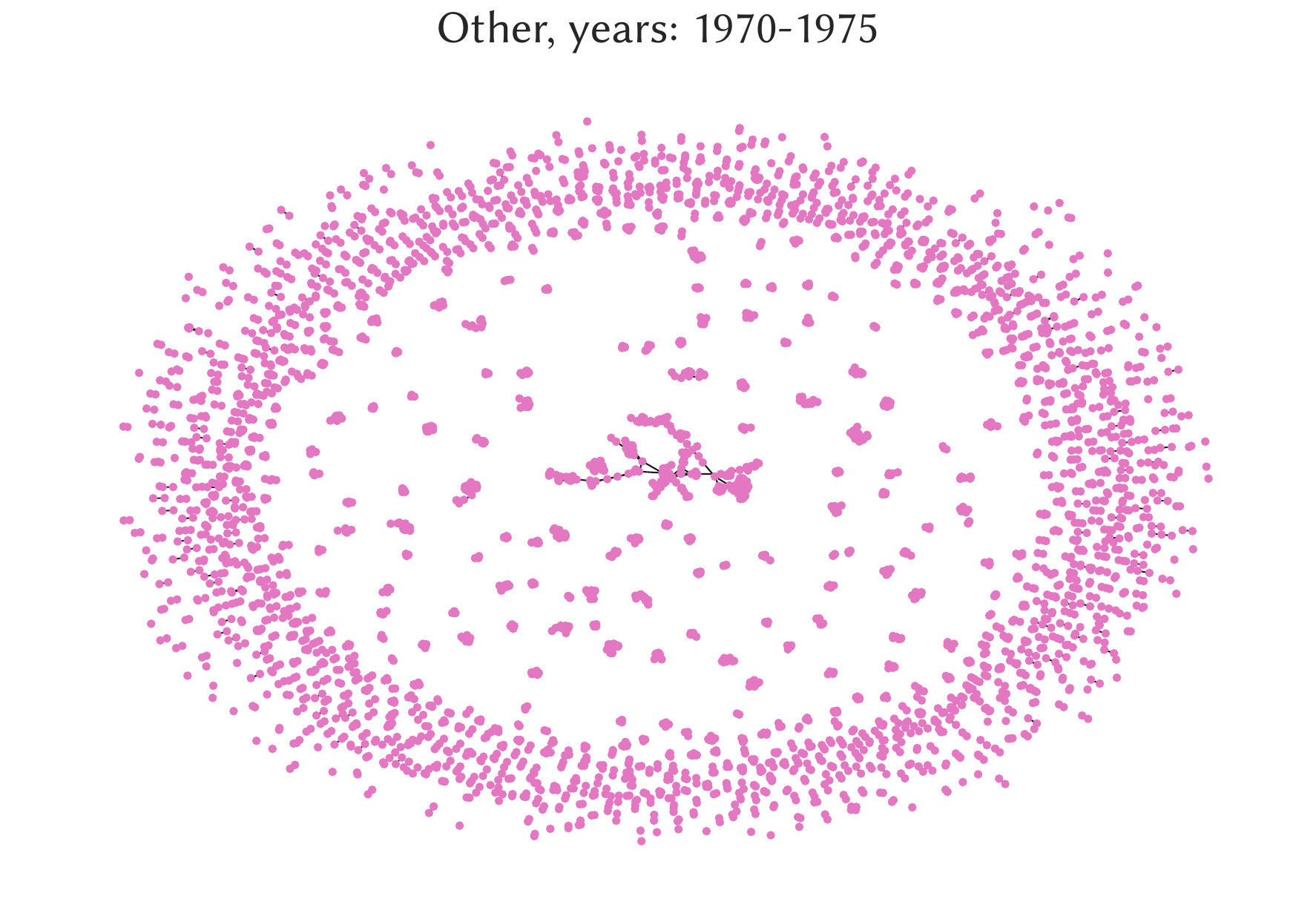}\\
		\includegraphics[width=0.5\columnwidth]{figures/case_study/clusts/1_c_0}&
		\includegraphics[width=0.5\columnwidth]{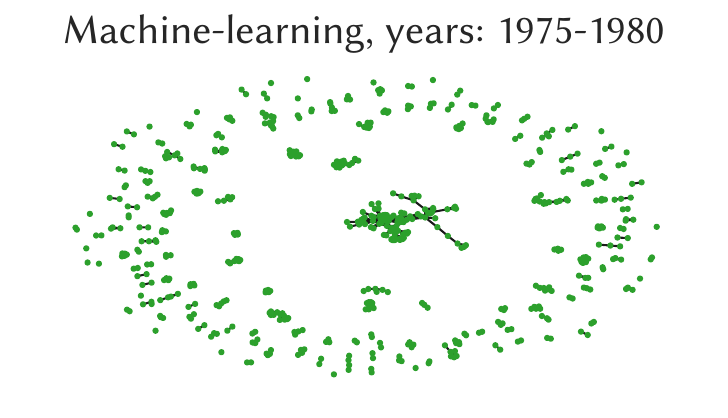}&
		\includegraphics[width=0.5\columnwidth]{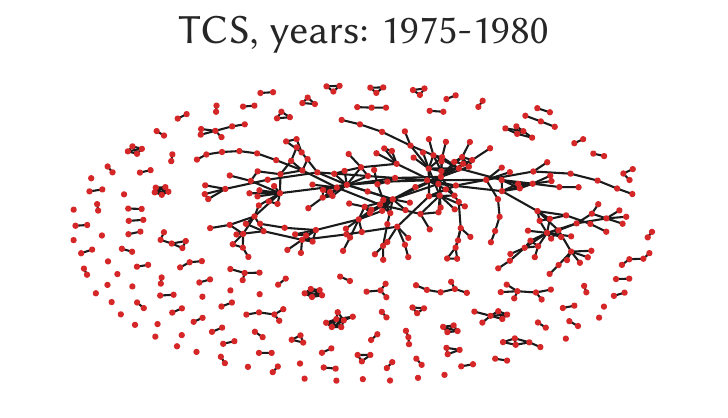}&
		\includegraphics[width=0.5\columnwidth]{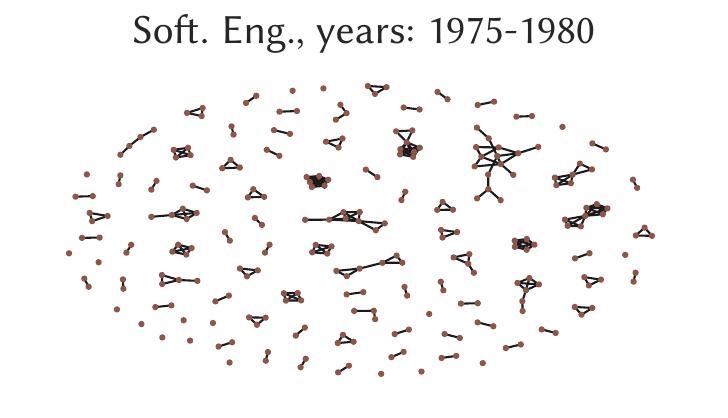}\\
		\includegraphics[width=0.5\columnwidth]{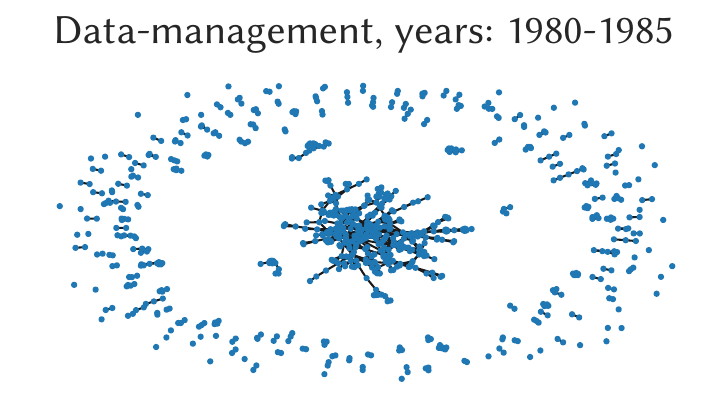}&
		\includegraphics[width=0.5\columnwidth]{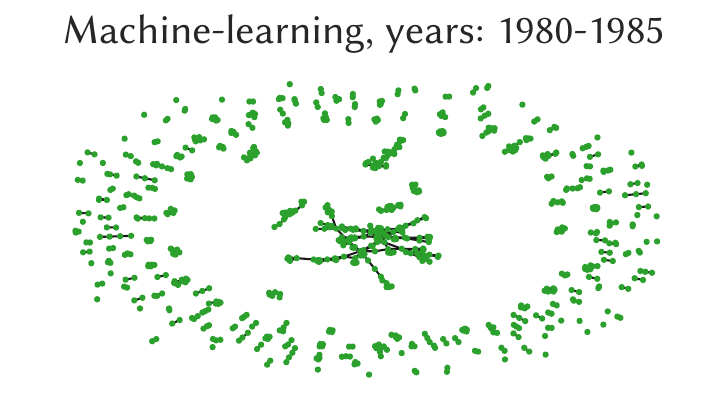}&
		\includegraphics[width=0.5\columnwidth]{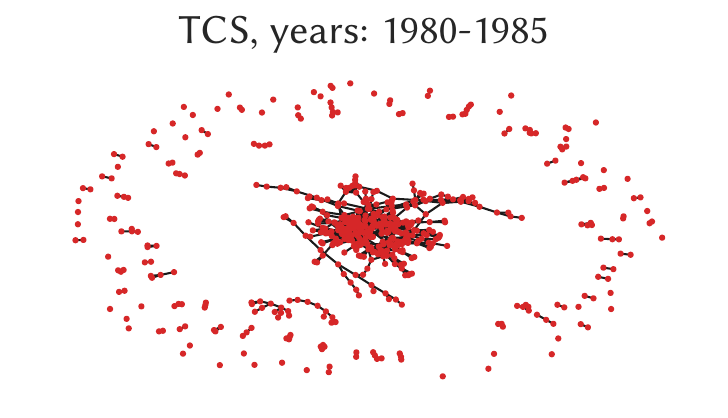}&
		\includegraphics[width=0.5\columnwidth]{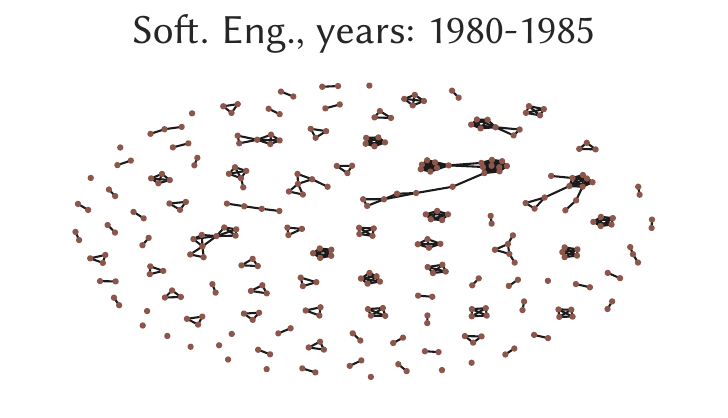}\\
		\includegraphics[width=0.5\columnwidth]{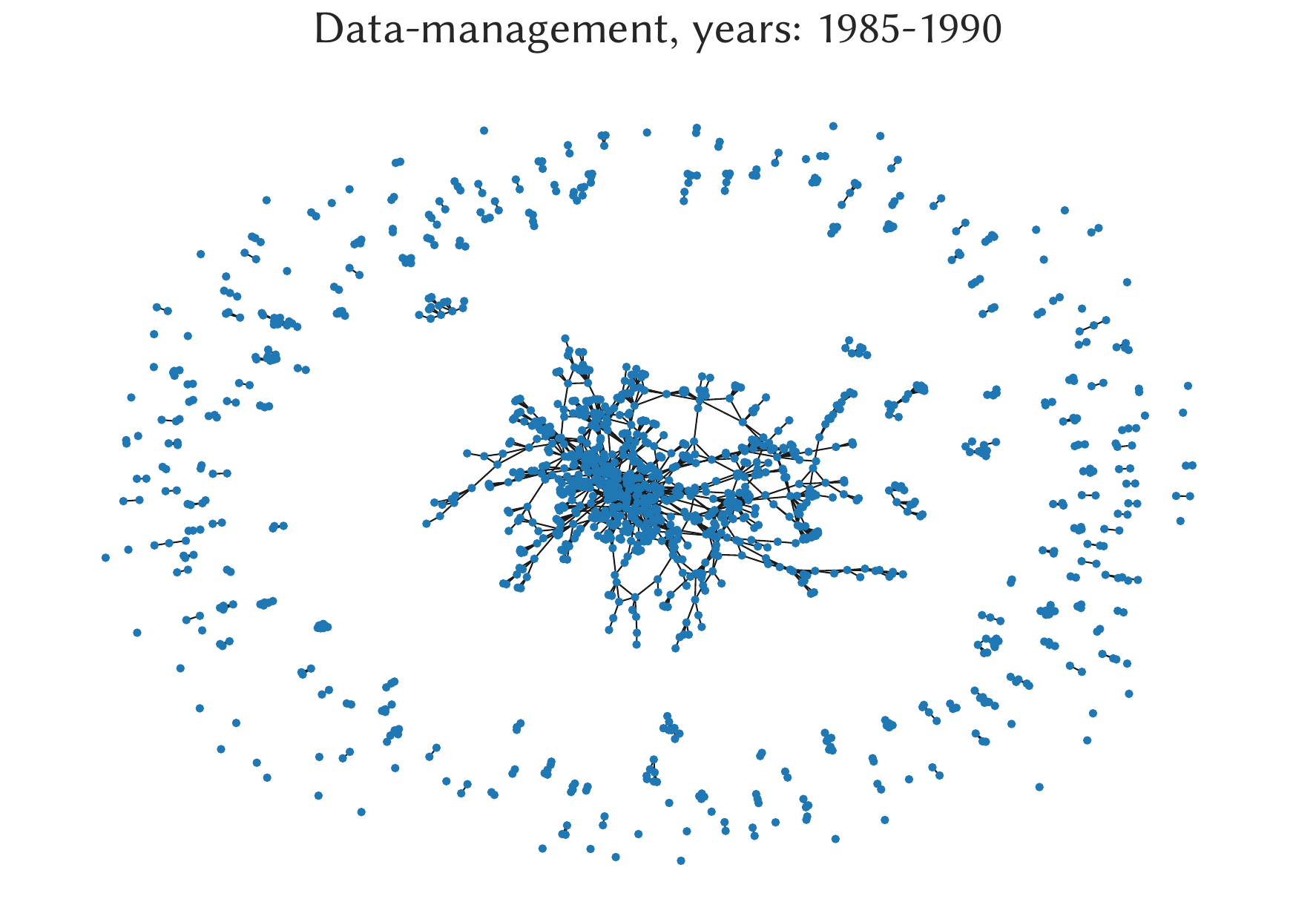}&
		\includegraphics[width=0.5\columnwidth]{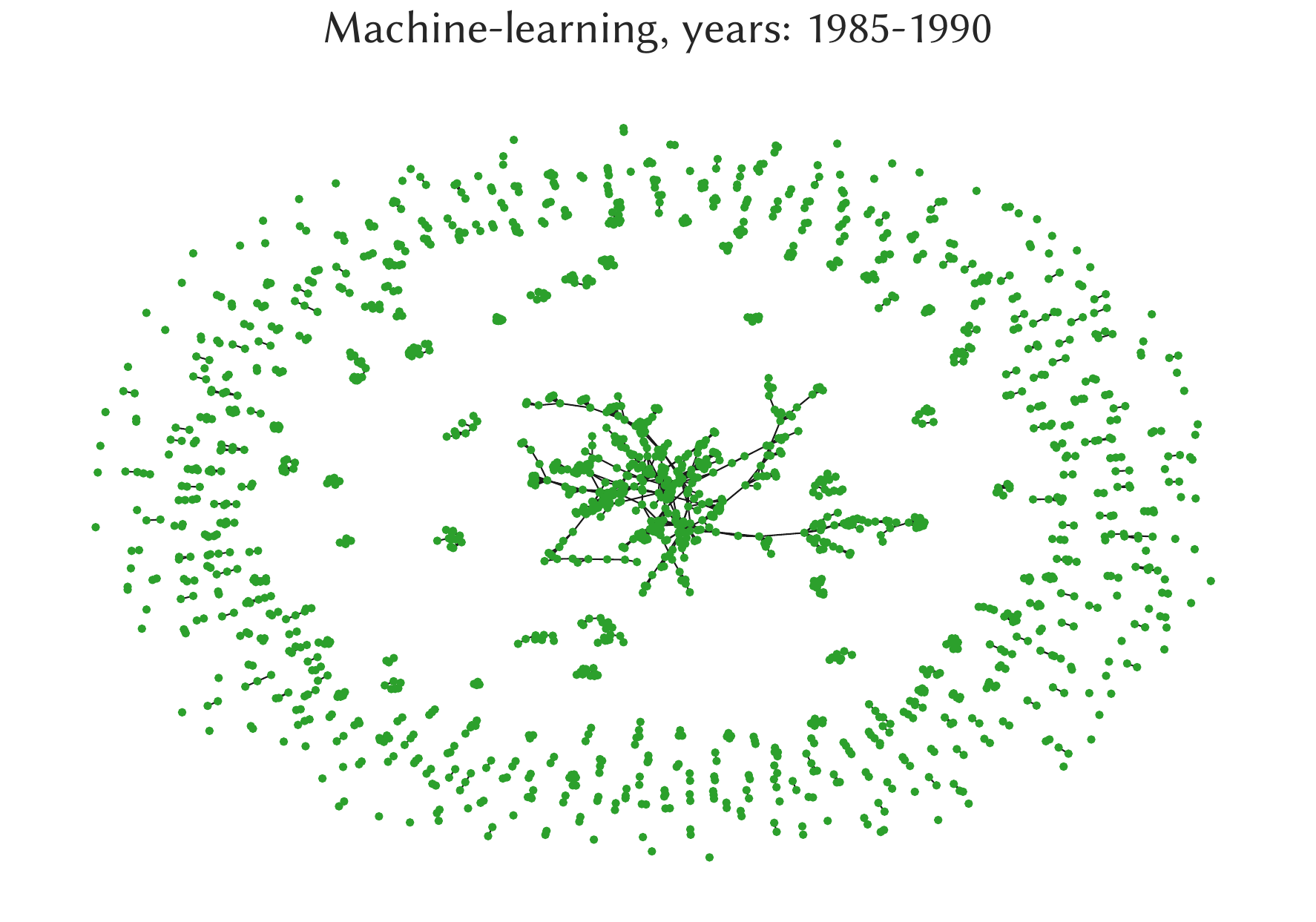}&
		\includegraphics[width=0.5\columnwidth]{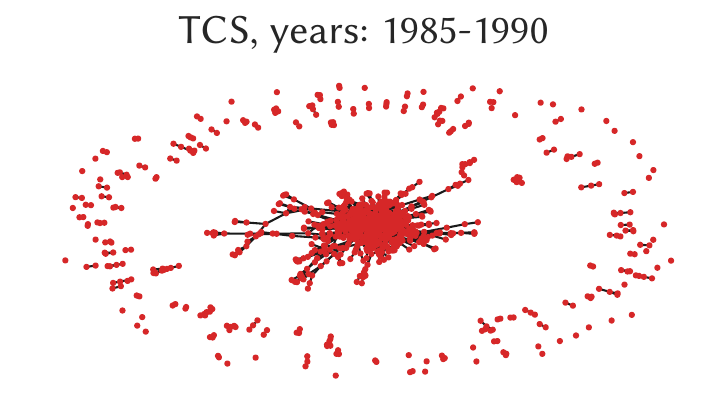}&
		\includegraphics[width=0.5\columnwidth]{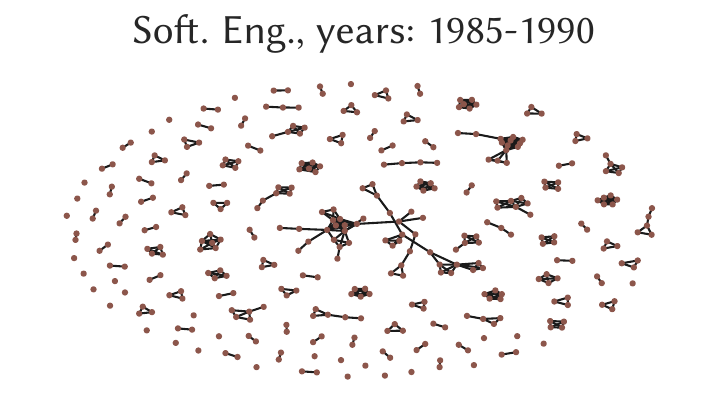}\\
		\includegraphics[width=0.5\columnwidth]{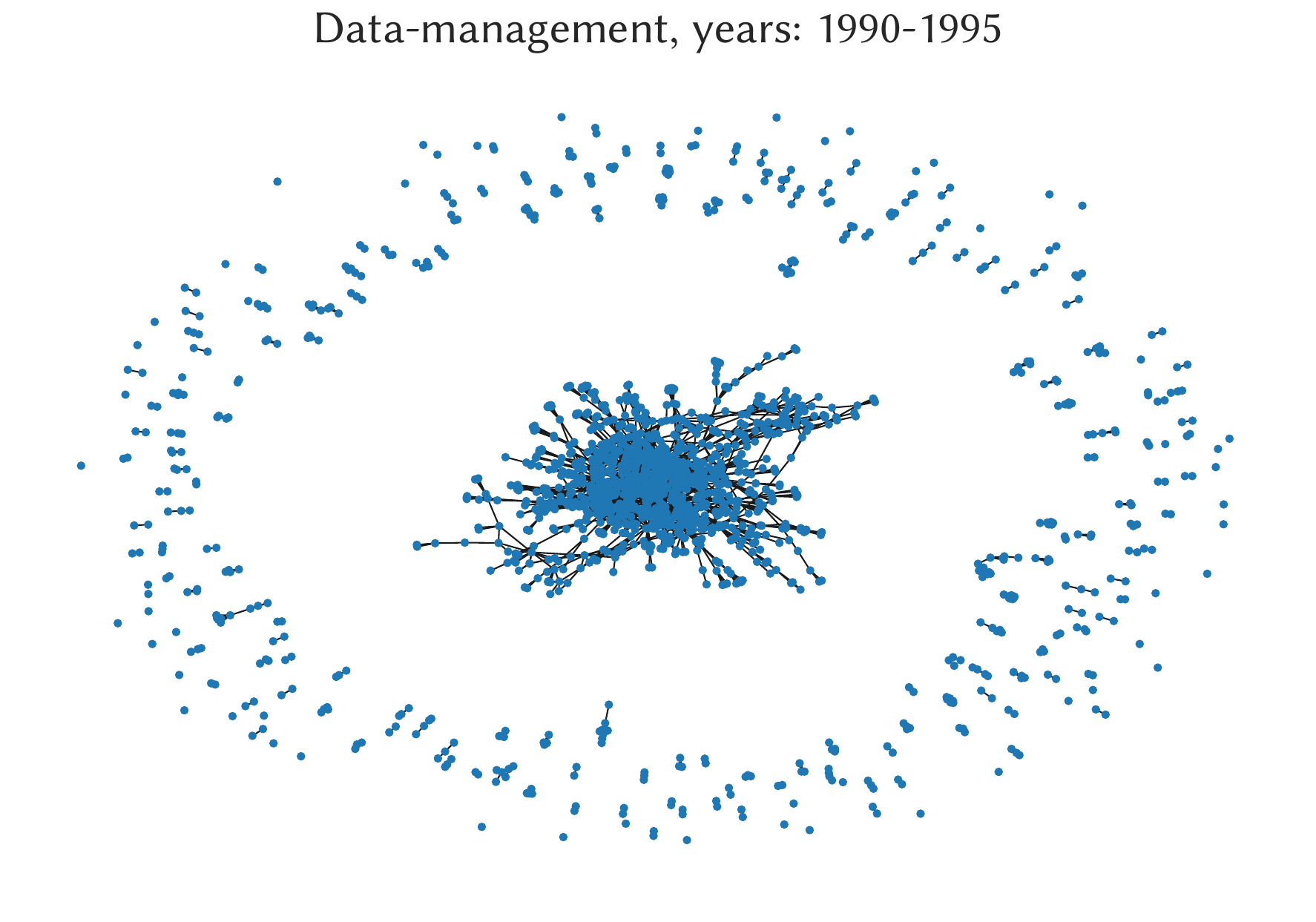}&
		\includegraphics[width=0.5\columnwidth]{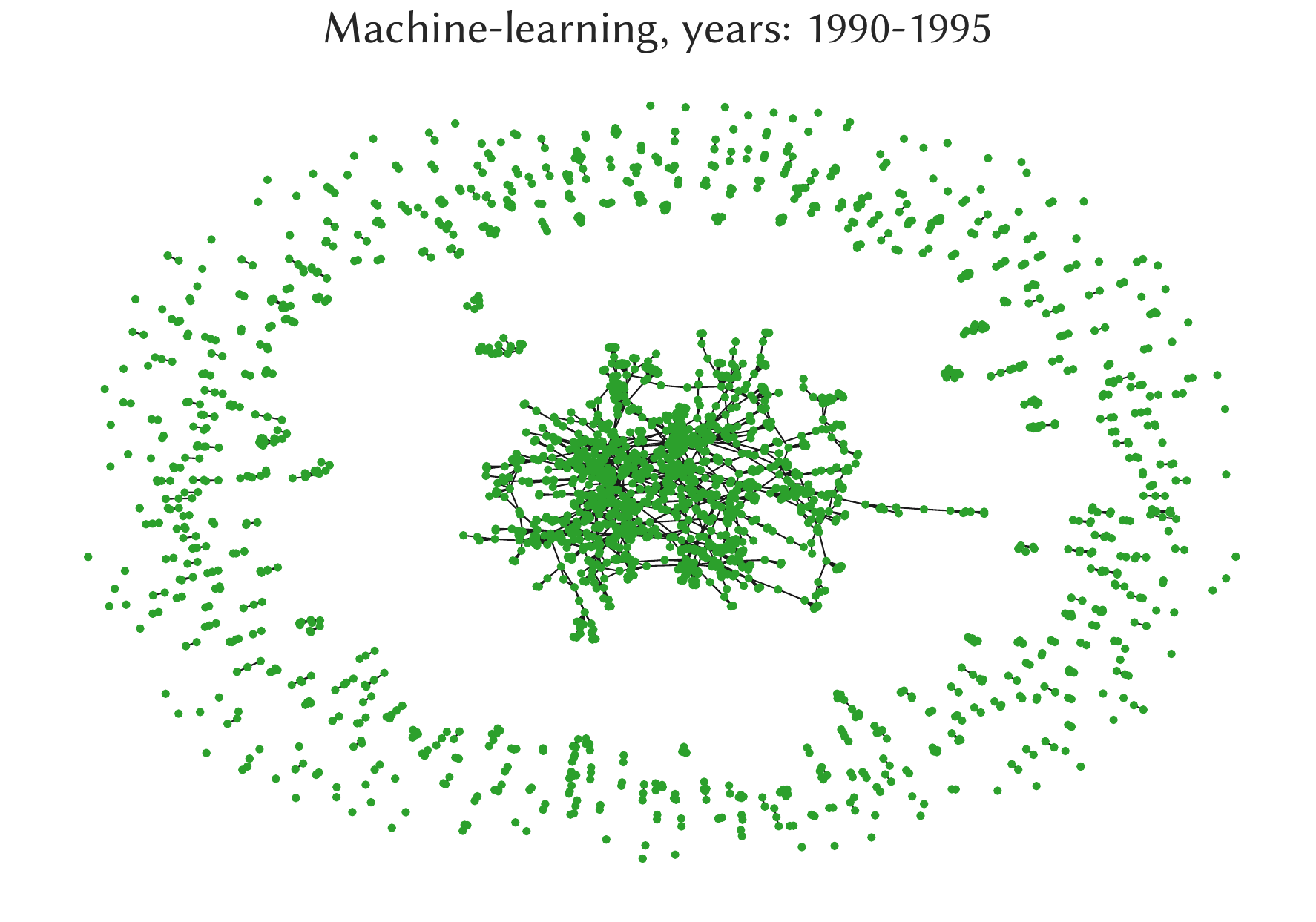}&
		\includegraphics[width=0.5\columnwidth]{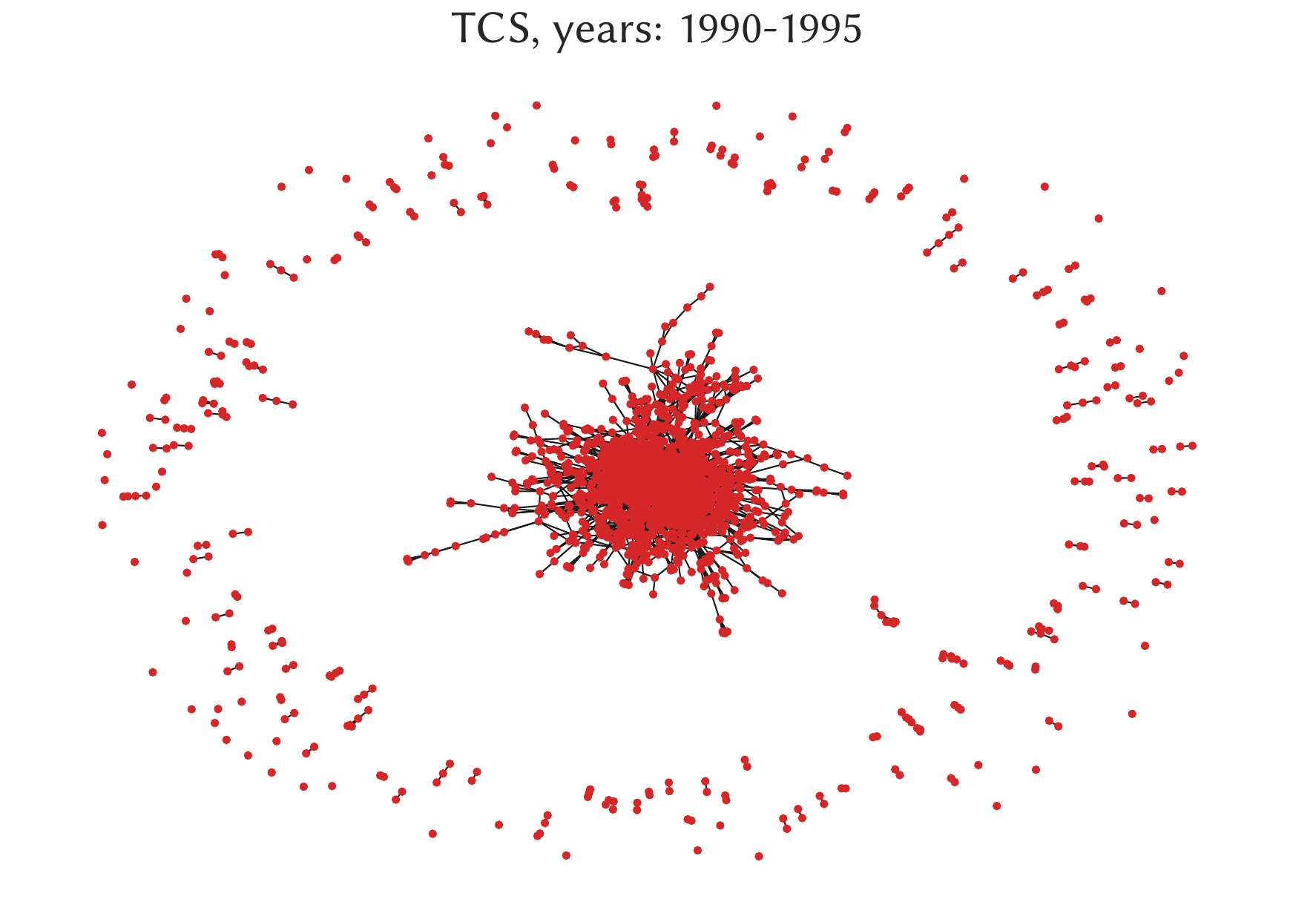}&
		\includegraphics[width=0.5\columnwidth]{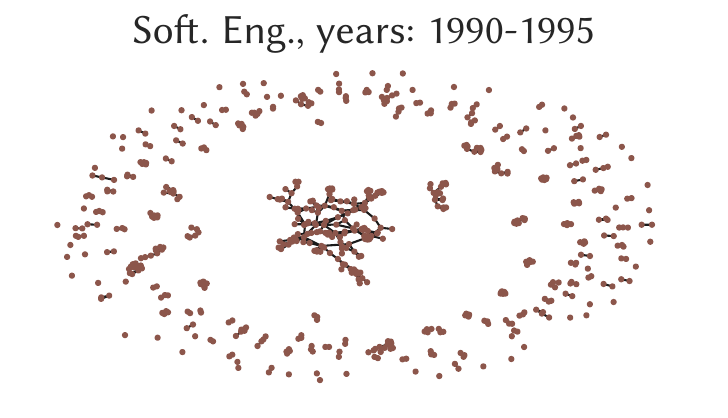}\\
		\includegraphics[width=0.5\columnwidth]{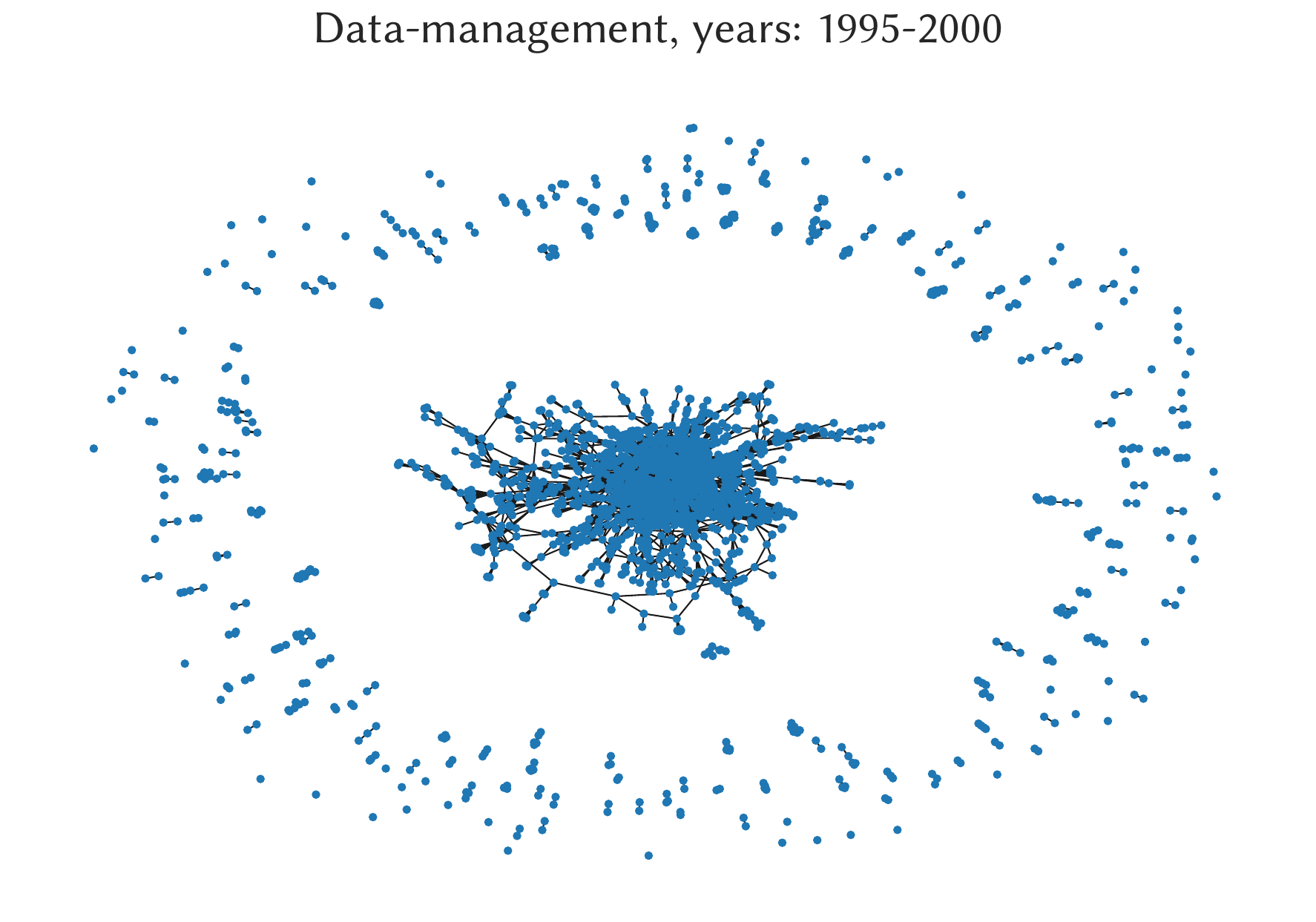}&
		\includegraphics[width=0.5\columnwidth]{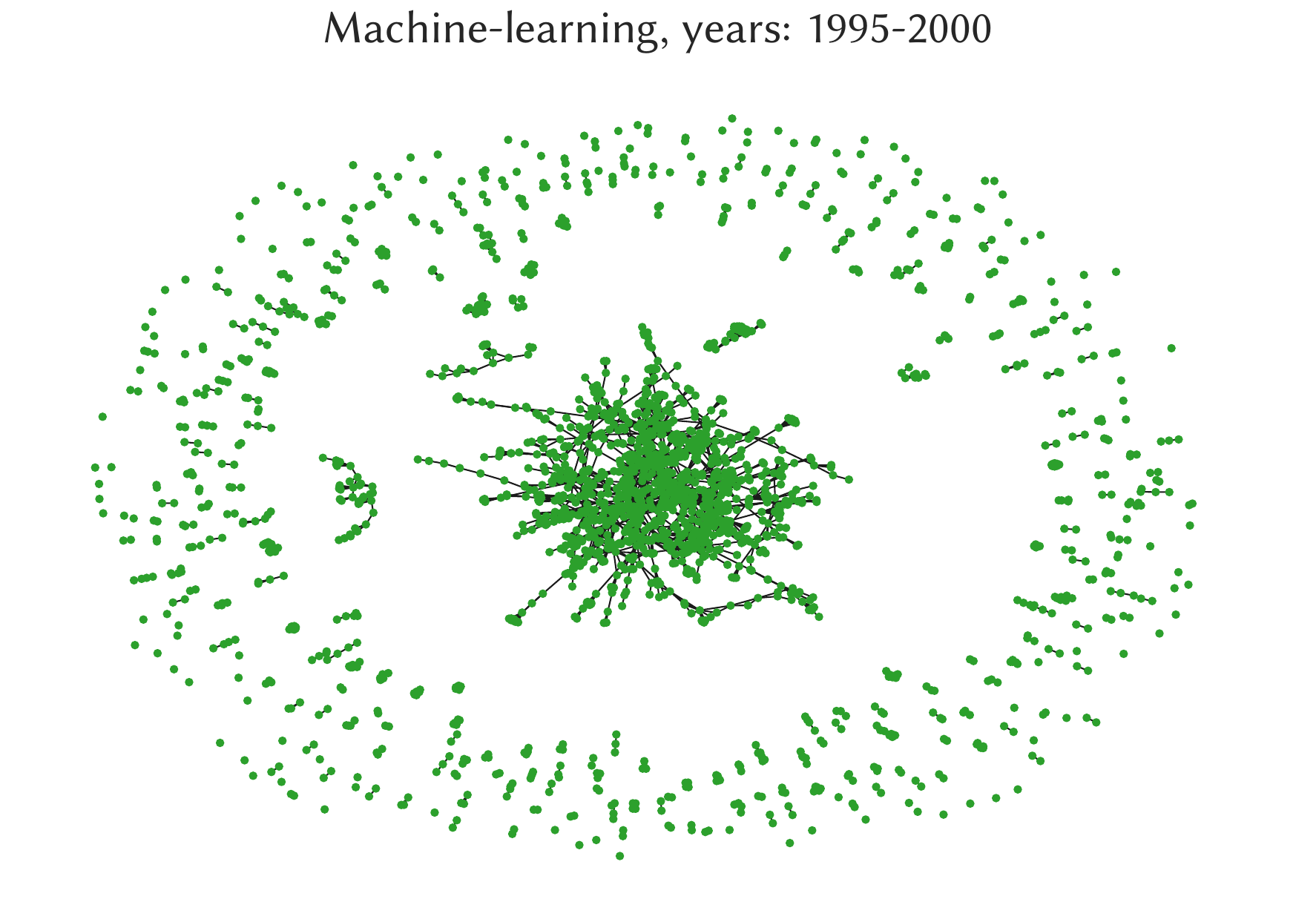}&
		\includegraphics[width=0.5\columnwidth]{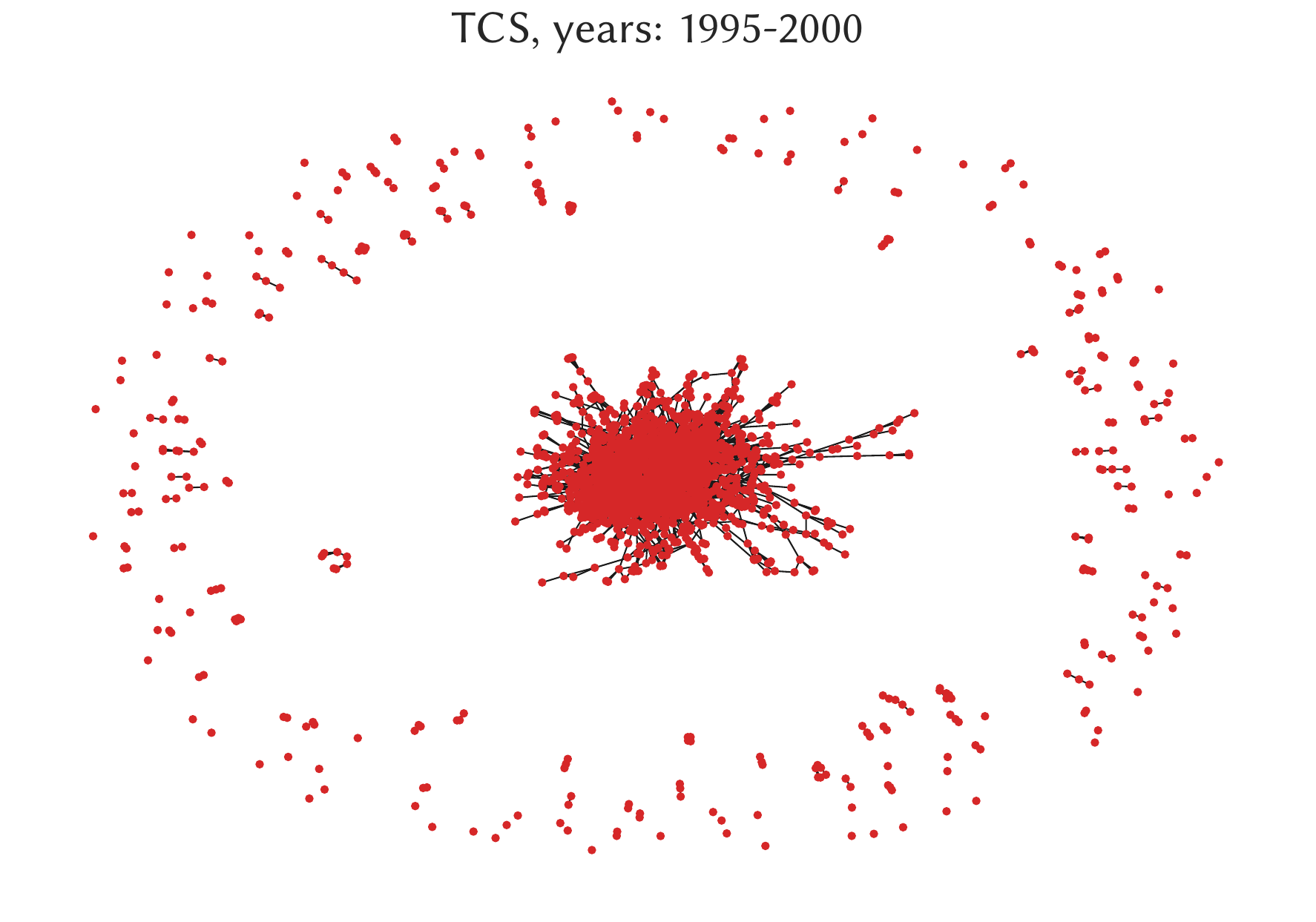}&
		\includegraphics[width=0.5\columnwidth]{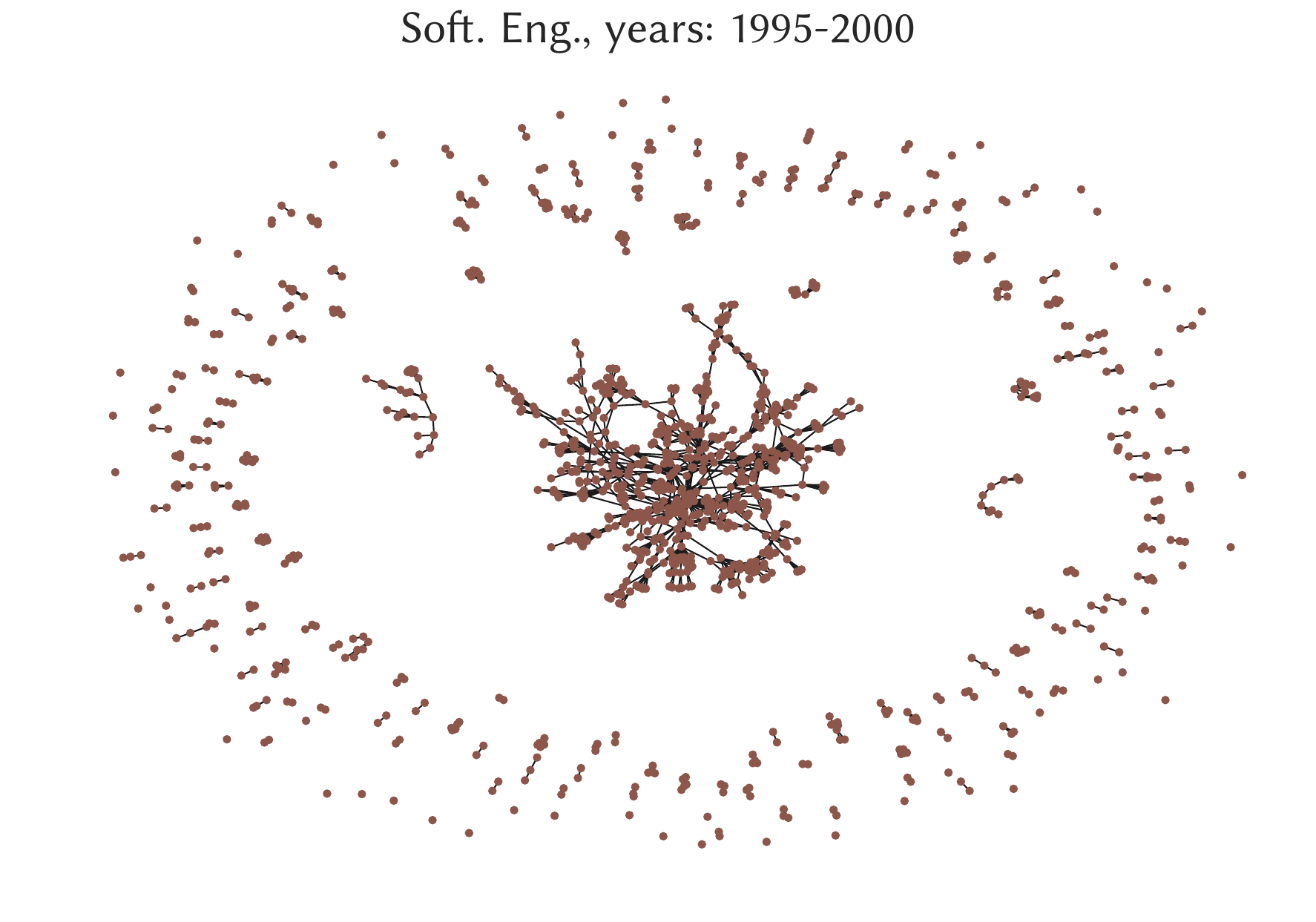}\\
		\includegraphics[width=0.5\columnwidth]{figures/case_study/clusts/4_c_4}&
		\includegraphics[width=0.5\columnwidth]{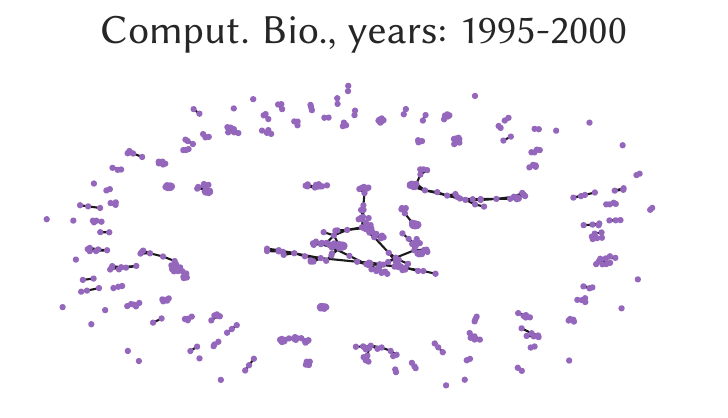}&
		\includegraphics[width=0.5\columnwidth]{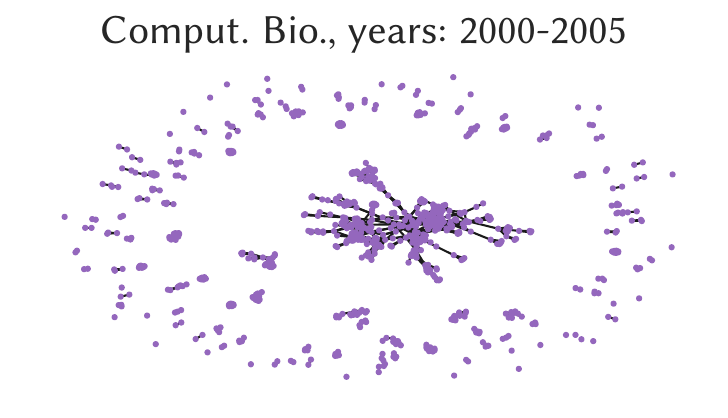}&
		\includegraphics[width=0.45\columnwidth]{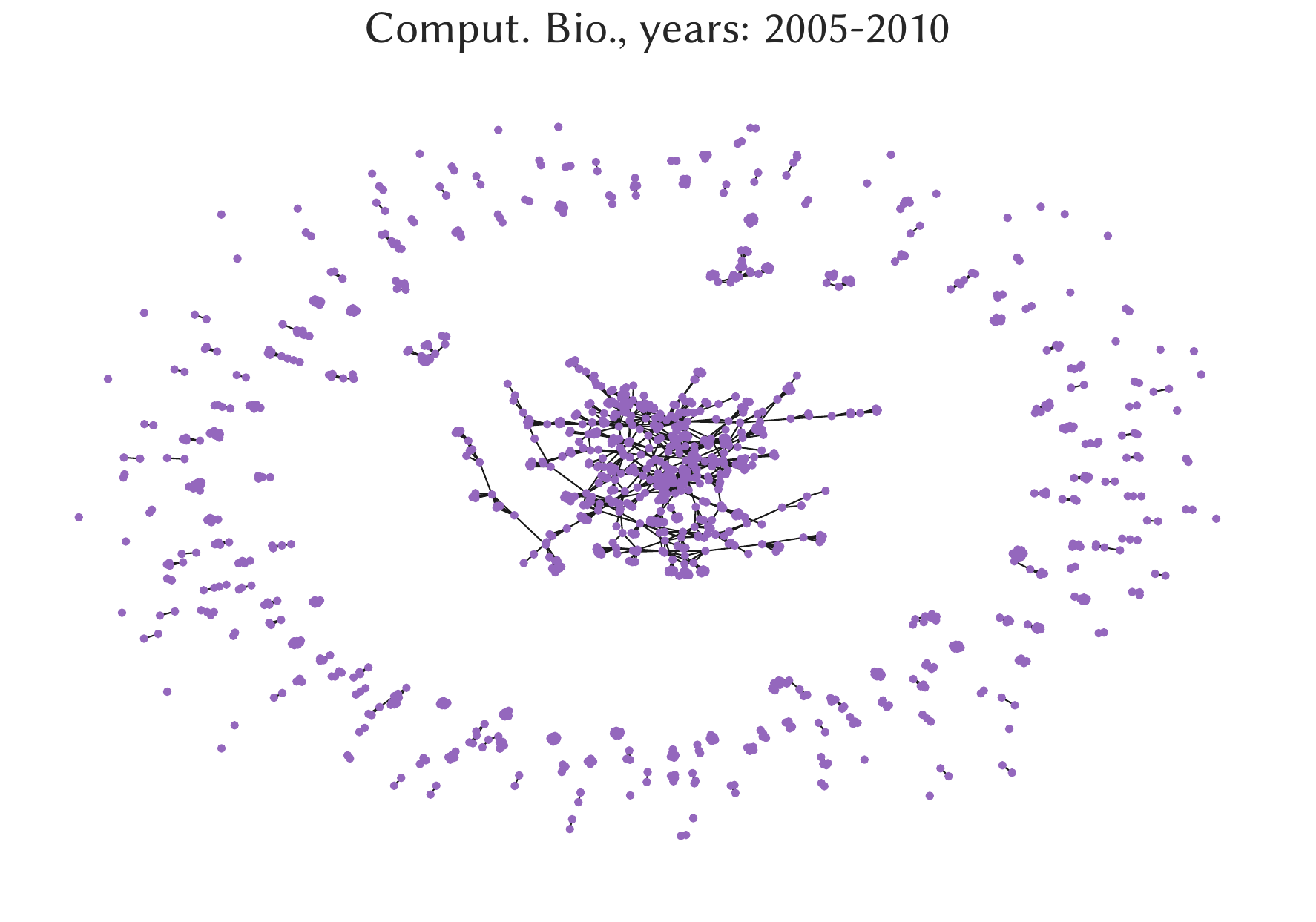}\\
	\end{tabular}
	\caption{\edit{Induced subgraphs by the various research communities over time. We select sufficiently small-sized subgraphs for ease of visualization.}}
	\label{fig:clustsMissing}
\end{figure*}

\subsection{Extending the results of Section~\ref{subsec:caseStudy}}\label{appsec:community}
\edit{\Cref{fig:clustsMissing} reports additional visualization of the structure of the induced subgraphs considered in~\Cref{subsec:caseStudy} and for which some representative examples were shown in~\Cref{fig:clusts}.}

\Cref{fig:degreeDistrib} reports the degree distribution of the various authors in each category, to account for the relative of each category we report, for each bin $[x_1,x_2)$ representing a range for the degrees of the various nodes the \emph{fraction} of nodes attaining a value of their degree in the range $[x_1,x_2)$. 
We can observe that the data-mining category and the computational biology community both have similar degree distributions (e.g., years 1990-1995 or 2000-2005). Hence their very different values of the triadic coefficients, which can be seen in~\Cref{fig:DBLPTemporal} cannot be uniquely explained by the degree distribution. In fact, we believe such values to be related to the different high-order collaboration patterns, as captured by the community structures in~\Cref{fig:clustsMissing}.

\begin{figure*}
	\addtolength{\tabcolsep}{-0.5em}
	\begin{tabular}{ll}
		\includegraphics[width=\columnwidth]{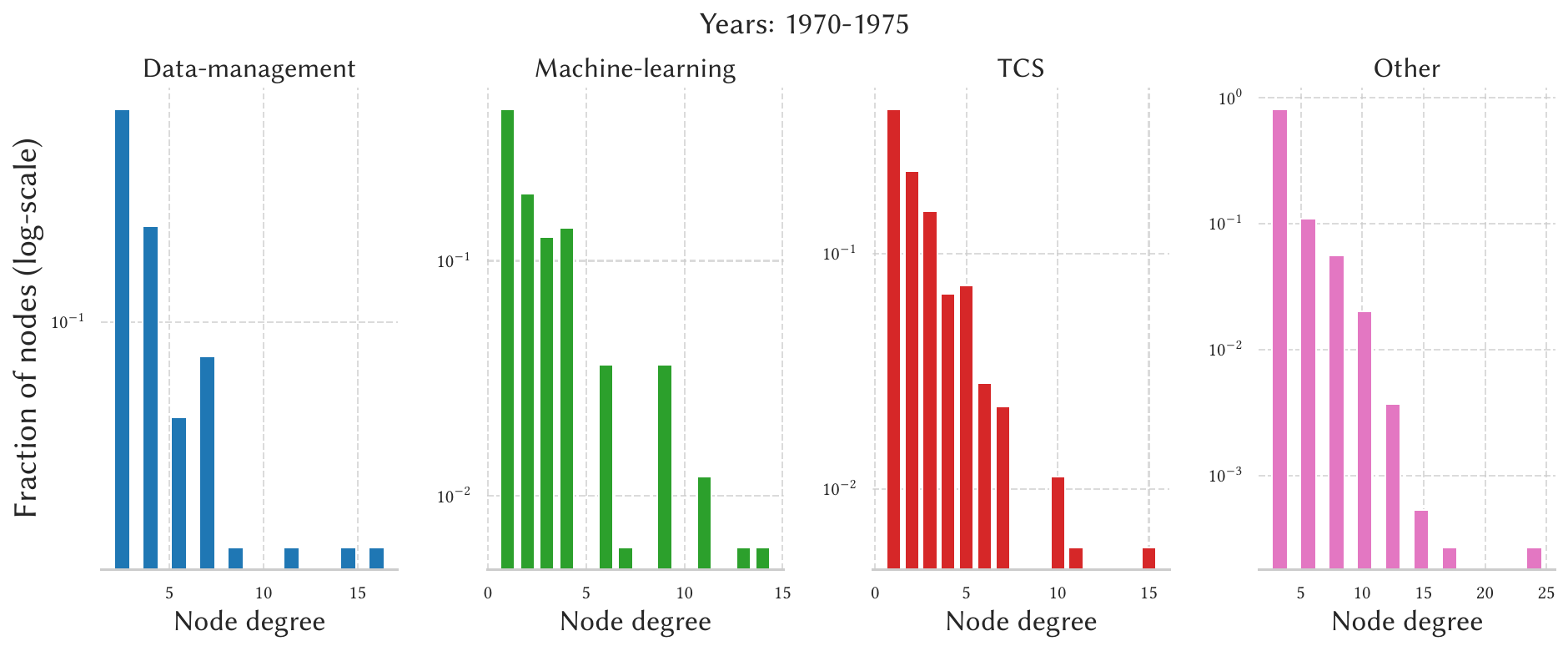}&
		\includegraphics[width=\columnwidth]{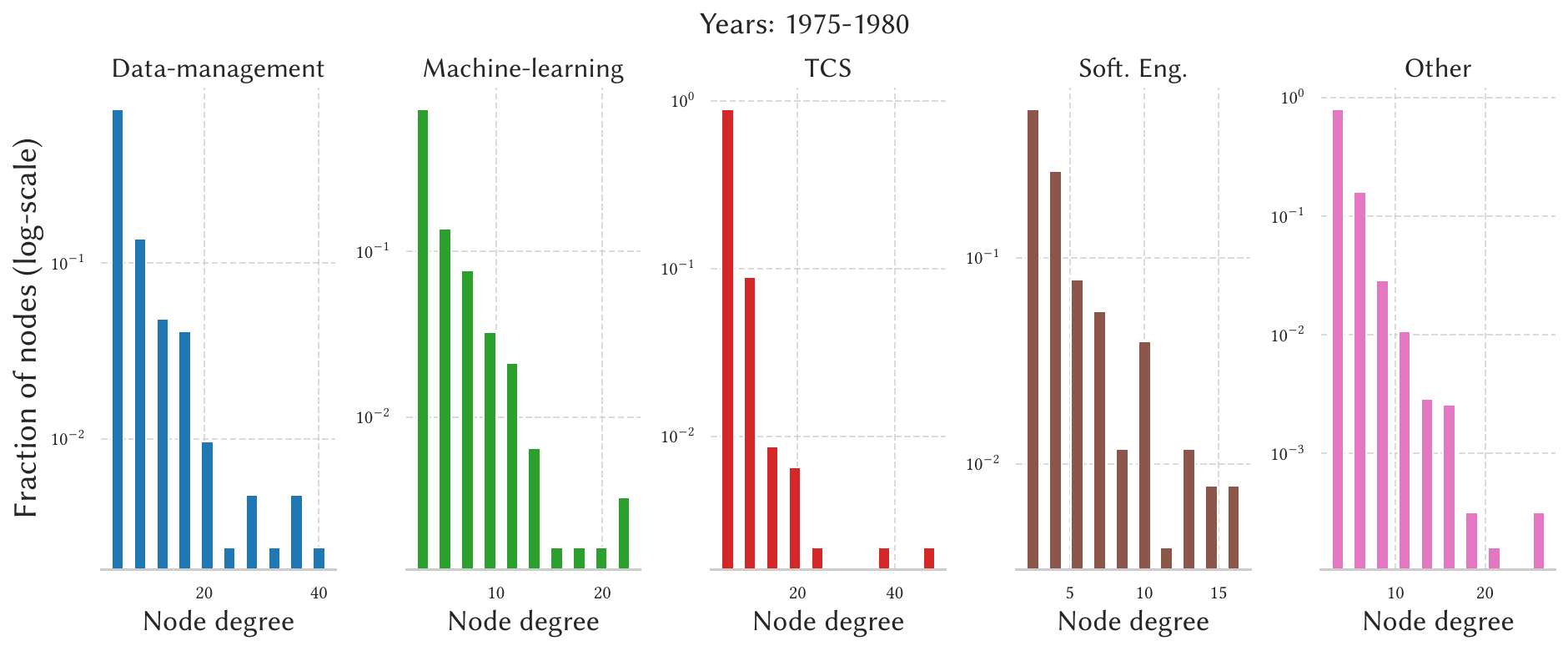}\\
		\includegraphics[width=\columnwidth]{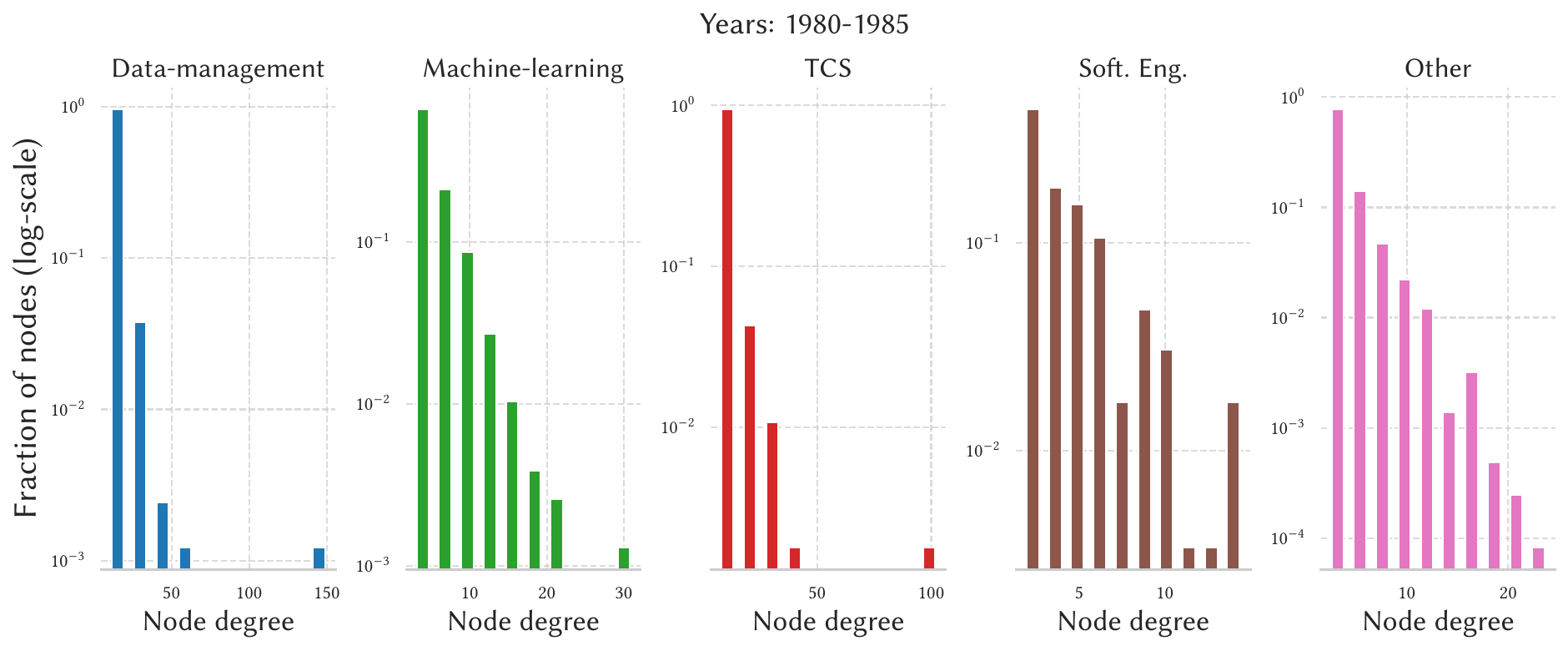}&
		\includegraphics[width=\columnwidth]{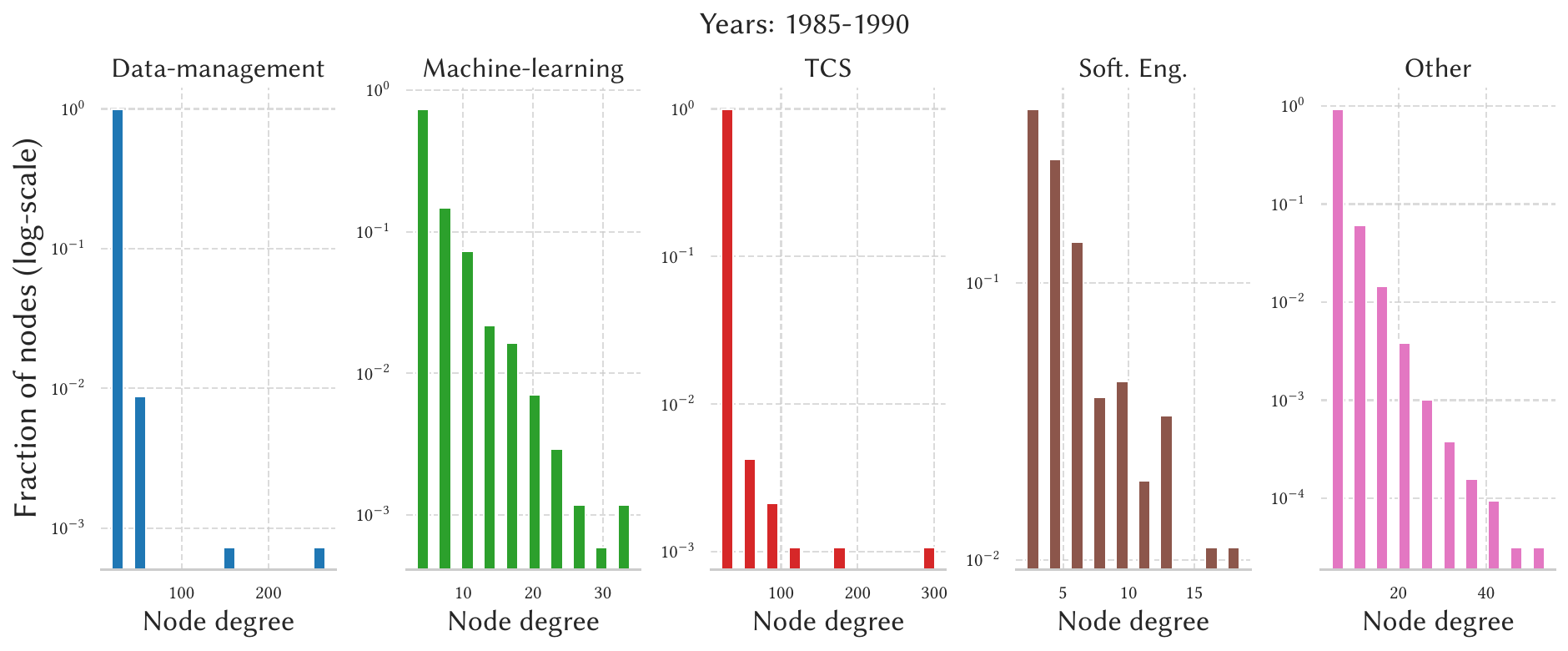}\\
		\includegraphics[width=\columnwidth]{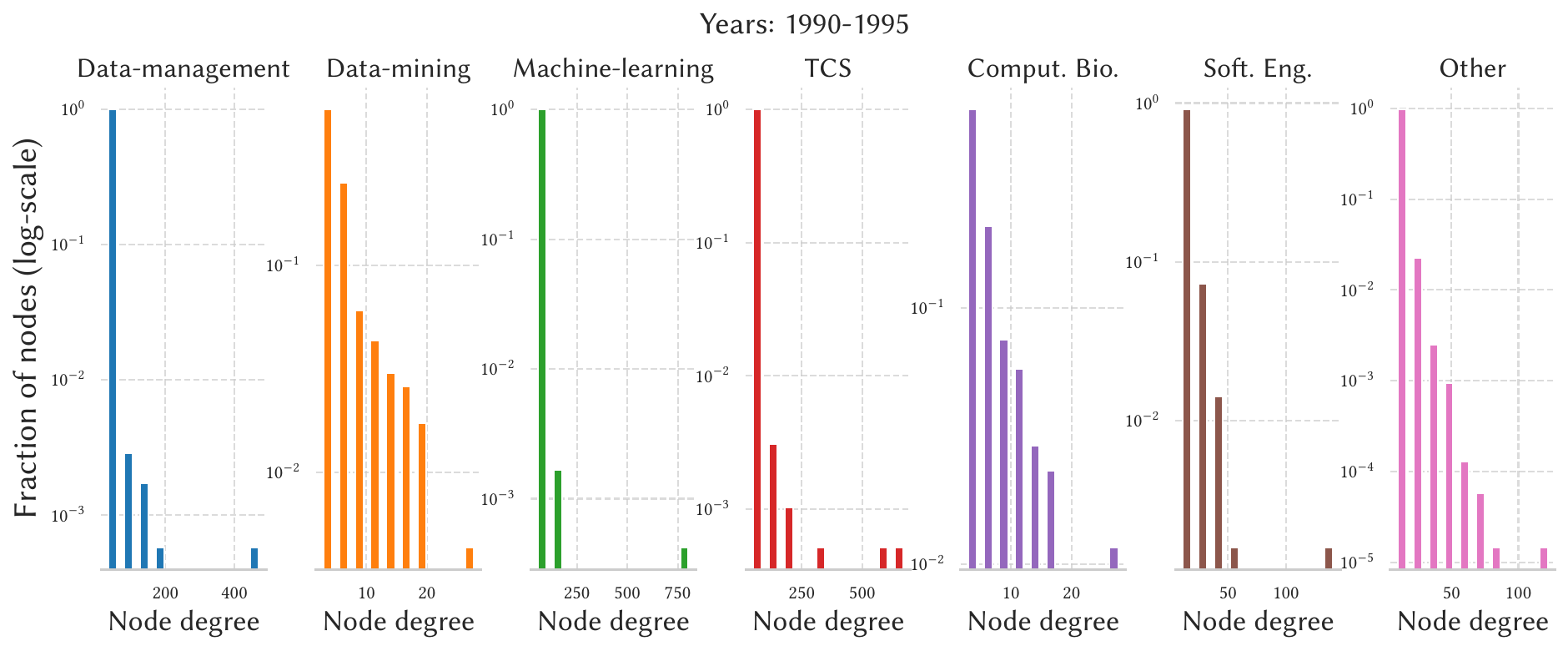}&
		\includegraphics[width=\columnwidth]{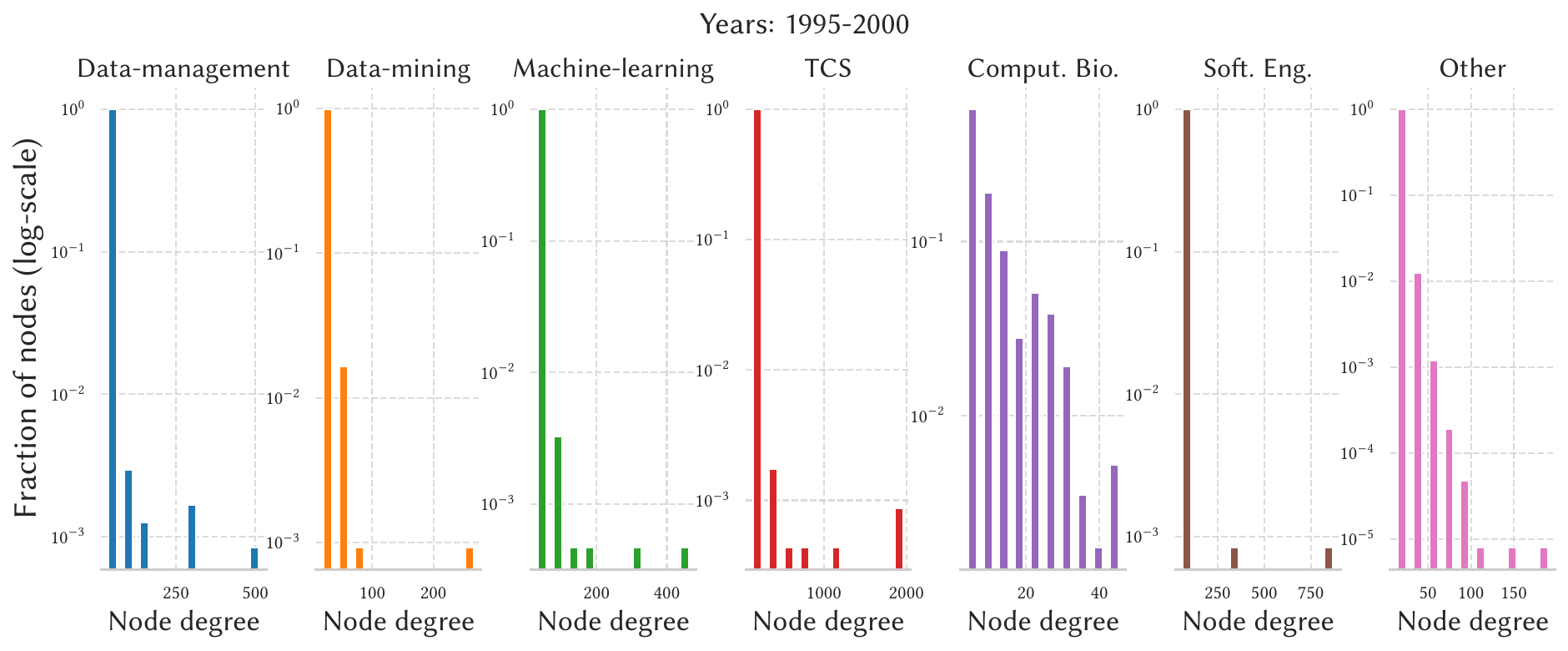}\\
		\includegraphics[width=\columnwidth]{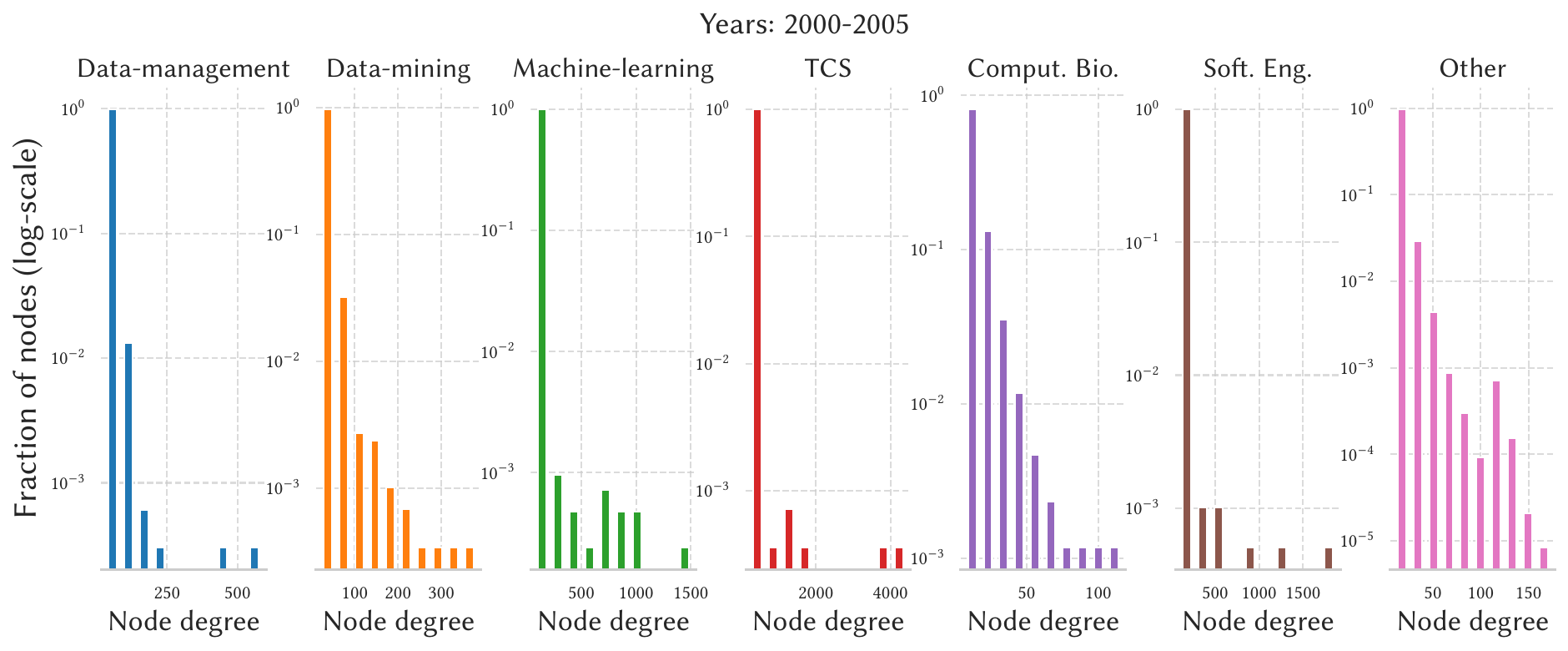}&
		\includegraphics[width=\columnwidth]{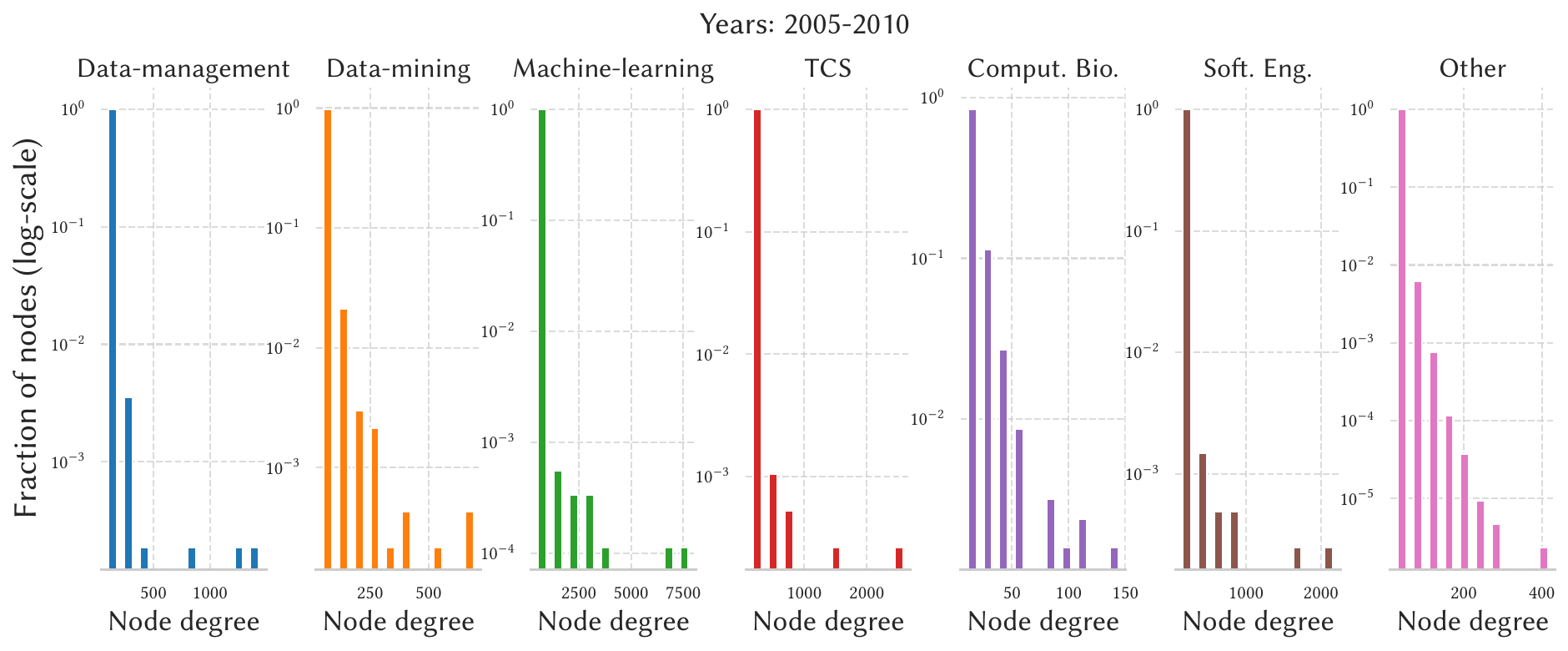}\\
		\includegraphics[width=\columnwidth]{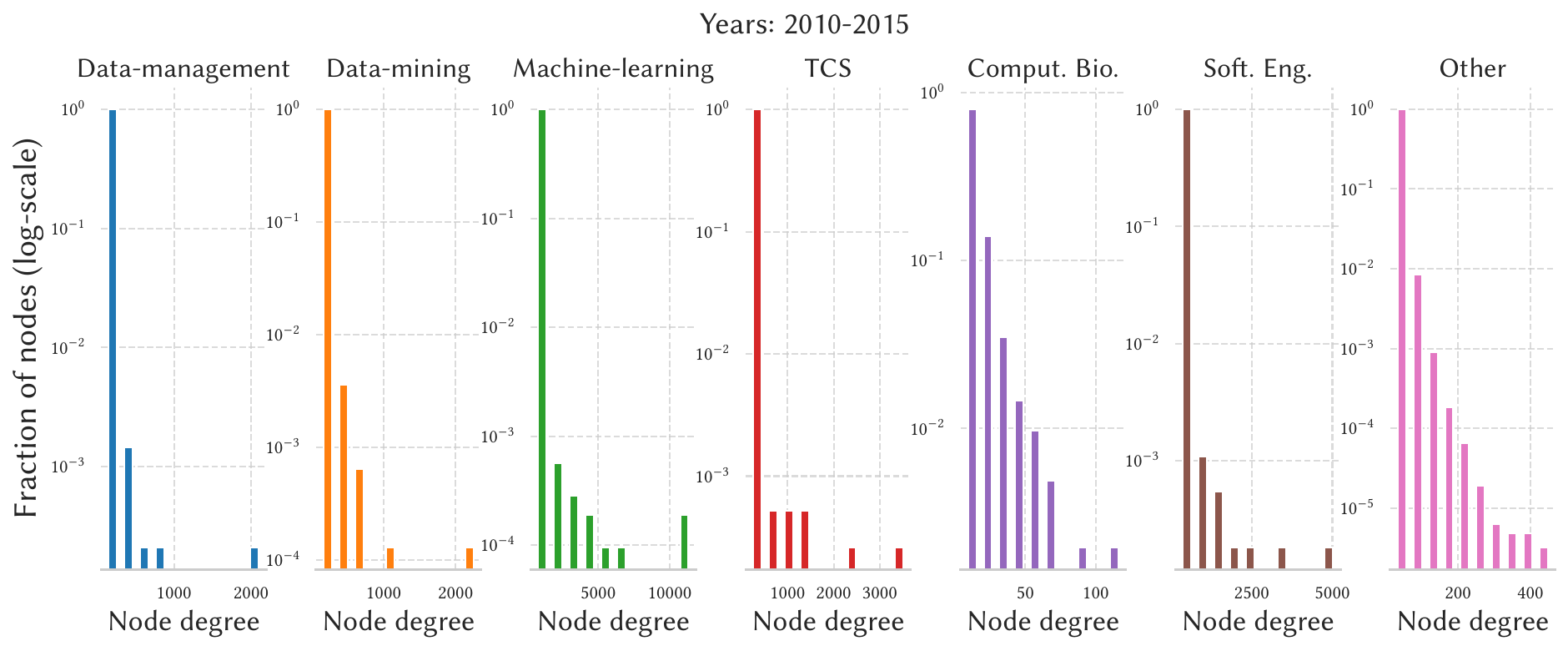}&
		\includegraphics[width=\columnwidth]{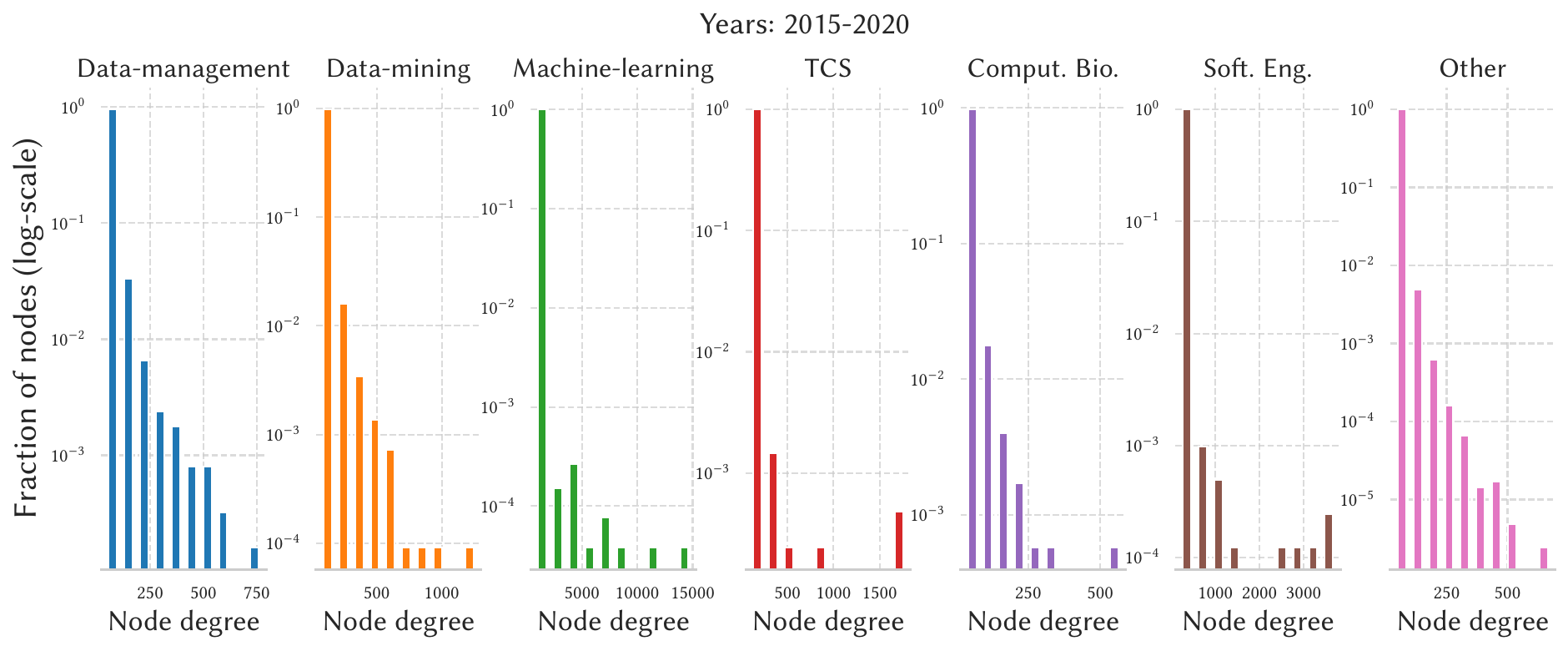}\\
	\end{tabular}
	\centering
	\includegraphics[width=1.4\columnwidth]{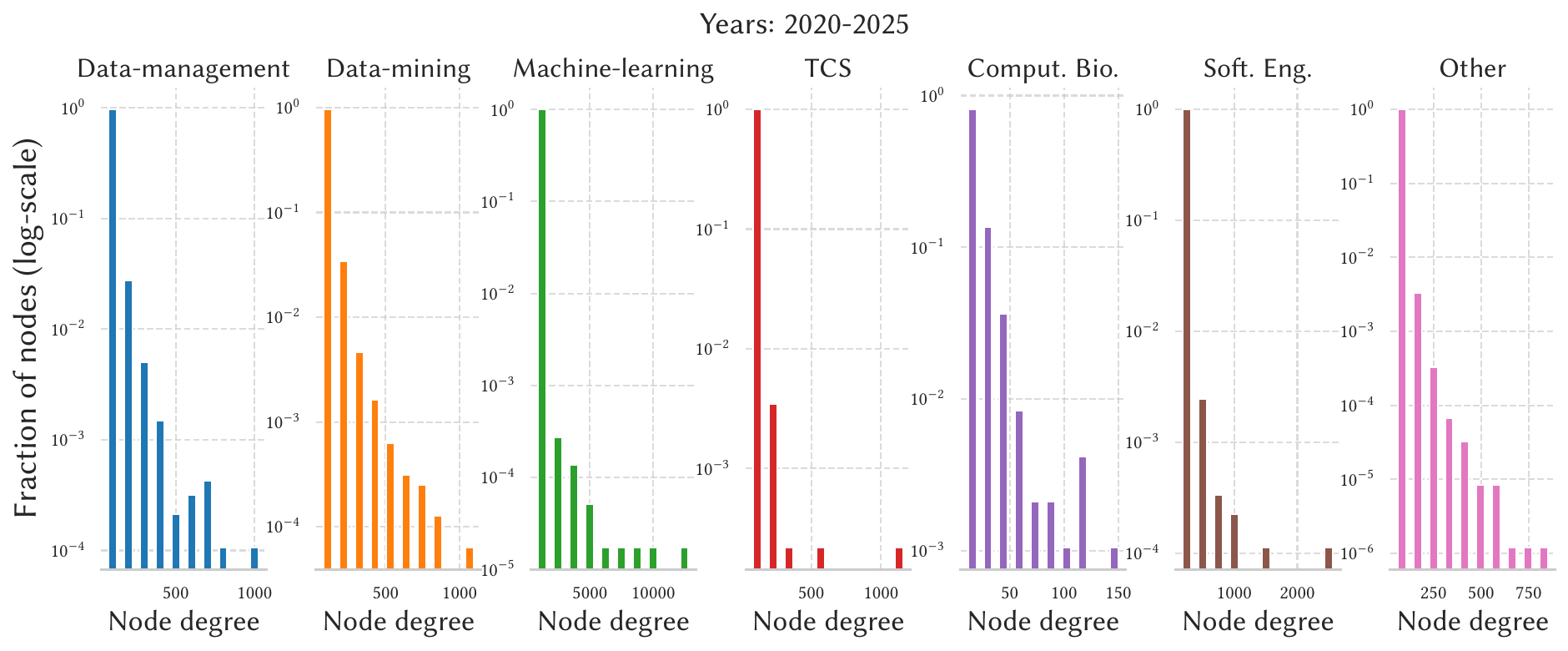}
	\caption{\edit{Binned degree distribution of the various research categories over the DBLP graph. Y-axis is in log-scale for ease of visualization.}}
	\label{fig:degreeDistrib}
\end{figure*}
\else
\fi
\end{document}